\def\DRAFT{}
\documentclass[12pt]{article}

\usepackage[T1]{fontenc}
\usepackage{lmodern, microtype}
\usepackage{multirow}
\usepackage{tabularx}
\usepackage[left=1.25in, right=1.25in, top=1.25in, bottom=1.25in]{geometry}
\usepackage[onehalfspacing]{setspace}
\usepackage[small]{titlesec}
\usepackage{mathtools,comment,paralist}
\usepackage{epigraph}

\usepackage{amssymb,amsmath,amsthm,fancyhdr,bbm}
\usepackage{booktabs}
\usepackage{graphicx}
\usepackage{subcaption}
\usepackage{xcolor}

\usepackage{pgfplots}
\pgfplotsset{compat=1.18}

\usepackage{hyperref}
\hypersetup{
  colorlinks=true,
  linkcolor=blue,
  citecolor=blue
}
\usepackage[capitalise]{cleveref}
\crefname{figure}{Figure}{Figures}

\usepackage{natbib}
\usepackage[normalem]{ulem}
\setcitestyle{authoryear,open={(},close={)}}

\makeatletter

\newtheorem{lemma}{Lemma}
\newtheorem*{lemma*}{Lemma}

\newtheorem*{axiom*}{Axiom}
\newtheorem{proposition}{Proposition}

\newtheorem{remark}{Remark}

\newtheorem*{theorem*}{Theorem}
\newtheorem{as}{Assumption}

\newtheorem{definition}{Definition}
\newtheorem*{definition*}{Definition}
\usepackage{graphicx}
\usepackage{pstricks, enumerate, pst-node, pst-text, pst-plot}

\usepackage{mleftright}
\usepackage{xparse}
\DeclareDocumentCommand\E{ m g }{\ensuremath{
    {   \IfNoValueTF {#2}
      {\mathbb{E}\mleft[{#1}\mright]}
      {\mathbb{E}\mleft[{#1}\middle\vert{#2}\mright]}%
    }
}}

\newcommand{\esssup}{\operatorname*{ess\,sup}}

\DeclareMathOperator\con{conv}

\newcommand{\R}{\mathbb{R}}

\newcommand{\mP}{\mathbb{P}}

\usepackage{accents}

\def\dd{\mathrm{d}}

\def\cA{\mathcal{A}}

\def\yist{y_i^*}

\ifdefined\DRAFT

\definecolor{ForestGreen}{rgb}{.13,.54,.13}
\definecolor{violet}{cmyk}{0.79,0.88,0,0}
\definecolor{darkmagenta}{rgb}{0.55, 0.0, 0.55}

 \newcommand{\benpo}[1]{{\color{darkmagenta}{(\textbf{Ben \& Po:} #1)}}}

\definecolor{darkBrickRed}{rgb}{.70,.13,.16}

\else
\newcommand{\fed}[1]{}

\newcommand{\omer}[1]{}
\newcommand{\benpo}[1]{}
\fi

\definecolor{redPart}{rgb}{1.0, 0.0, 0.0}
\definecolor{greenPart}{rgb}{0.0, 1, 0.0}
\definecolor{bluePart}{rgb}{0.0, 0.0, 1.0}
\definecolor{yellowPart}{rgb}{1, 1, 0}
\definecolor{magentaPart}{rgb}{1, 0, 1}
\definecolor{cyanPart}{rgb}{0.0, 1, 1}

\title{\Large Measuring Economic Preferences in the Presence of Noise: \\
The Connections Between Choices and Valuations}

\author{Ted O'Donoghue\thanks{Cornell University,  Department of Economics. Email: edo1@cornell.edu} \ \ Charles Sprenger\thanks{Caltech, Division of Humanities and Social Sciences. Email: sprenger@caltech.edu} \ \ Po Hyun Sung\thanks{Caltech, Division of Humanities and Social Sciences. Email: psung@caltech.edu} \ \
        Ben Wincelberg\thanks{Caltech, Division of Humanities and Social Sciences. Email: bwincelb@caltech.edu}}

\date{August 30, 2026}

\begin{document}
\maketitle
\begin{abstract}
Past research highlights failures of ``procedural invariance'' when measuring economic preferences using choices versus valuations. We reassess these failures by examining theoretical connections between choices and valuations when preferences are stable but measurements are noisy and individuals are heterogeneous. Even under strong assumptions governing noise and heterogeneity, stability does not generally imply identical measurements. We develop new tests of stable preferences in conjunction with various ancillary assumptions about heterogeneity and noise. We implement these tests using existing data to understand if, in the domain of risk preferences, choices and valuations truly differ and to provide quantitative assessments of any deviations. Limiting to the types of data used in the prior literature, we rarely reject the null of stable preferences. With richer data linking individual choices and valuations and structural assumptions, we find evidence of instability which differs qualitatively from the received wisdom that choices implicate greater risk aversion than valuations.
\end{abstract}



\section{Introduction}

Since the original economic experiment of \citet{thurstone1931indifference}, researchers have been interested in using behavioral measurement to assess individual preferences, evaluate adherence to underlying axioms, and develop predictions based on models fit to measurement data. To date, there exists no consensus on how such measurement should be conducted---that is, on what type of elicitation technique to use---with different researchers advocating for different methods. 

Two prominent elicitation techniques are \textit{choices} and \textit{valuations}: On one hand, we could ask people to make binary choices between options; on the other hand, we could ask people to provide valuations of options using some index of value. Ideally, these two elicitation techniques would yield measurements that are consistent with each other, as they are both meant to measure the same underlying preferences. However, the literature suggests that they do not. \citeauthor{lichtenstein1971reversals}'s \citeyearpar{lichtenstein1971reversals} famous ``preference reversals'' in risk preferences, popularized in economics by \cite{grether1979economic}, suggest apparent inconsistencies between choices versus valuations: The proportion of subjects choosing a safer over a riskier alternative is frequently found to exceed the proportion of subjects whose valuations are higher for the safer alternative, suggesting greater risk aversion in choices. More recently, an experimental-economics literature \citep[see, e.g.,][]{brown2018separated, freeman2019eliciting, freeman2019choice} documents similar apparent inconsistencies when comparing choices versus multiple price lists (MPLs), which are a type of valuation.

Researchers have interpreted these failures of ``procedural invariance'' in (at least) two ways. The preference-reversals literature typically interprets the evidence as suggesting there do not exist stable economic primitives like preferences in the first place, with preferences constructed contextually in each environment \citep{tversky1990anomalies, slovic1995construction, Lichtenstein2006construction}. The choices-versus-MPLs literature interprets the evidence as suggesting that valuations, choices, or both fail to credibly measure the preferences of interest.\footnote{For example, \citep{freeman2019eliciting} argue in favor of choices as the ``gold standard'' for measurement, and thus argue that, given the inconsistencies between measures, MPLs must be yielding biased measurement.}

In this paper, we assess these claims. We develop a model of how people respond in choice tasks and valuation tasks. Inspired by the literature above, we frame this model in terms of elicitations of risk preferences, but we highlight in our concluding section how our theoretical insights apply to other domains. We use our  model to investigate the conditions under which one would expect no apparent inconsistencies, and more generally to develop tests of the null of stable preferences in conjunction with various ancillary assumptions. We then implement these new tests using existing data to understand if, in the domain of risk preferences, choices and valuations truly differ and, if so, to provide quantitative assessments of any deviations.\footnote{A more recent literature provides cognitive accounts for violations of procedural invariance under the maintained hypothesis of stable preferences---e.g., \cite{Bouchouicha-et-al-wp-26} for choices presented independently versus in list form, and \cite{Shubatt-Yang-wp-26} for different types of valuation tasks. Within these frameworks, noisy coding of value (in the spirit of e.g., \cite{Khaw-Li-Woodford-21}) interacts with the differential demands of different procedures to generate apparent inconsistencies. \cite{Bouchouicha-et-al-wp-26} and \cite{Shubatt-Yang-wp-26} both provide experimental evidence of apparent inconsistencies that are in line with their proposed cognitive accounts. The fundamental distinction with our own work is that we develop formal tests of stable preferences across unbiased measurements while assuming symmetric noise, while they develop potential theories for how cognitive noise can generate differentially biased measurements and thus lead to distinct measures under stable preferences.}

Our exercise begins with three basic observations. First, economic measurement of preferences---like all scientific measurement---is inherently noisy.\footnote{This observation is most often attributed to \cite{mosteller1951experimental}, though even \citet{thurstone1931indifference} carries some recognition of noise in measurement. In the domain of economics experiments, it is natural to interpret noise in measurement as decision noise, where people's decisions reflect a noisy instantiation of their underlying preferences; we often use the terms measurement noise and decision noise interchangeably.} Second, there is no a priori reason to expect that different elicitation techniques are subject to identical noise, and so we do not impose this assumption (and indeed our empirical analysis seems to contradict it). Third, most experiments study aggregate data across individuals, and thus we must account for the fact that individuals are fundamentally heterogeneous in their underlying preferences and possibly also in their decision noise. 

With these preliminaries in hand, we provide three groups of theoretical results. First, we derive predictions for the empirical objects used in the prior literature: When subjects are asked to compare a safer option to a riskier option, we can study the proportion of subjects who choose the safer option (which we denote by $\widehat{\mP_c}(\text{safe})$), or the proportion of subjects whose valuations indicate a preference for the safer option (which we denote by $\widehat{\mP_v}(\text{safe})$). The prior literature tests the null hypothesis of stable preferences by examining whether $\widehat{\mP_c}(\text{safe})=\widehat{\mP_v}(\text{safe})$, and the common finding is $\widehat{\mP_c}(\text{safe}) > \widehat{\mP_v}(\text{safe})$. We highlight that, even under the stable-preference null, our model need not predict that $\mP_c(\text{safe})=\mP_v(\text{safe})$.\footnote{Throughout, we use $\widehat{\mP_c}(\text{safe})$ and $\widehat{\mP_v}(\text{safe})$ to denote empirical objects, and $\mP_c(\text{safe})$ and $\mP_v(\text{safe})$ to denote theoretical predictions for those objects.} 

Hence, we derive the attainable set for $(\mP_c(\text{safe}),\mP_v(\text{safe}))$---that is, combinations consistent with our model---under the assumption of stable preferences combined with different sets of assumptions about heterogeneity and noise. Under our most basic set of assumptions about noise at the level of individual subjects, the decision probabilities for each subject must satisfy the simple consistency condition of being on the same side of $\frac{1}{2}$ (i.e., $\mP_{ci}(\text{safe})$ and $\mP_{vi}(\text{safe})$ both greater than $\frac{1}{2}$ or both less than $\frac{1}{2}$). Aggregating across subjects with our weakest assumptions on heterogeneity in preferences and noise---which permit arbitrary combinations of individuals each satisfying the consistency condition above---generates an attainable set that encompasses much of the $(\mP_c(\text{safe}),\mP_v(\text{safe}))$ unit square.  Stronger assumptions shrink the attainable set, but even under our strongest set of assumptions about heterogeneity and noise, the attainable set still encompasses half the unit square. These theoretical results motivate new tests of the null of stable preferences in which one compares an observed combination $(\widehat{\mP_c}(\text{safe}),\widehat{\mP_v}(\text{safe}))$ to null attainable regions instead of the traditional test of identical measurements.

Our second group of theoretical results demonstrates how richer data can permit additional tests. If, in addition to $(\widehat{\mP_c}(\text{safe}),\widehat{\mP_v}(\text{safe}))$, the data also provide information on the mean preference within the population, we can develop predictions for $\mP_c(\text{safe})$ conditional on the mean preference. Similarly, if we have access to within-subject data on both choices and valuations, we can develop subject-level predictions for $\mP_c(\text{safe})$ conditional on both the mean preference and their own stated valuation.

Our third group of theoretical results highlights what one can do if one has within-subject data on both choices and valuations and is willing to take a more structural approach by imposing functional assumptions on preference heterogeneity and noise. We focus on one example in which we assume that preference heterogeneity and noise are jointly normal and independent. We derive how this structure naturally leads to probit and maximum likelihood estimators that can be used both for testing the null of stable preferences and, when stability is rejected, for estimating the magnitude of deviations from stability.

Motivated by our theoretical development, we conduct empirical analyses that implement our new tests. We use two datasets constructed from existing data. First, we assemble a dataset of prior experiments from the preference-reversal and choices-versus-MPLs literatures (323 experiments in total), where we observe $\widehat{\mP_c}(\text{safe})$ and $\widehat{\mP_v}(\text{safe})$ for each experiment. Second, we pull data from two recent papers \citep{MNOSS-2024-distinguishing, MNOSS-2026-connecting} that collect choices and valuations for the same subjects (424 experiments in total).\footnote{\cite{MNOSS-2024-distinguishing, MNOSS-2026-connecting} collect this data for different purposes from ours, and neither directly compares choices and valuations in the way that we do here.} 

We document that, when approached from the perspective of the prior literature, both datasets are consistent with the received wisdom that people are more likely to choose a safer option than valuations would imply: across these datasets, approximately 70-80\% of experiments indicate $\widehat{\mP_c}(\text{safe}) > \widehat{\mP_v}(\text{safe})$. However, when comparing $(\widehat{\mP_c}(\text{safe}),\widehat{\mP_v}(\text{safe}))$ to the attainable sets that we derive under our various sets of assumptions, we cannot reject the null of stable preferences for the majority of experiments. For instance, even under our strongest set of assumptions, only 45\% of the prior literature's experiments and 23\% of the \citet{MNOSS-2024-distinguishing, MNOSS-2026-connecting} experiments lie outside of the attainable region for stable preferences. Moreover, even fewer (20\% and 12\%) can be statistically differentiated from the attainable regions---that is, across the two datasets, 80-90\% of experiments are consistent with stable preferences.

Because the \citet{MNOSS-2024-distinguishing, MNOSS-2026-connecting} dataset contains within-subject data on both choices and valuations, we can use it to assess our conditional predictions. Here, we see more mixed evidence. When conditioning on an estimate for the mean preference within the population, $\widehat{\mP_c}(\text{safe})$ looks largely consistent with stable preferences. In contrast, when conditioning on both the mean preference and a subject's own stated valuation, $\widehat{\mP_c}(\text{safe})$ often looks inconsistent with stable preferences.

Finally, we can also use the \citet{MNOSS-2024-distinguishing, MNOSS-2026-connecting} dataset to pursue our more structural approach. The data contain 84 distinct groupings on which we can conduct our analysis (see Section \ref{sec:empirical} for details). Under the maintained assumption of joint normality, we are remarkably well-powered, and we reject the null of stable preferences in roughly 80\% of the 84 groupings. That said, when we study the nature of these deviations, the message differs from the received wisdom that choices induce greater risk aversion than valuations. In fact, overall, there appears to be no systematic direction to violations: The estimates indicate that choices induce greater risk aversion in roughly half of groupings, whereas valuations induce greater risk aversion in the other half. Perhaps more interesting, different types of experiments reveal different directional deviations. For decisions between a certain amount and a binary lottery, the estimates indicate that choices tend to induce more risk aversion; for decisions between a certain amount and a trinary lottery, the estimates indicate that valuations tend to induce more risk aversion; and for decisions between two binary lotteries, the estimates indicate no systematic deviation. 

Our analysis has three important implications for the central question of whether there is preference stability across measurement techniques. First, and most importantly, prima facie inconsistencies across measurement techniques do not necessarily imply a failure of stable preferences. Our analysis highlights how, to formally test for preference stability, the analyst must specify assumptions about heterogeneity and noise so as to identify the appropriate prediction to test. Our analysis further highlights that, for any test, rejection could mean a rejection of stable preferences but also could be a rejection of those ancillary assumptions.

Second, we provide guidance for experiments designed to assess preference stability. In particular, our analysis highlights the value of collecting within-subject data on both choices and valuations so as to conduct more powerful tests. One interpretation of our analysis is that tests based on whether observed combinations of $(\widehat{\mP_c}(\text{safe}),\widehat{\mP_v}(\text{safe}))$ fall within some acceptable range have low power to detect deviations from stability relative to our tests based on richer data. Indeed, under our strongest assumptions of full joint normality, we infrequently reject the null of stable preferences when using tests based on observed combinations of $(\widehat{\mP_c}(\text{safe}),\widehat{\mP_v}(\text{safe}))$, whereas with the same data we frequently reject the null of stable preferences when we use the connections in the data to study richer empirical objects. 

Third, our empirical analysis that implements our tests suggest a very different message from the received wisdom that choices induce greater risk aversion than valuations. Our results add nuance to a literature that often draws broad conclusions about the robustness and directionality of preference reversals. Further work with purpose-built experiments would be valuable for assessing whether our nuanced findings are robust and for investigating their source.

This manuscript proceeds as follows: Section \ref{sec:priorwork} reviews prior work evaluating procedural invariance and testing for preference stability;  Section \ref{sec:theory} presents our theoretical framework for evaluating the relationship between choices and valuations under various assumptions about heterogeneity and noise and develops our new tests; Section \ref{sec:empirical} reports on our empirical implementation of our new tests; and Section \ref{sec:conclusion} provides a discussion drawing out implications for the empirical literature on procedural invariance, and concludes.

\section{Prior Evidence on Choices versus Valuations}\label{sec:priorwork}

While our theoretical insights can be applied more broadly (as we discuss in Section \ref{sec:conclusion}), we focus on apparent inconsistencies between two elicitation procedures: choices versus valuations. In this section, we clarify exactly what we mean by inconsistencies between these two procedures, and then provide an overview of prior evidence.

\subsection{Choices and Valuations}

Consider two types of tasks:

\begin{itemize}
    \item \textbf{Choice task}: We ask a person to make a binary choice between two options.
    \item \textbf{Valuation task}: We ask a person to state a scalar value $y_i$ that makes two options of equal value.
\end{itemize}

While choice tasks are straightforward, valuation tasks require additional description. In the domain of choice between lotteries, there are a variety of scalar values that we could elicit. Here are four examples:\footnote{To simplify notation, whenever we use lottery notation, we suppress zero outcomes. Hence, $(30,0.6)$ represents a lottery to yields 30 with probability 0.6 and 0 with probability 0.4.}
\begin{itemize}
     \item[] Valuation Task A: State a $y_i$ such that $(y_i,1)$ is of equal value to $(30,0.6)$.
     \item[] Valuation Task B: State a $y_i$ such that $(20,1)$ is of equal value to $(y_i,0.2;20,0.3)$.
     \item[] Valuation Task C: State a $y_i$ such that $(y_i,0.4)$ is of equal value to $(30,0.2)$.
     \item[] Valuation Task D: State a $y_i$ such that $(20,y_i)$ is of equal value to $(30,0.6)$.
\end{itemize}
\noindent The simplest form for a valuation task is when we ask a person to state a certainty equivalent for a lottery (as in Valuation Task A). But valuation tasks can come in many other forms depending on which value is elicited. For instance, even when comparing a certain amount to a lottery, rather than elicit the sure amount, we could elicit one of the amounts in the lottery (as in Valuation Task B). Alternatively, we might compare two lotteries, in which case we could elicit one of the amounts in one of the lotteries (as in Valuation Task C) or one of the probabilities in one of the lotteries (as in Valuation Task D). In our language, all of these fall in the broad category of valuation tasks.

For any specific valuation task, there are multiple ways to elicit a scalar indifference value. For instance, we might simply ask a person to state a value (i.e., fill in the blank) or make a bid in the Becker-DeGroot-Marschak random second-price auction mechanism. Alternatively, we might use a multiple price list (MPL)---that is, a sequence of binary choices presented in list form where only the value of interest changes across rows.

\subsection{Apparent Inconsistencies Between Choices and Valuations}

Suppose we want to measure a person's preference between two lotteries. One way to measure that preference is to use a choice task that  asks the person to select between the two lotteries. A second way to measure that preference is through a valuation task (or tasks) such that the stated value (or values) convey information on the preference between the two lotteries.\footnote{For instance, suppose the two lotteries of interest are $(14,1)$ versus $(30,\frac{1}{2})$. If we use a valuation task in which we ask the person to state a $y_i$ such that $(y_i,1)$ is of equal value to $(30,\frac{1}{2})$, then a stated $y_i<14$ suggests a preference for $(14,1)$ over $(30,\frac{1}{2})$, whereas a stated $y_i>14$ suggests a preference for $(30,\frac{1}{2})$ over $(14,1)$.}

Our interest is studying apparent inconsistencies between these two approaches. Hence, for our purposes, an \textit{experiment} fixes a pair of lotteries, and then one treatment asks subjects to make a binary choice between the two lotteries, while a second treatment (within- or between-subjects) elicits a valuation (or valuations) that suggests a preference between the two lotteries. In all experiments we consider, one lottery can clearly be labeled the safer lottery, and the empirical quantities (at least initially) are the proportion of subjects who choose the safer option in the choice task, which we denote by $\widehat{\mP_c}(\text{safe})$, and the proportion of subjects whose valuation(s) suggests they prefer the safer option, which we denote by $\widehat{\mP_v}(\text{safe})$. In line with the prior literature, we say there are apparent inconsistencies when these empirical quantities differ.

\subsection{Prior Evidence}

The earliest comparison of choices versus valuations (at least in the domain of risk) comes from the literature on ``preference reversals'' initially developed by \citet{lichtenstein1971reversals} and later popularized in economics by \citet{grether1979economic}. In a typical experiment, subjects are asked to compare two lotteries, one with a high chance of a moderate prize (called a $P$-bet) and one with a low chance of a high prize (called a \$-bet); note that the $P$-bet is the safer option. First, subjects are presented with both bets and asked to make a choice. Second, subjects are asked to state a reservation value (i.e., a certainty equivalent) for each lottery independently. The researchers then compare (i) the proportion of subjects who choose the safer option but state a higher valuation for the riskier option to (ii) the proportion of subjects who choose the riskier option but state a higher valuation for the safer option. The typical finding is that the former substantially exceeds the latter. This finding is formally equivalent to $\widehat{\mP_c}(\text{safe}) > \widehat{\mP_v}(\text{safe})$, that is, that people appear more risk averse when making choices than they do when providing valuations.\footnote{To see this equivalence, let $(a,b)$ denote the pattern choosing option $a \in \{\text{safe},\text{risky}\}$ and stating a higher valuation for option $b \in \{\text{safe},\text{risky}\}$. Then $\widehat{\mP_c}(\text{safe}) = \widehat{\mP}(\text{safe},\text{safe}) + \widehat{\mP}(\text{safe},\text{risky})$ while $\widehat{\mP_v}(\text{safe}) = \widehat{\mP}(\text{safe},\text{safe}) + \widehat{\mP}(\text{risky},\text{safe})$, and thus $\widehat{\mP_c}(\text{safe}) > \widehat{\mP_v}(\text{safe})$ if and only if $\widehat{\mP}(\text{safe},\text{risky}) > \widehat{\mP}(\text{risky},\text{safe})$.}

More recently, a small literature has investigated whether people behave consistently when making a binary choice versus when that binary choice is part of an MPL \citep{brown2018separated, freeman2019eliciting, freeman2019choice}. This literature uses a between-subjects design with two treatments. In the choice treatment, subjects are asked to make a choice between two lotteries. In the MPL treatment, subjects are asked to complete an MPL in which one of the rows is the binary choice from the choice treatment. The researchers then compare the proportion of subjects who choose the safer option in the choice treatment versus the proportion of subjects who choose the safer option in the focal row in the MPL treatment. Again, the common empirical finding is that $\widehat{\mP_c}(\text{safe}) > \widehat{\mP_v}(\text{safe})$, that is, that people appear more risk averse when making choices than they do when providing valuations.

Most recently, \cite{MNOSS-2024-distinguishing,MNOSS-2026-connecting} conduct experiments in which they collect choices and valuations from the same subjects. Their goal was not to compare choices versus valuations; rather, their focus was comparing pairs of valuations to study the phenomena of the common ratio effect, the common consequence effect, and preferences for probabilistic mixtures. They collected data on corresponding pairs of choices to illustrate an inference problem that arises when studying paired-choice tasks (in \cite{MNOSS-2024-distinguishing}) and to show that paired choices and paired valuations can yield similar conclusions in some instances (in \cite{MNOSS-2026-connecting}). However, their data provide a convenience dataset that we can use for the purposes of this paper. When we reanalyze their data to directly compare individual choices versus valuations, again it turns out that the vast majority (80\%) of comparisons exhibit $\widehat{\mP_c}(\text{safe})>\widehat{\mP_v}(\text{safe})$.

In Section \ref{sec:empirical}, we provide more details on the literatures above and their data, and we then use that data for our empirical application. But first we turn to our main research question: What exactly can we infer when we observe apparent inconsistencies of the form described above?

\section{Theoretical Analysis of Choices versus Valuations}\label{sec:theory}

Apparent inconsistencies between choices and valuations have been interpreted in two primary ways. On one hand, the preference-reversals literature uses apparent inconsistencies to suggest that there do not exist stable underlying preferences. On the other hand, the choices-versus-MPLs literature uses apparent inconsistencies to suggest that one or both elicitation procedures are generating biased measures of underlying preferences. Our goal is to assess these claims by developing a theoretical model that makes predictions for observed behavior under each elicitation procedure.

To simplify our exposition, our analysis in the text focuses on the simple case where we want person $i$ to compare a lottery $X$ to a sure payment of amount $y$. However, the main results and intuitions generalize to the case where a person makes lottery-versus-lottery comparisons, and they also generalize to the  paradigm from the preference-reversals literature where a person chooses between two lotteries $X$ and $Z$ and also provides certainty equivalents for both (see \cref{ap:general_model}).

There are two ways in which we can ask a person to make a comparison between a lottery $X$ and a sure payment of amount $y$:

\begin{itemize}
    \item \textbf{Choice task}: We ask person $i$ to make
     a binary choice between $X$ and $y$.
     \item \textbf{Valuation task}: We ask person $i$ to state the value $y_i$ that is of equal value to them as the lottery $X$. 
\end{itemize}

As described in Section \ref{sec:priorwork}, the prior literature compares two quantities: $\widehat{\mP_c}(y)$, the empirical choice frequency of $y$ over $X$, and $\widehat{\mP_v}(y)$, the empirical proportion of individuals who exhibit a valuation for $X$ that is lower than the given $y$.\footnote{For a comparison between a lottery $X$ and a sure amount $y$, the latter is clearly the safer option, and thus $\widehat{\mP_c}(y)$ and $\widehat{\mP_v}(y)$ correspond to $\widehat{\mP_c}(\text{safe})$ and $\widehat{\mP_v}(\text{safe})$ from Section \ref{sec:priorwork}.} In experimental data, $\widehat{\mP_c}(y)$ and $\widehat{\mP_v}(y)$ depend on a combination of underlying population values $\mP_c(y)$ and $\mP_v(y)$ and sampling variation due to having finite experimental samples. We address the latter in Section \ref{sec:empirical}; in this section, we develop a theoretical model of the underlying population values $\mP_c(y)$ and $\mP_v(y)$.

The underlying population values $\mP_c(y)$ and $\mP_v(y)$ will reflect a combination of noise at the individual level and heterogeneity in preferences at the aggregate level. Hence, in building a stochastic model to make predictions for these quantities, we start with a model of stochastic behavior at the individual level (in \cref{sub:ind}), and we then incorporate heterogeneity in preferences to generate stochastic behavior at the aggregate level (in \cref{sub:agg}).

\subsection{Individual Behavior}\label{sub:ind}
We begin with a model of individual-level behavior, and in particular how underlying preferences combine with noise to generate stochastic behavior at the individual level. Throughout, we assume:
\begin{itemize}
    \item  For a choice task between $X$ and $y$, individual $i$'s underlying preferences involve an indifference value $y^*_{ci}$ such that $y^*_{ci} \sim_{ci} X$.\footnote{We use $\sim_{ci}$, $\prec_{ci}$, and $\succ_{ci}$ to denote underlying preferences for choice tasks, and $\sim_{vi}$, $\prec_{vi}$, and $\succ_{vi}$ to denote underlying preferences for valuation tasks. Under the null of stable preferences across the two tasks, we instead use $\sim_{i}$, $\prec_{i}$, and $\succ_{i}$ for both.} For their actual choice, individual $i$ chooses $y$ over $X$ when $y$ is larger than $y_{ci} \equiv y_{ci}^* + \varepsilon_{c}$, where choice noise $\varepsilon_{c}$ has CDF $F_{c}$. Individual $i$ thus chooses $y$ with probability $\mP_{ci}(y) = \mP(y_{ci}^*+\varepsilon_{c} \leq y) = F_{c}(y-y_{ci}^*)$.
    \item Likewise, for a valuation task, individual $i$'s underlying preferences involve an indifference value  $y^*_{vi}$ such that $y^*_{vi} \sim_{vi} X$. For their actual valuation, individual $i$ reports valuation $y_{vi}=y_{vi}^*+\varepsilon_{v}$, where valuation noise $\varepsilon_{v}$ has CDF $F_{v}$. Thus, the probability that individual $i$ states a valuation for $X$ that is lower than $y$ is $\mP_{vi}(y) = \mP(y_{vi}^*+\varepsilon_{v} \le y)=F_{v}(y-y_{vi}^*)$.
    \item We assume that $\varepsilon_{c}$ and $\varepsilon_{v}$ are mean zero.
\end{itemize}

The key null hypothesis that we assess in this paper is whether each individual has a stable underlying preference that is used for both choice tasks and valuation tasks. We state this formally as Assumption \ref{as_key-null}:

\setcounter{as}{-1}
\begin{as}[stable preferences]\label{as_key-null}
Each individual $i$ has a stable indifference value; in other words, each individual $i$ has $y_{ci}^* = y_{vi}^* \equiv y_i^*$.
\end{as}

With stable preferences, the two key equations are:
\begin{equation}\label{eqn:indiv_key_eqns}
    \mP_{ci}(y) = \mP(\varepsilon_c \le y - y_i^*) = F_c(y-y_i^*) ~~\text{ and } ~~ \mP_{vi}(y) = \mP(\varepsilon_v \le y - y_i^*) = F_v(y-y_i^*).
\end{equation}

This formulation assumes that decision noise enters at the level of the response variable---that is, for the response variable $y$, a person has an underlying value $y^*_i$ that reflects their preferences, but their behavior is influenced by an additive distortion to $y^*_i$. In many analyses, it is assumed instead that decision noise enters at the level of utility---for instance, researchers often assume that decisions under risk are determined from expected utility preferences combined with an additive distortion to utility. In \cref{ap:general_model}, we demonstrate that most of the results below remain unchanged with that formulation.

We are interested in developing predictions for what combinations of individual-level decision probabilities $(\mP_{ci}(y),\mP_{vi}(y))$ are possible under stable preferences combined with various ancillary assumptions. Specifically, we focus on the following notion of attainability:

\begin{definition}
    Given a set of assumptions on the distribution over $(\varepsilon_c,\varepsilon_v)$, the \textit{attainable set} $\mathcal{A}_i$ is the set of all $(a,b) \in [0,1]^2$ such that $\mP_{ci}(y)=F_{c}(y-y_i^*)=a$ and $\mP_{vi}(y)=F_{v}(y-y_i^*)=b$ for some $y_i^*$ and some distribution over $(\varepsilon_c,\varepsilon_v)$ that is consistent with the assumptions. 
\end{definition} 

In this framework, choices and valuations are both noisy measures of underlying preferences. The individual decision probabilities $\mP_{ci}(y)$ and $\mP_{vi}(y)$ depend on a combination of (i) how the offered $y$ compares to the individual's underlying preference $y_i^*$ and (ii) how the person's behavior is impacted by choice noise versus valuation noise. To build intuition, we consider a few examples. 

First, note that if the choice noise and the valuation noise are identical (i.e., $F_{c}(x)=F_{v}(x)$ for all $x$), then \[\mP_{ci}(y) =F_{c}(y-y_i^*)=F_{v}(y-y_i^*)=\mP_{vi}(y).\] Hence:

\begin{remark}\label{rem:identical}
If noise is identical for choices and valuations, then $\mathcal{A}_i = \{(a,b)|a=b\}$.
\end{remark}

Thus, one way to justify the simple null of $\mP_{ci}(y)=\mP_{vi}(y)$ is to assume that noise is identical for choices and valuations. Without supporting evidence, however, this seems a very strong assumption, and indeed for our structural analysis in \cref{sec:MLEestimates}, if we were to assume stable preferences, we would conclude that choices are noisier than valuations. Once we permit differential noise for choices versus valuations, the attainable set can expand considerably. Indeed, Example 1 highlights that, if the only restriction on choice and valuation noise is that they both must be mean zero, any pattern is possible for $\mP_{ci}(y)$ and $\mP_{vi}(y)$.

\bigskip

\noindent \textbf{Example 1}: For any $y$ and any $(a,b)\in (0,1)^2$, one can find a $y_i^*$ and mean-zero $F_{c}$ and $F_{v}$ such that 
$\mP_{ci}(y)=F_c(y-y_i^*)=a$ and $\mP_{vi}(y)=F_v(y-y_i^*)=b$. For instance, given a $y$, choose a $y_i^*<y$ and a $\delta > y-y_i^*$, and then assume
$$
\varepsilon_{c} = \left\{ 
\begin{matrix}
\delta & \text{with probability $1-a$} \\
\frac{-(1-a)\delta}{a} & \text{with probability $a$}
\end{matrix}
\right.
\qquad
\varepsilon_{v} = \left\{ 
\begin{matrix}
\delta & \text{with probability $1-b$} \\
\frac{-(1-b)\delta}{b} & \text{with probability $b$}.
\end{matrix}
\right.
$$
By construction, a person with this $y_i^*$ and these noise distributions would have $\mP_{ci}(y)=a$ and $\mP_{vi}(y)=b$.

\bigskip

Example 1 highlights that, if we want to have any hope of rejecting the null of stable preferences, we must impose restrictions on the nature of noise. One common restriction in applications is that noise is symmetric about its mean of zero.\footnote{To simplify the exposition, whenever we assume a distribution is symmetric, we also assume it is continuous; however, except for a few minor details, continuity can be relaxed without affecting our results.}

\begin{as}[symmetric noises]\label{as_even_c_v}
The distributions of choice noise and valuation noise are both continuous and symmetric, i.e., $F_c$ and $F_v$ are both continuous with $F_{c}(x)=1-F_{c}(-x)$ and $F_{v}(x)=1-F_{v}(-x)$ for all $x$.
\end{as}

\cref{fig:individual_curves} illustrates what the attainable set $\mathcal{A}_i$ might look like under some known distributions for $\varepsilon_c$ and $\varepsilon_v$ that are consistent with symmetric noise. The curve in panel A depicts a generic $\mathcal{A}_i$ for fixed $F_{c}$ and $F_{v}$ that are symmetric and continuous with the same convex support. In general, once we fix an $F_{c}$, an $F_{v}$, and a $y$---so the only unknown is where $y_i^*$ falls in relation to that $y$---the set of possible $(\mP_{ci}(y),\mP_{vi}(y))$ will be a weakly increasing curve, where different points on the curve correspond to different $y^*_i$ (i.e., a smaller $y^*_i$ implies a larger $y-y^*_i$ and thus both $\mP_{ci}(y)$ and $\mP_{vi}(y)$ must be weakly larger). If $F_{c}$ and $F_{v}$ are symmetric, then this curve must pass through $(\frac{1}{2},\frac{1}{2})$ and be symmetric around $(\frac{1}{2},\frac{1}{2})$. Finally, if $F_{c}$ and $F_{v}$ have the same convex support, this curve must be strictly increasing from $(0,0)$ to $(1,1)$, as depicted by the curve in panel A.

\begin{figure}[t!]
\begin{center}
    \includegraphics[width=\textwidth]{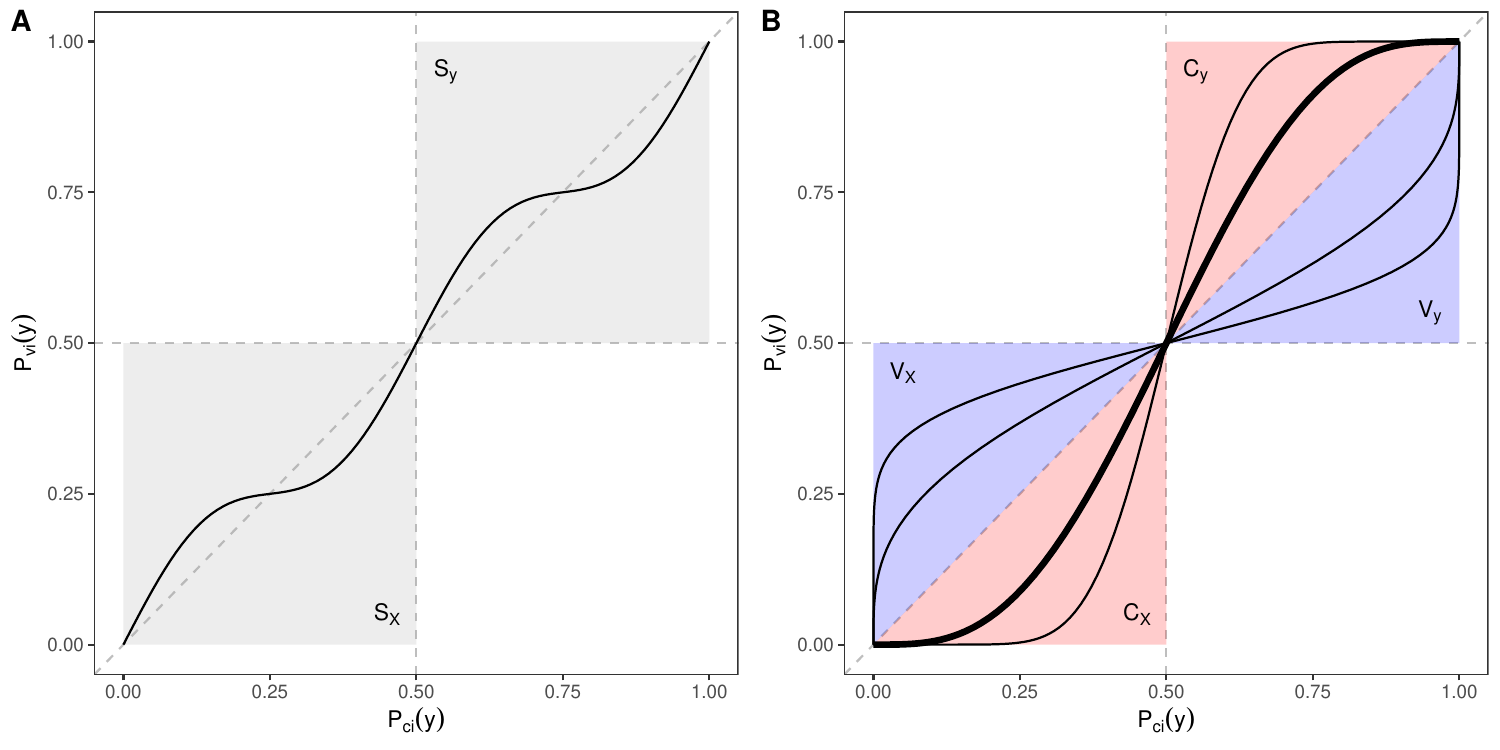}
    \caption{Individual Choice and Valuation Probabilities with Known Noise}
    \label{fig:individual_curves}
    \end{center}
     \footnotesize{\emph{Notes:} Figure presents possible $(\mP_{ci}(y),\mP_{vi}(y))$ combinations given fixed choice and valuation noise distributions $F_{c}$ and $F_{v}$. Panel A depicts possible combinations for generic $F_{c}$ and $F_{v}$ that are symmetric and continuous with the same convex support; any such curve must be contained in regions $S_X$ and $S_y$. Panel B  depicts possible combinations when $F_{c}$ and $F_{v}$ are both mean-zero normal distributions with $ k \equiv \frac{\sigma_{c}^2}{\sigma_{v}^2} \in \lbrace 4, 2, \frac{1}{2}, \frac{1}{4} \rbrace$; bold line corresponds to $k = 2$. When $k>1$ (choices noisier than valuations), curve must be contained in regions $C_X$ and $C_y$; when $k<1$ (valuations noisier than choices), curve must be contained in regions $V_X$ and $V_y$.} 
\end{figure} 

Any weakly increasing curve that passes through $(\frac{1}{2},\frac{1}{2})$ must be contained within the two shaded regions in panel A denoted $S_y$ and $S_X$. Since these will appear in our results, we provide formal definitions:
$$
S_y \equiv \left[\frac{1}{2},1\right]^2 \qquad \text{and} \qquad S_X \equiv \left[0,\frac{1}{2}\right]^2.
$$
The labels reflect that, with symmetric noises, the outcome for individual $i$ will be in $S_y$ when $y \succ_i X$ and in $S_X$ when $X \succ_i y$. 

Panel B of Figure \ref{fig:individual_curves} depicts the case where $F_{c}$ and $F_{v}$ are both known mean-zero normal distributions with variances $\sigma^2_{c}$ and $\sigma^2_{v}$. The four curves reflect four different values of $k \equiv \sigma^2_{c}/\sigma^2_{v}$. When the ratio of the variances is larger than one, which means that choices are noisier than valuations, within $S_y$ the curve must lie above the 45-degree line in the region labeled $C_y$, whereas within $S_X$ the curve must lie below the 45-degree line in the region labeled $C_X$. Intuitively, when choices are noisier than valuations, $\mP_{ci}(y)$ will be closer to $\frac{1}{2}$ (which reflects random choice), while $\mP_{vi}(y)$ will be closer to 0 or 1 (which reflects implementing one's underlying preference). Analogously, when the ratio of the variances is smaller than one, which means that valuations are noisier than choices, within $S_y$ the curve must lie below the 45-degree line in the region labeled $V_y$, whereas within $S_X$ the curve must lie above the 45-degree line in the region labeled $V_X$. Finally, the further is the ratio of variances from one, the more the curve deviates from the 45-degree line.

Again, the curves in Figure \ref{fig:individual_curves} reflect the attainable set for known noise distributions. When analyzing experimental data, however, one won't know the specific noise distributions, and thus we want to identify the attainable set when we permit all possible noise distributions that satisfy a set of assumptions. From Figure \ref{fig:individual_curves}, one can see that under symmetric normal-normal noise, if we have the flexibility to choose any $\sigma^2_{c}$ and $\sigma^2_{v}$, we can produce a curve that goes through any point strictly in the interior of $S_y \cup S_X$. \cref{lm_weak_ind} establishes that under more general symmetric noise, the conclusion is much the same except we can generate outcomes along the edges as well.\footnote{The proof for \cref{lm_weak_ind} and all propositions are collected in Appendix \ref{ap:proof}.}

\begin{lemma}\label{lm_weak_ind}
Under symmetric noises:    
 \begin{enumerate}
     \item $\mathcal{A}_i = S_y \cup S_X$;
     \item If $y > y^*_i$, $(\mP_{ci}(y),\mP_{vi}(y))\in S_y$, and if $y < y^*_i$, $(\mP_{ci}(y),\mP_{vi}(y))\in S_X$; and 
     \item If $y = y^*_i$, then $(\mP_{ci}(y),\mP_{vi}(y)) = \left( \frac12,\frac12 \right)$.
 \end{enumerate}  
\end{lemma}

\bigskip

The intuition for \cref{lm_weak_ind} is exactly as already expressed above. Whenever $y\ge\yist$ and thus $y \succsim_i X$ in terms of underlying preferences, symmetric noise implies $\mP_{ci}(y) \ge \frac{1}{2}$ and $\mP_{vi}(y) \ge \frac{1}{2}$ (because $\mP_{ci}(y)=F_{c}(y-\yist)\ge F_{c}(0) = \frac{1}{2}$ and $\mP_{vi}(y)=F_{v}(y-\yist)\ge F_{v}(0) = \frac{1}{2}$). An analogous intuition holds whenever $y\le\yist$ and thus $y \precsim_i X$ in terms of underlying preferences.\footnote{Note that \cref{as_even_c_v} is sufficient but not necessary for \cref{lm_weak_ind}. All that is required is that $F_{c}$ and $F_{v}$ have the same median, in which case $y-y_i^*$ larger than that median must yield an outcome in $S_y$ and $y-y_i^*$ smaller than that median must yield an outcome in $S_X$. Also note that part (3) of Lemma \ref{lm_weak_ind} is a detail that relies on continuity of $F_c$ and $F_v$; without continuity, $y=y^*_i$ implies $(\mP_{ci}(y),\mP_{vi}(y))\in S_y$.}

The symmetric normal-normal case in Figure \ref{fig:individual_curves} also illustrates the impact of choices being noisier than valuations versus valuations being noisier than choices. We first formalize the regions in panel B of Figure 1:
\[
C_y \equiv S_y \cap \{(a,b)\mid b \geq a\} \qquad \text{and} \qquad C_X \equiv S_X \cap \{(a,b)\mid b \leq a\},
\]
and
\[
V_y \equiv S_y \cap \{(a,b)\mid b \leq a\} \qquad \text{and} \qquad V_X \equiv S_X \cap \{(a,b)\mid b \geq a\}.
\]
In the normal-normal case, when choices are noisier than valuations (i.e., $\sigma^2_c > \sigma^2_v)$, the outcome must lie in $C_y \cup C_X$, whereas when valuations are noisier than choices (i.e., $\sigma^2_v > \sigma^2_c)$, the outcome must lie in $V_y \cup V_X$. This conclusion generalizes to a more general definition of the noises being ordered (which we introduce later as \cref{as_ordered}). However, since one could easily imagine noises not being ordered or being ordered in different ways for different people, and since one cannot easily assess whether the noises are ordered without assuming additional structure, we do not yet make any assumptions on this dimension.

One might wonder whether additional restrictions on noise at the individual level could shrink $\mathcal{A}_i$ to something smaller than $S_y \cup S_X$. However, because the normal-normal case can attain any point in the interior of $S_y \cup S_X$, any such additional restriction would need to be something that excludes the normal-normal case. Since this seems unappealing from an applied perspective, we impose no further restrictions.

\cref{lm_weak_ind} characterizes the attainable set at the individual level under stable preferences and symmetric noises. If we had many repeated observations of choices and valuations at the individual level so that we could measure $\mP_{ci}(y)$ and $\mP_{vi}(y)$, we could assess the null of stable preferences combined with symmetric noises at the individual level. However, in experiments, we virtually always observe aggregate data from a heterogeneous population. Hence, before we can develop tests of the null of stable preferences, we need to build a model of aggregate behavior.

\subsection{Aggregate Behavior}\label{sub:agg}

To specify a model of stochastic behavior at the aggregate level, we need to introduce heterogeneity in preferences and heterogeneity in decision noise. To do so, we use the following structure:

\begin{itemize}
    \item Heterogeneity in preferences: Individual preferences are realizations of a random variable $y_i^* = \mu + \eta$, where $\mu \in \R$ is the population mean preference, and $\eta$ is a mean-zero random variable with CDF $H$. Under Assumption \ref{as_key-null}, each individual gets a single draw of $\eta$ (and thus a single $y_i^*$) that applies for both choices and valuations. 
    \item Heterogeneity in decision noise: The additive noise shocks $\varepsilon_c$ and $\varepsilon_v$ can be correlated with $\eta$.
\end{itemize}

With this structure, the aggregate decision probabilities are:
\begin{align*}
     \mP_c(y) &= \mP(y^*_i + \varepsilon_c\le y) = \mP(\mu+\eta+\varepsilon_c\le y); \\
   \mP_v(y) &= \mP(y^*_i + \varepsilon_v\le y) = \mP(\mu+\eta+\varepsilon_v\le y).
\end{align*}

We are now interested in the attainable set $\mathcal{A}$ for the population proportions $(\mP_c(y),\mP_v(y))$, which we define analogously to the attainable set $\mathcal{A}_i$ for $(\mP_{ci}(y),\mP_{vi}(y))$. 

\begin{definition}
    Given a set of assumptions on the distribution over $(\eta,\varepsilon_c,\varepsilon_v)$, the \textit{attainable set} $\mathcal{A}$ is the set of all $(a,b) \in [0,1]^2$ such that $\mP_{c}(y)=a$ and $\mP_{v}(y)=b$ for some $\mu$ and some distribution over $(\eta,\varepsilon_c,\varepsilon_v)$ that is consistent with the assumptions. 
\end{definition} 

We first note that if we make no assumptions beyond symmetry of each individual's noises, that is, symmetry of $\varepsilon_c|\eta$ and $\varepsilon_v|\eta$, then the attainable set at the population level is the convex hull of $S_y \cup S_X$. 

\begin{proposition}\label{prop_big_hexagram}
    Under stable preferences and (conditionally) symmetric noises, $\mathcal{A} = \con(S_y \cup S_X)$.
\end{proposition}

Panel A in \cref{fig:aggregate} depicts the attainable set from \cref{prop_big_hexagram}. The intuition is straightforward. \cref{lm_weak_ind} establishes that for any point in $S_y \cup S_X$, an individual could exhibit those decision probabilities. If we then put no restrictions on the distribution of preferences and permit people with different $y_i^*$ to have different noise distributions, we can create a population that consists of any distribution across points in $S_y \cup S_X$. It immediately follows that we can get average behavior anywhere in $\con(S_y \cup S_X)$.

\begin{figure}[h!]
\begin{center}
    \includegraphics[width=\textwidth]{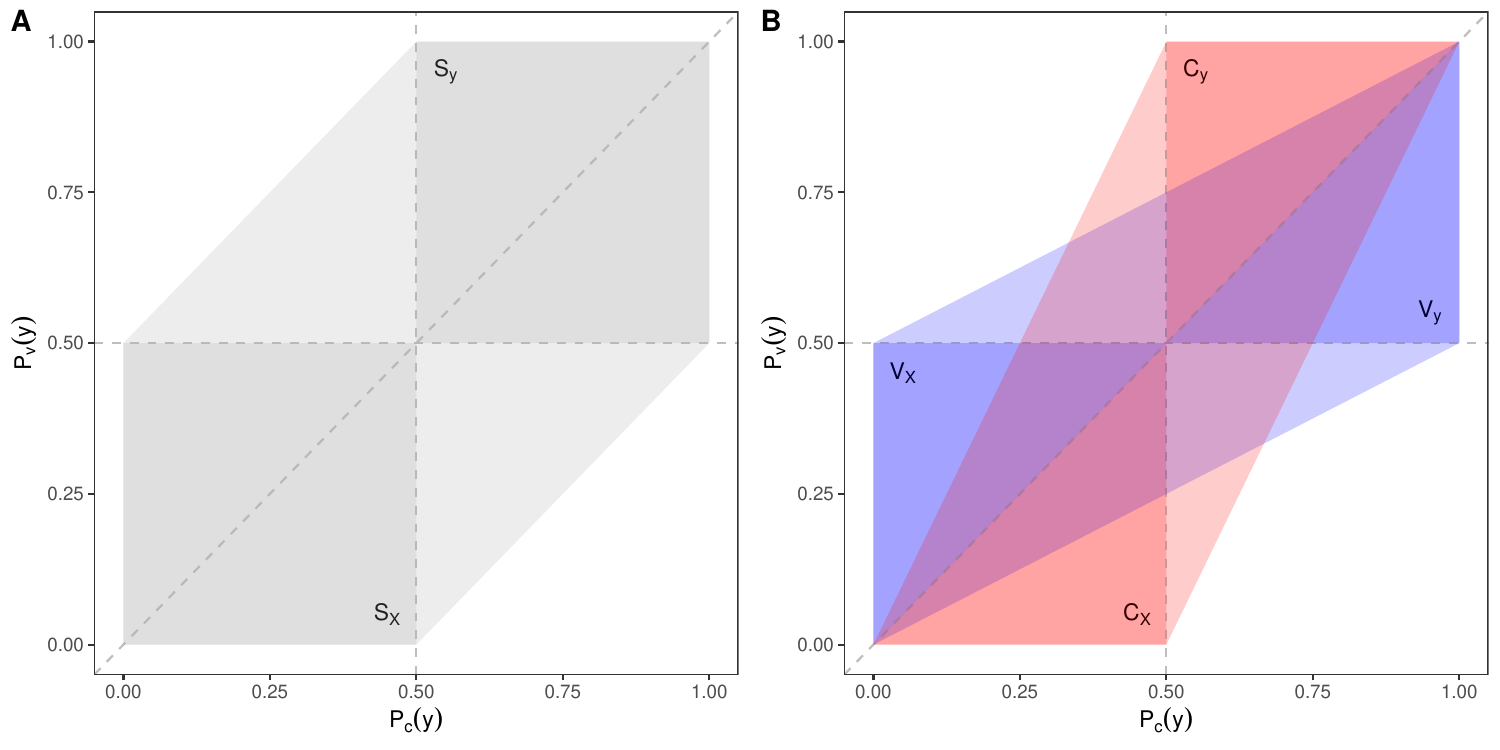}
    \caption{Attainable Sets from \cref{prop_big_hexagram,prop_weak_pop}}
    \label{fig:aggregate}
    \end{center}
     \footnotesize{\emph{Notes:} 
     Figure presents possible $(\mP_{c}(y),\mP_{v}(y))$ combinations given unknown choice and valuation noise that satisfy assumptions of \cref{prop_big_hexagram,prop_weak_pop}. Panel A depicts possible combinations under unknown noise that is only symmetric (Assumption \ref{as_even_c_v}), with attainable set $\mathcal{A} = \con(S_y \cup S_X)$. Panel B depicts possible combinations under unknown noise that is both symmetric and independent of preferences (Assumptions \ref{as_even_c_v} and \ref{as_prefs-noise-independent}), with attainable set $\mathcal{A} = \con(C_y \cup C_X) \cup \con(V_y \cup V_X)$.}
\end{figure}

\cref{prop_big_hexagram} suggests a possible test of the null of stable preferences: If we want to assume very little---specifically, only that everyone has symmetric noises---then to reject the null of stable preferences, we would need to reject that an experiment lies in $\con(S_y \cup S_X)$. As we will see, nearly the entirety of the preference-reversals literature lies within this set. Thus, one might be interested in imposing additional assumptions to gain restrictions on the set of attainable outcomes.

A natural next step to restrict the attainable set is to assume that preferences and decision noise are independent from each other:

\begin{as}[preferences independent of noise]\label{as_prefs-noise-independent}
$(\varepsilon_c,\varepsilon_v)$ is independent of $\eta$, and thus all individuals have the same CDF $F_c$ for $\varepsilon_c$, and all individuals have the same CDF $F_v$ for $\varepsilon_v$.
\end{as}

Under Assumption \ref{as_prefs-noise-independent}, for any fixed $F_c$ and $F_v$ that are symmetric and the same for everyone, the set of $(\mP_{ci}(y),\mP_{vi}(y))$ that could emerge across different $y_i^*$ must be a weakly increasing curve that is symmetric around $(\frac{1}{2},\frac{1}{2})$ (as in Figure \ref{fig:individual_curves}). Without any restrictions on the distribution of $\eta$, we could then attain any $(\mP_c(y),\mP_v(y))$ in the convex hull of that curve. Because we can choose any symmetric $F_c$ and $F_v$ that we want, it follows that a point is in the attainable set $\mathcal{A}$ if it is part of the convex hull of any weakly increasing curve that is symmetric around $(\frac{1}{2},\frac{1}{2})$. Applying this insight yields Proposition \ref{prop_weak_pop}.

\begin{proposition}\label{prop_weak_pop}
    Under stable preferences and symmetric noises, if preferences are independent of noise, then $\mathcal{A} = \con(C_y \cup C_X) \cup \con(V_y \cup V_X)$.
\end{proposition}  

To provide some intuition for Proposition \ref{prop_weak_pop}, consider the restriction of the attainable set relative to Proposition \ref{prop_big_hexagram}. Specific values of $(\mP_c(y), \mP_v(y))$ are eliminated; these values correspond to combinations of correlated types (e.g., combining individuals exhibiting $y_i^*>y$ and valuations noisier than choices with those exhibiting  $y_i^*<y$ and choices noisier than valuations). With preferences independent of noise, $(\mP_c(y), \mP_v(y))$  values attainable by relying on such correlation are no longer permissible.

\cref{prop_weak_pop} suggests another possible test of the null of stable preferences: If we are willing to assume that everyone has symmetric noises and that preferences and noise are independent, then to reject the null of stable preferences we would need to reject that an experiment lies in $\con(C_y \cup C_X) \cup \con(V_y \cup V_X)$. This region is illustrated in panel B of \cref{fig:aggregate}. Notably, whether preferences are required to be independent of noise makes relatively little difference to the set of attainable outcomes, such that many studies (as we will see) will lie within the attainable regions for both Proposition \ref{prop_big_hexagram} and \ref{prop_weak_pop}.

Hence, further assumptions are required to more meaningfully restrict the set of attainable outcomes. Propositions \ref{prop_big_hexagram} and \ref{prop_weak_pop} rely on the freedom to use any distribution of preferences. We next restrict preferences to be symmetric.

\begin{as}[symmetric preferences]\label{as_sure_sym_delta}
    The distribution of preferences is continuous and symmetric, i.e., $H$ is continuous with $H(x)=1-H(-x)$ for all $x$.
\end{as} 

Under Assumptions \ref{as_even_c_v}, \ref{as_prefs-noise-independent}, and \ref{as_sure_sym_delta}, it is natural to define the random variables $y_{c}$ and $y_{v}$ as
$$
y_{c} \equiv y_i^* + \varepsilon_{c} = \mu + \eta + \varepsilon_{c} = \mu + \nu_{c}
$$
and
$$
y_{v} \equiv y_i^* + \varepsilon_{v} = \mu + \eta + \varepsilon_{v} = \mu + \nu_{v}
$$
where $\nu_{c}=\eta + \varepsilon_{c}$ and $\nu_{v}=\eta + \varepsilon_{v}$. Note that $\nu_{c}$ and $\nu_{v}$ are both sums of independent symmetric mean-zero random variables and thus are themselves symmetric mean-zero random variables. Letting $G_c$ and $G_v$ denote the CDFs for $\nu_{c}$ and $\nu_{v}$, we can write the aggregate decision probabilities as:
\begin{equation}\label{eqn:forRemark2_b}
\mP_c(y)=\mP(y \ge \mu + \nu_{c}) = G_c(y-\mu) \quad \text{and} \quad \mP_v(y)=\mP(y \ge \mu + \nu_{v}) = G_v(y-\mu).
\end{equation}
Comparing \eqref{eqn:indiv_key_eqns} and \eqref{eqn:forRemark2_b}, we have:

\bigskip

\begin{remark}\label{rem:observe}
  In terms of the attainable set, an aggregate-level model with symmetric noise and symmetric preferences with mean $\mu$ that are independent of noise is observationally equivalent to an individual-level model with $y_i^* = \mu$ and symmetric mean-zero noise.   
\end{remark}
\cref{rem:observe} immediately leads to Proposition \ref{prop:symm_prefs}, which is the analogue of Lemma \ref{lm_weak_ind}:

\begin{proposition}\label{prop:symm_prefs}
Under stable preferences and symmetric noises, if preferences are symmetric and independent of noise, then:
 \begin{enumerate}
     \item $\mathcal{A} = S_y \cup S_X$;
     \item If $y > \mu$, $(\mP_c(y),\mP_v(y)) \in  S_y$, and if $y < \mu$,  $(\mP_c(y),\mP_v(y)) \in  S_X$; and
     \item If $y = \mu$, then $(\mP_c(y),\mP_v(y)) = \left( \frac12,\frac12 \right)$.
 \end{enumerate}  
\end{proposition}

Part 1 of \cref{prop:symm_prefs} provides another test of the null of stable preferences when comparing aggregate shares in an experiment. If we are willing to make the ancillary assumptions that noise is symmetric and that preferences are symmetric and independent of noise, then we can test the null of stable preferences by assessing whether an experiment lies outside of $S_y \cup S_X$. Stated differently, \cref{prop:symm_prefs}(1) requires congruence between the signs of $\mP_c(y)-\frac{1}{2}$ and $\mP_v(y)- \frac{1}{2}$.\footnote{Analogous to part (3) of Lemma \ref{lm_weak_ind}, part (3) or Proposition \ref{prop:symm_prefs} is a detail that relies on continuity of $F_c$ and $F_v$; without continuity, $y=\mu$ implies $(\mP_{c}(y),\mP_{v}(y))\in S_y$.}

While Proposition \ref{prop:symm_prefs} permits $\mP_c(y)=\mP_v(y)$, it does not provide conditions under which we would predict $\mP_c(y)=\mP_v(y)$, which would then reflect conditions under which standard tests are valid. It turns out that, under two further assumptions, we would predict $\mP_c(y)=\mP_v(y)$ only in one special case.

\begin{as}[unimodal preferences]\label{as_unimodal}
    $H$ has a strictly unimodal density $h$ on a convex support $S$. That is, for any $x_1 > x_2 \ge 0$, $h(x_1) \le h(x_2)$ and $h(-x_1) \le h(-x_2)$, where both inequalities are strict when $x_1 \in S$.
\end{as}

\begin{as}[ordered noise]\label{as_ordered}
At least one of the following holds:
\begin{enumerate}
    \item Choices are noisier than valuations:
    $F_v(z)<F_c(z)$ whenever $z<0$ and $F_v(z)>0$, and
    $F_v(z)>F_c(z)$ whenever $z>0$ and $F_v(z)<1$.
    
    \item Valuations are noisier than choices:
    $F_c(z)<F_v(z)$ whenever $z<0$ and $F_c(z)>0$, and
    $F_c(z)>F_v(z)$ whenever $z>0$ and $F_c(z)<1$.
\end{enumerate}
\end{as}

The assumption of unimodal preferences is strong, but it is commonly used (e.g., if preferences are normally distributed). The assumption of ordered noise generalizes what we saw for the normal-normal case in \cref{fig:individual_curves}, where $\sigma^2_c/\sigma^2_v$ greater or less than one determined whether choices or valuations are noisier. 

\begin{proposition}\label{prop:cali}
Under stable preferences and symmetric and ordered noises, if preferences are symmetric, unimodal, and independent of noise, then $\mP_c(y)=\mP_v(y)\in (0,1)$ if and only if $y = \mu$, in which case $\mP_c(y)=\mP_v(y)=\frac{1}{2}$.
\end{proposition}

\cref{prop:cali} highlights a different problem with the typical approach of testing whether the aggregate shares are the same. Up to now, we have emphasized that a simple test of the null $\mP_c(y)=\mP_v(y)$ is not enough to test the null of stable preferences unless one is willing to make very strong ancillary assumptions (e.g., that the choice noise and the valuation  noise are identical). In \cref{prop:cali}, we see a second problem: Under the types of regularity assumptions frequently made in economics---symmetric continuous noise and symmetric unimodal preferences---stable preferences combined with differential noise across elicitation techniques would in fact yield $\mP_c(y)=\mP_v(y)$ only if we happen to be conducting the experiment at $y=\mu$. Since the population mean preference is typically something one would not know when designing an experiment, a strong interpretation is that, in fact, under stable preferences, we should not expect to see $\mP_c(y)=\mP_v(y)$ in any experiments. 

\subsection{Leveraging Richer Data for Additional Tests}\label{sec:leveraging_richer-data}

In line with the prior literatures on preference reversals and on comparing choices versus valuations, Section \ref{sub:agg} focused on predictions for $\mP_c(y)$ and $\mP_v(y)$. However, the null of stable preferences combined with various ancillary assumptions also makes qualitative predictions for other objects. In this section, we highlight two such predictions, where each takes advantage of richer data structures (if they are available).

\subsubsection{Using Quantitative Information on Mean Preferences from Valuations}\label{subsub:prop-3-2}

Part 2 of Proposition \ref{prop:symm_prefs} suggests another test of stable preferences (combined with Assumptions~\ref{as_even_c_v}--\ref{as_sure_sym_delta}): If $y>\mu$, then $\mP_c(y) \ge \frac{1}{2}$, and if $y<\mu$, then $\mP_c(y) \le \frac{1}{2}$. Testing this prediction requires knowledge of $\mu$. While we cannot observe $\mu$, if a dataset contains the full distribution of $y_{vi}$ in the valuation task, we can obtain an estimate of $\mu$. We can then evaluate Proposition \ref{prop:symm_prefs}(2) using that empirical analog of $\mu$. We describe one such implementation in Section \ref{subsec:qual_test}.

\subsubsection{Analyzing Conditional Choices}

If one has linked data at the individual level---that is, for each individual we observe both their choice between $y$ and $X$ and their stated valuation $y_{vi}$ for $X$---additional tests become possible.

Consider the following natural conjecture: Because it is true under symmetric noise that
$$
\mP_c(y|y=y_i^*) = \frac{1}{2}
$$
and that $\E{y_{vi}}=y_i^*$, one might conjecture that\footnote{While the conditioning event is null, this is formally a regular conditional probability distribution, which is uniquely defined up to a null set. Also note that we can observe $\mP_c(y|y=y_{vi})$ only if $y_{vi}-\mu$ is in the support of $\eta+\varepsilon_v$; hence, everything below makes this assumption.}
$$
\mP_c(y|y=y_{vi}) = \frac{1}{2}.
$$

It turns out that this property is unlikely to hold even in the types of symmetric, smooth, and independent environments often used in economics. In particular, suppose we assume symmetric noises and stable preferences that are symmetric, unimodal, and independent of noise. In addition, assume:

\begin{as}[independent noise]\label{as_indep}
    $\varepsilon_{c}$ and $\varepsilon_{v}$ are independent.
\end{as}

\begin{as}[positive densities]\label{as_f_pos}
$\varepsilon_c$ and $\varepsilon_v$ have densities $f_c$ and $f_v$, respectively, that are positive on an open interval containing zero. 
\end{as}

The following proposition establishes that, under these conditions, this conjecture holds only if $y$ is equal to the mean preference.

\begin{proposition}\label{prop:cond_cali}
Under stable preferences and symmetric noises, if preferences are symmetric, unimodal, and independent of noise, and the noise is independent with positive densities, then $ \mP_c(y\mid y=y_{vi}) =\frac{1}{2}$ if and only if $y = \mu$. Moreover, if $y > \mu$ then $\mP_c(y\mid y=y_{vi}) > \frac{1}{2}$, and if $y<\mu$ then $\mP_c(y\mid y=y_{vi}) < \frac{1}{2}$.
\end{proposition}

The intuition for \cref{prop:cond_cali} is easy to see when the noise is also unimodal. Consider an experiment with $y > \mu$. When we observe that a person reports a valuation $y_{vi} = y$, it need not be the case that $\E{y^*_i|y_{vi} =y}=y$; indeed, with symmetric and unimodal preferences and symmetric and unimodal noise, it must be that $\E{y^*_i|y_{vi} =y}$ is in between $\mu$ and $y$. Hence, for this person, $y$ will be relatively attractive, and thus it leads to a conditional choice probability greater than $\frac{1}{2}$. An analogous argument holds for an experiment with $y < \mu$, leading to a conditional choice probability less than $\frac{1}{2}$.

On one hand, this result is a point of caution. It is tempting to think that one way to test for stable preferences is to test whether $\mP(y|y_{vi}=y)=\frac{1}{2}$. Analogous to our discussion after \cref{prop:cali}, however, such a test is valid only if we are confident that we are conducting an experiment precisely at $y=\mu$.

On the other hand, \cref{prop:cond_cali} suggests another test of stable preferences (combined with Assumptions \ref{as_even_c_v}-\ref{as_unimodal} and \ref{as_indep}-\ref{as_f_pos}): If $y>\mu$ then $\mP_c(y|y_{vi}=y) > \frac{1}{2}$, and if $y<\mu$ then $\mP_c(y|y_{vi}=y) < \frac{1}{2}$. As for the prediction in \cref{subsub:prop-3-2}, testing this prediction requires knowledge of $\mu$. Again, if a linked dataset contains the full distribution of $y_{vi}$ in the valuation task, we can obtain an estimate of $\mu$ and conduct the exercise with the empirical analog of $\mu$. We describe one such implementation in Section \ref{sec:prop3_part2}.

\subsection{Tests Based on Estimates of Structural Parameters}\label{sec:nnn_analysis}

Up to now, our assumptions have involved qualitative features of the distributions of noise and preferences and have led to assessments of stability free from precise distributional assumptions. If one is willing to make structural assumptions about the distributions, and if one has the appropriate variation in one's data, one can develop tests based on estimates of structural parameters. Moreover, with such estimates, one can also obtain quantitative estimates of the magnitudes of instability. In this section, we present an example of such a structural approach in which we assume all distributions are normal.

\subsubsection{Underlying Structural Assumptions}

Suppose that the population distribution of underlying preferences is given by
\begin{equation*}
 \begin{pmatrix}
 y^*_{ci} \\
 y^*_{vi} \\
\end{pmatrix}
\sim N \left (
 \begin{pmatrix}
 \mu_c  \\
\mu_v \\
\end{pmatrix} ,
 \begin{pmatrix}
\gamma_c^2   & \gamma_{cv} \\
\gamma_{cv} & \gamma_v^2 \\
\end{pmatrix} 
\right).
\end{equation*}
With this structure, the null of stable preferences would mean that $\mu_c = \mu_v$ and $\gamma_c^2=\gamma_v^2=\gamma_{cv}$ (which together imply that every individual has $y^*_{ci}=y^*_{vi}$). Because the approaches below do not provide direct estimates for $\gamma_c^2$ and $\gamma_v^2$, we focus our analysis on testing the null of $\mu_c = \mu_v$, that is, of stable mean preferences.

We assume that noise is independent of preferences, and we further assume that $\varepsilon_{c}$ and $\varepsilon_{v}$ both have mean-zero normal distributions with variances $\sigma^2_c$ and $\sigma^2_v$, respectively. Finally, we assume the two types of noise are independent of each other.\footnote{Note that this structure satisfies all of Assumptions \ref{as_even_c_v}-\ref{as_f_pos}, unless $\sigma^2_c=\sigma^2_v$, in which case ordered noise only weakly holds.}

Recall that a person chooses $y$ over $X$ when $y$ is larger than $y_{ci} \equiv y_{ci}^*+\varepsilon_{c}$, and the person reports valuation for $X$ of $y_{vi} \equiv  y_{vi}^*+\varepsilon_{v}$. Under the assumptions above, the joint distribution of  $y_{ci}$ and $y_{vi}$ is bivariate normal: 
\begin{equation}\label{eqn:struct_distn}
 \begin{pmatrix}
 y_{ci} \\
 y_{vi} \\
\end{pmatrix}
\sim N \left (
 \begin{pmatrix}
 \mu_c  \\
\mu_v \\
\end{pmatrix} ,
 \begin{pmatrix}
\sigma^2_c + \gamma_c^2   & \gamma_{cv} \\
\gamma_{cv} & \sigma^2_v + \gamma_v^2 \\
\end{pmatrix} 
\right).
\end{equation}

Equation \eqref{eqn:struct_distn} can be used to generate structural predictions from which one can base additional tests of stable preferences and quantification of deviations therefrom.

\subsubsection{Probit Formulation}\label{subsub:probit_theory}

If one collects within-subjects responses for both  choices and valuations, this framework generates a prediction for $\mP_c(y|y_{vi})$. In particular, because  $y_{ci}$ and $y_{vi}$ are joint normal, the conditional distribution of $y_{ci}$ given a specific valuation of $y_{vi}$ is normal.\footnote{\label{FN:probit_cond_distn}Specifically, $y_{ci}|_{y_{vi}} \sim N\left( \mu_c + \frac{\gamma_{cv}}{ \sigma^2_v + \gamma_v^2 }(y_{vi} - \mu_v),V^2 \right)$, where $V^2 \equiv \sigma^2_c + \gamma^2_c - \frac{(\gamma_{cv})^2}{ \sigma^2_v + \gamma_v^2}$.} From that conditional distribution one can derive that 
\[
\mP_c(y|y_{vi}) = \Phi\left( \frac{\mu_v-\mu_c}{V}  + \frac{\left( 1 - \frac{\gamma_{cv}}{ \sigma^2_v + \gamma_v^2 }  \right)}{V} (y - \mu_v) + \frac{\frac{\gamma_{cv}}{ \sigma^2_v + \gamma_v^2 }}{V}(y - y_{vi})  \right)
\]
where $V^2 \equiv \sigma^2_c + \gamma^2_c - \frac{(\gamma_{cv})^2}{ \sigma^2_v + \gamma_v^2}$ and $\Phi$ is the CDF for a standard normal distribution. 
Replacing $y-\mu_v$ with its sample analog  from valuations, $y-\overline{y_{vi}}$, yields a simple probit regression:
\begin{equation}\label{eq:probit_condcal}
         \mP_c(y|y_{vi}) =   \Phi\left( \alpha +  \beta_1 (y - \overline{y_{vi}})  + \beta_2 (y - y_{vi}) \right),  
\end{equation}
where $\alpha = \frac{\mu_v-\mu_c}{V}$, $\beta_1 = \frac{\left( 1 - \frac{\gamma_{cv}}{ \sigma^2_v + \gamma_v^2 }  \right)}{V}$ $\beta_2=  \frac{\frac{\gamma_{cv}}{ \sigma^2_v + \gamma_v^2 }}{V}$,  and $\beta_1 + \beta_2 = \frac{1}{V}$. Hence, if one has ``groupings'' of experiments (i.e.,  multiple values of $y$ for a fixed $X$ such that there is variation in $y - \overline{y_{vi}}$), one can obtain estimates for $\alpha$, $\beta_1$, and $\beta_2$.\footnote{\label{fn:struct_ID}The key challenge is obtaining enough variation in $y-\overline{y_{vi}}$ to separately identify $\alpha$ and $\beta_1$ since $y-\overline{y_{vi}}$ will vary only across experiments through variation in $y$.}

With these estimates in hand, one can test the null of stable mean preferences. Specifically, an estimate for $\mu_v - \mu_c$ can be derived from combining regression coefficients,
$$
\widehat{\mu_v - \mu_c} = \frac{\hat{\alpha}}{\hat{\beta_1} + \hat{\beta_2}},
$$
with standard error derived from the delta method. We can then conduct a formal test of the null of $\mu_v - \mu_c = 0$.

The probit estimates can also be used to quantify the behavioral impact of deviations from stable preferences. The quantity $\mP_c(\mu_v|y_{vi}=\mu_v)$ is the choice probability in the idealized experiment where we are able to select $y=\mu_v$ and in addition focus on participants who state a valuation equal to that $\mu_v$. Under joint normality and stable preferences, this quantity would be $\frac{1}{2}$. Within our probit structure, $\Phi(\hat{\alpha})$ represents an estimate of $\mP_c(\mu_v|y_{vi}=\mu_v)$.

\subsubsection{MLE Formulation}\label{subsub:MLE_theory}

The probit estimator of equation (\ref{eq:probit_condcal}) relies on replacing $y-\mu_v$ with its sample analog $y-\overline{y_{vi}}$. This introduces the possible challenge of measurement error biasing corresponding estimates. Attenuation bias could naturally lead to measures of $\alpha$ closer to zero, working against finding evidence of preference instability. A natural alternative is to estimate mean preferences alongside the other variables of interest by pursuing a maximum likelihood approach.

For such an approach, one again needs within-subjects responses for both  choices and valuations. For an MLE approach, it is also natural to account for the interval nature of most valuations data. In particular, imagine a person provides a valuation in the region $(y_L, y_H)$ (e.g., the interval that corresponds to the switching point in  a MPL eliciting $y_{vi}$), and that when offered a choice between $X$ and  $y$, they choose $y$. The likelihood of such an observation is the probability of $(y_{ci},y_{vi})$ falling within the rectangle $(-\infty, y) \times (y_L, y_H)$. Under the normal-normal-normal specification here, it is straightforward to calculate this probability given the joint normal distribution in equation (\ref{eqn:struct_distn}). Maximizing the log likelihood of observations with respect to the parameters $\mu_c$, $\mu_v$, $(\sigma^2_c + \gamma^2_c)$, $(\sigma^2_v + \gamma^2_v)$, and $\gamma_{cv}$ yields estimates of the corresponding parameters using conventional techniques. See Appendix \ref{appsec:MLE_details} for details.

Because the MLE approach yields estimates for both $\hat \mu_c$ and $\hat \mu_v$, we can directly test the null of $\mu_v - \mu_c = 0$. In addition, under the MLE approach, we obtain an estimate of $\mP_c(\mu_v|y_{vi}=\mu_v)$, which again provides a sense of the behavioral impact of deviations from stable preferences, using $\Phi\left( \frac{\widehat{\mu_v - \mu_c}}{\hat V} \right)$.\footnote{$\hat V$ is the estimate for $V$, the conditional standard deviation for $y_{ci}|_{y_{vi}}$; see footnote \ref{FN:probit_cond_distn}.}

\section{Empirical Application}\label{sec:empirical}

Our theoretical analysis in Section \ref{sec:theory} provides multiple ways to test the null of stable preferences; we summarize these tests in Table \ref{tab:tests_summary}. Note that each test is in fact a joint test of stable preferences combined with some combination of ancillary assumptions. In this section, we conduct an empirical implementation of these tests.

\begin{table}[t!]
\centering
\caption{Summary of Tests of Stable Preferences}\label{tab:tests_summary}
\begin{footnotesize} 
\setlength{\tabcolsep}{2pt} 
\begin{tabularx}{\linewidth}{>{\raggedright\arraybackslash\hangindent=1em}p{2.5cm} *{2}{>{\raggedright\arraybackslash\hangindent=1em}X}}
\toprule
Result & Assumptions 
& Prediction to Test \\
\midrule
\midrule
\addlinespace[0.5em]
\multicolumn{3}{c}{\textbf{Panel A: Simple Qualitative Tests}}\\
\cmidrule(lr){1-3}
\addlinespace[0.3em]
\cref{rem:identical}  & Stable preferences plus identical noise & $\mP_c(y) = \mP_v(y)$ \\
\addlinespace[0.5em]
\cref{prop_big_hexagram} & Stable preferences plus symmetric noise & $(\mP_c(y),\mP_v(y)) \in \con(S_y \cup S_X)$ \\
\addlinespace[0.5em]
\cref{prop_weak_pop} & Stable preferences plus symmetric noise that is independent of preferences & $(\mP_c(y),\mP_v(y)) \in \con(C_y \cup C_X) \cup \con(V_y \cup V_X)$ \\
\addlinespace[0.5em]
\cref{prop:symm_prefs}(1) & Stable preferences plus symmetric preferences and symmetric noise that is independent of preferences & $(\mP_c(y),\mP_v(y)) \in S_y \cup S_X$ \\
\midrule
\addlinespace[0.5em]
\multicolumn{3}{c}{\textbf{Panel B: Qualitative Tests Using Valuations to Estimate $\mu$}}\\
\cmidrule(lr){1-3}
\addlinespace[0.3em]
\cref{prop:symm_prefs}(2) & Stable preferences plus symmetric preferences and symmetric noise that is independent of preferences & If $y>\mu$, $\mP_c(y)>\frac{1}{2}$;
if $y<\mu$, $\mP_c(y)<\frac{1}{2}$ \\
\addlinespace[0.5em]
\cref{prop:cond_cali} & Stable preferences plus symmetric unimodal preferences and symmetric independent noise with positive densities that is independent of preferences & If $y>\mu$, $\mP_c(y|y_{vi}=y)>\frac{1}{2}$; if $y<\mu$, $\mP_c(y|y_{vi}=y)<\frac{1}{2}$ \\
\midrule
\addlinespace[0.5em]
\multicolumn{3}{c}{\textbf{Panel C: Structural Tests Using Linked Choices and Valuations}}\\
\cmidrule(lr){1-3}
\addlinespace[0.3em]
& $\begin{pmatrix}
 y^*_{ci} \\
 y^*_{vi} \\
\end{pmatrix}
\sim N \left (
 \begin{pmatrix}
 \mu_c  \\
\mu_v \\
\end{pmatrix} ,
 \begin{pmatrix}
\gamma_c^2   & \gamma_{cv} \\
\gamma_{cv} & \gamma_v^2 \\
\end{pmatrix} 
\right)$, $\varepsilon_c \sim N(0,\sigma_c^2)$, and $\varepsilon_v \sim N(0,\sigma_v^2)$, all three independent
& Structural (Probit or MLE) estimates for evaluation of $\mu_c = \mu_v$ \\
\bottomrule
\end{tabularx}
\end{footnotesize}
\end{table}

\subsection{Dataset Construction}\label{subsec:data-construct}

We begin by constructing two datasets. The first is constructed using data from the prior literature on preference reversals and the prior literature on choices-versus-MPLs. The second is constructed using data from two recent studies by  \citet{MNOSS-2024-distinguishing, MNOSS-2026-connecting} that are well-suited for the present paper. 

\bigskip

\noindent \textbf{Data from Prior Literature:} To construct a dataset based on the preference-reversal literature, we draw from a recent meta-analysis by \citet{Lu-2026-Meta-Anal}. In Table 2 from Lu's Online Supplementary Material, there are 462 experiments from 54 studies for which one could calculate the empirical quantities $(\widehat{\mP_{c}}(y), \widehat{\mP_{v}}(y))$. A research assistant attempted to collect the 54 original studies to verify (and, in very few instances, adjust) the numbers. The research assistant was able to confirm the numbers for 314 experiments from 42 studies, and that is the data that we use here. See Appendix \ref{appsubsec:PR_data} for details.\footnote{A number of the original studies, including the seminal work of \citet{grether1979economic}, report only combined numbers across multiple conditions---i.e., across multiple ``experiments''. Rather than exclude these, we chose to consider the combined numbers to be one experiment. See Appendix \ref{appsubsec:PR_data} for more details.}

To construct data based on the choices-versus-MPLs literature, we pull the needed data ourselves from the three key studies in this literature \citep{brown2018separated,freeman2019eliciting, freeman2019choice}. Specifically, we identify nine experiments in these studies; see Appendix \ref{appsubsec:CV_data} for details.

On net, then, our prior-literature dataset consists of 323 experiments. For each experiment, we identify the safer and riskier lotteries, and we harmonize the data such that $\widehat{\mP_{c}}(y)$ corresponds to the proportion of subjects choosing the safer alternative and $\widehat{\mP_{v}}(y)$ corresponds to the proportion of subjects stating a higher valuation for the safer alternative (in the case of a preference-reversal experiment) or having a valuation that would indicate a preference for the safer alternative (in the case of a choices-versus-MPLS experiment). Hence, for each experiment, the data include $\widehat{\mP_{c}}(y)$, $\widehat{\mP_{v}}(y)$, and the number of participants for each. Across all 323 experiments, valuations are provided by an average of 101 subjects and choices are provided by an average of 96 subjects, yielding 32,695 total valuations and 30,949 total choices.

\bigskip

\noindent\textbf{Data from \citet{MNOSS-2024-distinguishing, MNOSS-2026-connecting}:} These studies collect choices and valuations from the same subjects. For each study, the goal was not to compare choices versus valuations; rather, their focus was comparing pairs of valuations to study the phenomena of the common ratio effect, the common consequence effect, and preferences for probabilistic mixtures. They collected data on corresponding pairs of choices to illustrate an inference problem that arises when studying paired choice tasks (in \cite{MNOSS-2024-distinguishing}) and to show that paired choices and paired valuations can yield similar conclusions in some instances (in \cite{MNOSS-2026-connecting}). However, their data provide a convenience dataset that we can use for the purposes of this paper.

In \cite{MNOSS-2024-distinguishing}, there are 160 experiments---that is, 160 cases where a set of subjects make a binary choice and also provide a valuation linked to that binary choice. In \cite{MNOSS-2026-connecting}, there are another 264 experiments. For each of these 424 experiments, we identify the safer and riskier lotteries, and we harmonize the data such that $\widehat{\mP_{c}}(y)$ corresponds to the proportion of subjects choosing the safer alternative and $\widehat{\mP_{v}}(y)$ corresponds to the proportion of subjects having a valuation that would indicate a preference for the safer alternative. Hence, for each experiment, the data include $\widehat{\mP_{c}}(y)$, $\widehat{\mP_{v}}(y)$, and the number of participants. Across all 424 experiments, there are 68,448 observations for both choices and valuations, and thus an average of 161 subjects per experiment (again, each experiment is within-subject).

These two studies have two additional advantages. First, the data contain each individual valuation, and thus we can calculate the average valuation  (i.e., $\overline{y_{vi}}$) as a proxy for mean preferences for different samples. Second, the experiments come in precisely the ``groupings'' needed to conduct the structural tests from Section \ref{sec:nnn_analysis}. Specifically, a grouping consists of a set of experiments that all use the same valuation task for a specific lottery $X$, but which differ in $y$.\footnote{This variation in $y$ for a fixed valuation task yields variation in $y-\overline{y_{vi}}$, which is key to identification in our structural approaches (see footnote \ref{fn:struct_ID}).} In \cite{MNOSS-2024-distinguishing}, there are 40 different groupings (40 different $X$'s and corresponding valuation tasks), where each grouping contains four different experiments (4 different values for $y$). In \cite{MNOSS-2026-connecting}, there are 44 different groupings, where each grouping contains six different experiments. On net, there are 84 groupings, with an average of 815 subjects in each grouping (i.e., completing each of the 84 valuation tasks across the two studies).

There are some differences across the 424 experiments and across the 84 groupings---see Appendix \ref{appsubsec:MNOSS_data} for more details. While most of those details are unimportant for our analysis in this paper, we do highlight one detail that turns out to be important. Across the two studies, subjects are asked to make three kinds of comparisons: (i) a comparison between a certain amount and a binary lottery (which they label an $AB$ decision), (ii) a comparison between a certain amount and a trinary lottery (an $AB'$ decision), and (iii) a comparison between two binary lotteries (a $CD$ decision).\footnote{For certain-vs.-binary comparisons, there are 64 experiments (on average, 403 subjects per experiment) across 14 groupings (on average, 1844 subjects per grouping). For certain-vs.-trinary comparisons, there are 120 experiments (140 subjects per experiment) across 20 groupings (840 subjects per grouping). And for binary-vs.-binary comparisons, there are 240 experiments (107 subjects per experiment) across 50 groupings (516 subjects per grouping).}

\subsection{The Standard Test from the Prior Literature}

Using our collected datasets, we first conduct the standard test used in the prior literature on preference reversals and the prior literature on choices-vs-MPLs---that is, we test the null of $\mP_c(y) = \mP_v(y)$.

In Figure \ref{fig:traditional}, we plot for each experiment the observed values of $\widehat{\mP_c}(y)$ and $\widehat{\mP_v}(y)$ with reference to the 45-degree line of equality. In panel A, we present the data for the prior literature; in panel B, we present the data from \citet{MNOSS-2024-distinguishing,MNOSS-2026-connecting}. Within the prior literature, 
232 of 323 experiments (72\%) exhibit $\widehat{\mP_c}(y) > \widehat{\mP_v}(y)$, and within the \citet{MNOSS-2024-distinguishing,MNOSS-2026-connecting} data 341 of 424 experiments (80\%) exhibit $\widehat{\mP_c}(y) > \widehat{\mP_v}(y)$. Hence, both datasets are consistent with the received wisdom that people are more likely to choose a safer option than valuations would imply. 

\begin{figure}[t!]
\begin{center}
    \includegraphics[width=\textwidth]{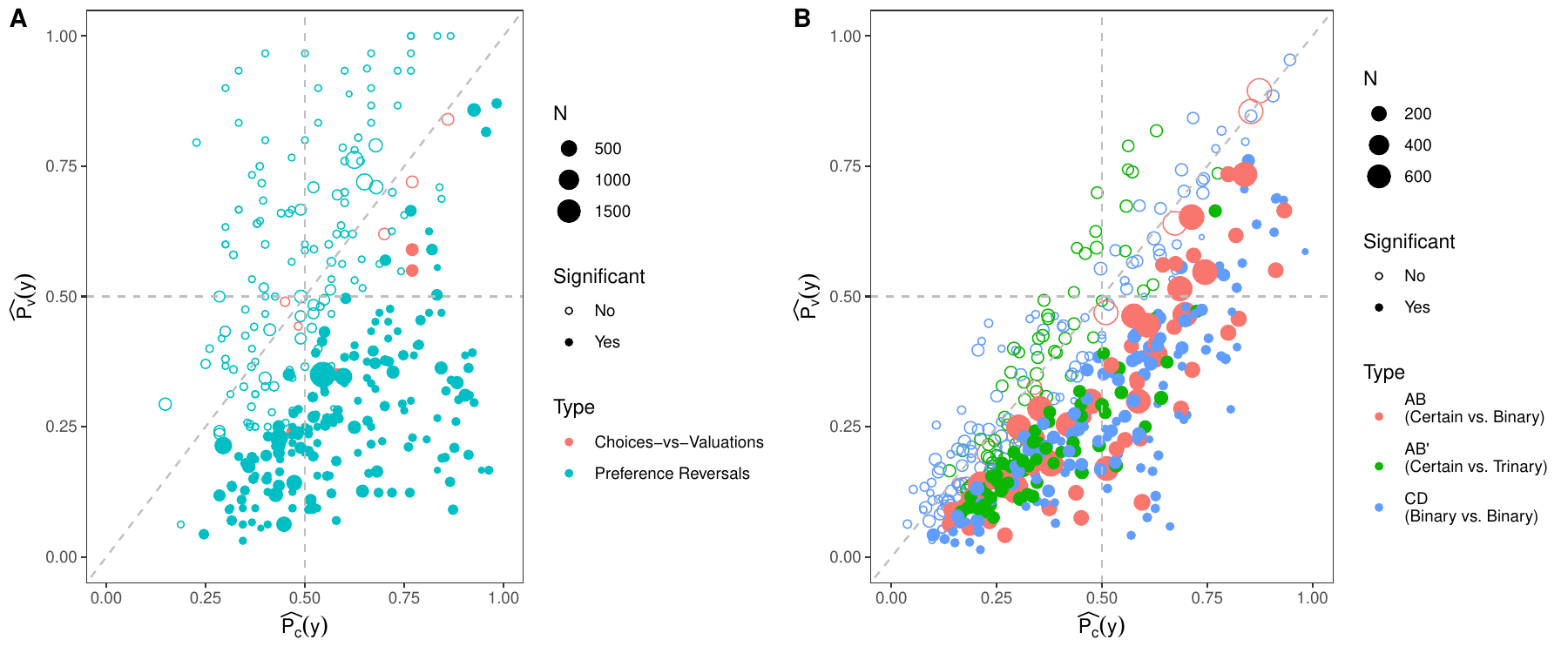}
    \caption{Applying the Standard Test of Stable Preferences}
    \label{fig:traditional}
    \end{center}
     \footnotesize{\emph{Notes:}
     Figure presents observed aggregate choice and valuation proportions, $\widehat{\mP_c}(y)$ and $\widehat{\mP_v}(y)$, for 323 experiments from prior literature (panel A) and 424 experiments from  \citet{MNOSS-2024-distinguishing,MNOSS-2026-connecting} (panel B). Data harmonized such that greater probability corresponds to greater likelihood of preferring safer alternative. Each circle reflects one experiment, size corresponds to experimental sample size.  Shading reflects application of standard test of stable preferences: solid circles correspond to experiments for which we reject the null of $\mP_c(y) = \mP_v(y)$ at the 5\% level in a one-sided test.}
\end{figure} 

Of course, merely observing $\widehat{\mP_c}(y) \neq \widehat{\mP_v}(y)$ is not sufficient to statistically reject the null of $\mP_c(y)=\mP_v(y)$ because even under the null we could observe different proportions due to sampling variation. To assess statistical significance, we use a standard test of proportions under the assumption of independence. Specifically, for each experiment, we calculate the $z$-score 
$$
z = \frac{ \widehat{\mP_c}(y) - \widehat{\mP_v}(y) }{ \sqrt{ \frac{\widehat{\mP_c}(y)(1-\widehat{\mP_c}(y))}{N_c} + \frac{\widehat{\mP_v}(y)(1-\widehat{\mP_v}(y))}{N_v} } },
$$
where $N_c$ and $N_v$ refer to the number of subjects completing choices and valuations, respectively. If we think a two-sided test is appropriate, we reject the null of $\mP_c(y) = \mP_v(y)$ if $|z| > 1.96$ (i.e., we impose a two-sided threshold of 0.05). Under a two-sided test, 67\% of prior experiments and 55\% of \citet{MNOSS-2024-distinguishing,MNOSS-2026-connecting} experiments reject the null hypothesis of $\mP_c(y) = \mP_v(y)$. Given the received wisdom from the prior literature that people are more risk averse for choices than for valuations, we might instead reject the null of $\mP_c(y) = \mP_v(y)$ in favor of $\mP_c(y) > \mP_v(y)$ if $z > 1.645$ (i.e., we impose a one-sided threshold of 0.05). Under this one-sided test, 56\% of prior experiments and 56\% of \citet{MNOSS-2024-distinguishing,MNOSS-2026-connecting} experiments reject the null hypothesis of  $\mP_c(y) = \mP_v(y)$. Figure \ref{fig:traditional} illustrates significance under the latter, one-sided test as shaded points.\footnote{Within the paired data of  \citet{MNOSS-2024-distinguishing,MNOSS-2026-connecting}, we can also evaluate the test of proportions accounting for the correlation between choices and valuations, adjusting the $z$-score. With this adjustment,  61\% of \citet{MNOSS-2024-distinguishing,MNOSS-2026-connecting} experiments exhibit $z> 1.645$, rejecting the null hypothesis of $\mP_c(y) = \mP_v(y)$ in a one-sided test.}  

Hence, if one uses the standard approach of testing the null of $\mP_c(y) = \mP_v(y)$, one indeed finds frequent rejections of the null in both datasets, and support for the received wisdom that people are more risk averse for choices than for valuations. As Table \ref{tab:tests_summary} highlights, however, rejecting the null of $\mP_c(y) = \mP_v(y)$ could reflect a failure of stable preferences, but it also could reflect a failure of identical noise. Because identical noise is a strong assumption to make without supporting evidence, we now turn to our new tests based on weaker assumptions regarding noise.

\subsection{Statistical Tests Based on Weaker Assumptions}\label{subsec:qual_test}

Propositions \ref{prop_big_hexagram}, \ref{prop_weak_pop}, and \ref{prop:symm_prefs}(1) provide the basis for tests of stable preferences based on weaker assumptions, as summarized in Table \ref{tab:tests_summary}. Each test relies on assessing whether observed values for $\widehat{\mP_c}(y)$ and $\widehat{\mP_v}(y)$ could come from some null region of attainable $(\mP_c(y), \mP_v(y))$ combinations derived in each proposition. In Figure \ref{fig:prop1-3analysis}, we again plot observed values of $\widehat{\mP_c}(y)$ and $\widehat{\mP_v}(y)$, but now with reference to the three attainable regions from Propositions \ref{prop_big_hexagram}-\ref{prop:symm_prefs}.

Under our strongest set of assumptions---when noise is symmetric and independent of preferences and preferences are symmetric as well---Proposition \ref{prop:symm_prefs} establishes that the attainable set is $S_X\cup S_y$. Within the prior literature, 45\% of observations lie outside of this region, and within the \citet{MNOSS-2024-distinguishing,MNOSS-2026-connecting}  datasets, 23\% of observations lie outside of this region.

\begin{figure}[t!]
\begin{center}
    \includegraphics[width=\textwidth]{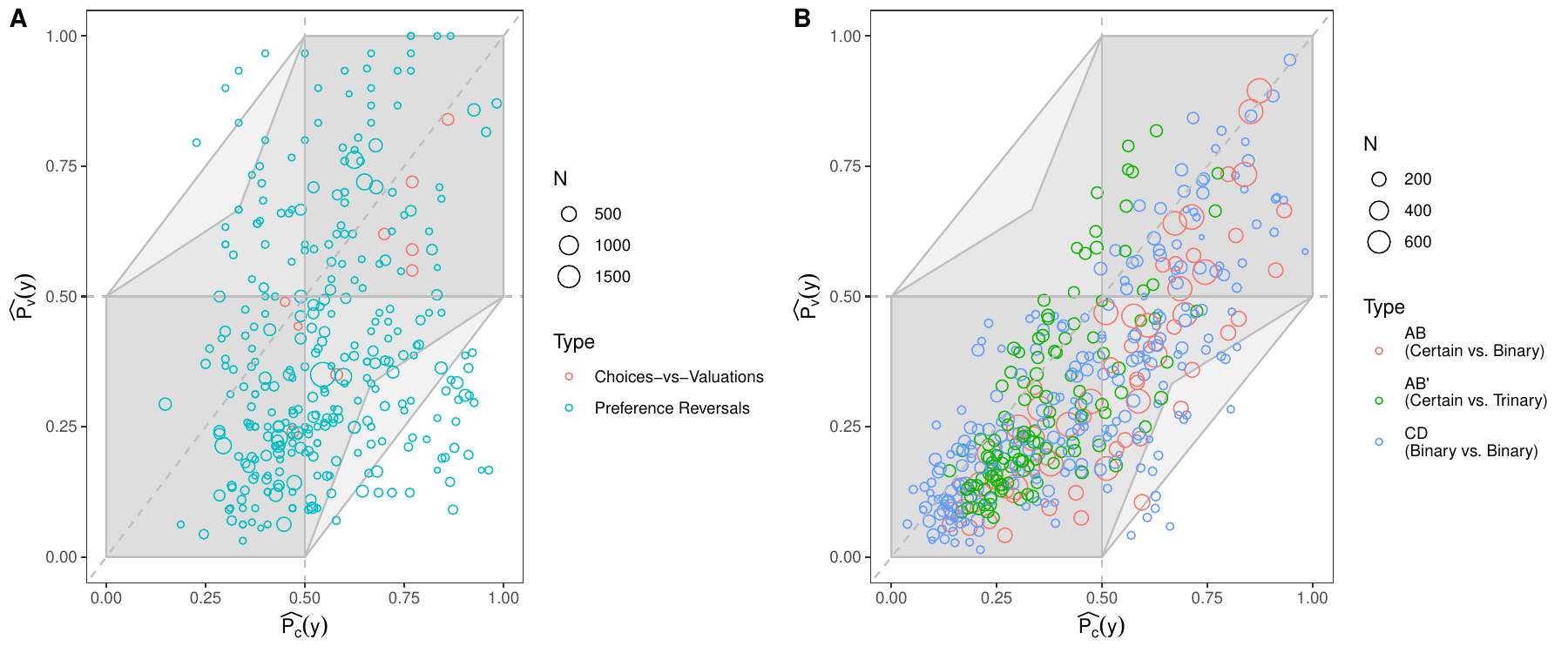}
    \caption{Applying Tests Based on Propositions \ref{prop_big_hexagram}, \ref{prop_weak_pop}, and \ref{prop:symm_prefs}(1)}
    \label{fig:prop1-3analysis}
    \end{center}
     \footnotesize{\emph{Notes:} 
        Figure presents observed aggregate choice and valuation proportions, $\widehat{\mP_c}(y)$ and $\widehat{\mP_v}(y)$, for 323 experiments from prior literature (panel A) and 424 experiments from  \citet{MNOSS-2024-distinguishing,MNOSS-2026-connecting} (panel B). Data harmonized such that greater probability corresponds to greater likelihood of preferring safer alternative. Each circle reflects one experiment, size corresponds to experimental sample size. Gray regions correspond to attainable sets (i.e., null hypotheses to test) under stable preferences and assumptions of Propositions \ref{prop_big_hexagram}, \ref{prop_weak_pop}, and \ref{prop:symm_prefs}(1).}
\end{figure}

To formally assess whether we can reject the null that an observation came from $S_X\cup S_y$, we deploy the test recommended by \citet{miller2025testing} for sign congruence between two estimators.\footnote{This test was originally proposed in the statistics literature by \cite{russek1993qualitative}.} Proposition \ref{prop:symm_prefs}(1) requires $\mP_c(y) - 0.5$ and $\mP_v(y) - 0.5$ to be sign-congruent such that $(\mP_c(y) - 0.5)(\mP_v(y) - 0.5) \ge 0$. The test is based on two relevant $z$-scores:
$$
z_c =  \frac{ \widehat{\mP_c}(y) -  0.5 }{ \sqrt{ \frac{\widehat{\mP_c}(y)(1-\widehat{\mP_c}(y))}{N_c}}}
\qquad \text{ and } \qquad
z_v =  \frac{ \widehat{\mP_v}(y) -  0.5 }{ \sqrt{ \frac{\widehat{\mP_v}(y)(1-\widehat{\mP_v}(y))}{N_v}}}.
$$ 
The test rejects the null hypothesis of sign congruence if  both $(\widehat{\mP_c}(y) - 0.5)(\widehat{\mP_v}(y) - 0.5) <0$ and $\min( |z_c|, |z_v|) > \overline{z}_\alpha$. The threshold $\overline{z}_\alpha$ is set to a critical value $\Phi^{-1}(1-\alpha)$, and we select  $\alpha = 0.05$ such that  $\overline{z}_\alpha = 1.645$. For details see Appendix \ref{app:signcongruence}.\footnote{If $\mP_c(y)$ and $\mP_v(y)$ are known to be highly negatively correlated (i.e., a correlation smaller than around $-0.8$) an adjustment to  $\overline{z}_\alpha$ is required, and when the correlation is $-1$, the threshold value for $\alpha =0.05$ is actually $\overline{z}_\alpha = 1.96$. \citet{russek1993qualitative} Table 2 provides a tabulation of the relevant threshold for $\alpha = 0.05$ depending on the correlation. For testing Proposition \ref{prop:symm_prefs}(1), the issue of negative correlation is largely irrelevant as $\mP_c(y)$ and $\mP_v(y)$ should generally be positively correlated; indeed, within the \citet{MNOSS-2024-distinguishing,MNOSS-2026-connecting} datasets 414 of 424 experiments have a positive correlation and the smallest correlation is $-0.19$.} Deploying this test, we reject the null that an observation came from $S_X\cup S_y$ for only 20\% of experiments from the prior literature and only 12\% of experiments from the \citet{MNOSS-2024-distinguishing,MNOSS-2026-connecting} datasets. In other words, even under our strongest set of assumptions, roughly 80-90\% of experiments across the two datasets are consistent with stable preferences.

Moving to weaker sets of assumptions, even more experiments become consistent with stable preferences. If we only assume that noise is symmetric and independent of preferences, Proposition \ref{prop_weak_pop} establishes that the attainable set is $\con(C_y \cup C_X) \cup \con(V_y \cup V_X)$. Within the prior literature, only 16\% of observations lie outside of this region, and within the \citet{MNOSS-2024-distinguishing,MNOSS-2026-connecting}  datasets, only 4\% of observations lie outside of this region. Deploying a modified version of the sign-congruence test, we reject the null that an observation came from $\con(C_y \cup C_X) \cup \con(V_y \cup V_X)$ for only 9\% of experiments from the prior literature and only 0.5\% of experiments from the \citet{MNOSS-2024-distinguishing,MNOSS-2026-connecting} datasets.\footnote{As suggested by \citet{miller2025testing}, we implement a base change corresponding to the linear restrictions of the boundaries for $\con(C_y \cup C_X) \cup \con(V_y \cup V_X)$, after which we can apply the sign-congruence test. Notably, the base change induces negative correlation that requires adjustment of the threshold value $\overline{z}_\alpha$, which we implement in our analysis. For details see Appendix \ref{app:signcongruence}.}

Finally, under our weakest set of assumptions, where we assume only that the noise is symmetric, Proposition \ref{prop_big_hexagram} establishes that the attainable set is $\con(S_y\cup S_X)$. Only 11\% of experiments from the prior literature and only 1\% of experiments from the \citet{MNOSS-2024-distinguishing,MNOSS-2026-connecting} datasets lie outside of this region; and we can reject the null that an observation came from $\con(S_y\cup S_X)$ for only 3\% of experiments from the prior literature and for no experiments from the \citet{MNOSS-2024-distinguishing,MNOSS-2026-connecting} datasets.\footnote{To develop a statistical test, 
we use the linear restrictions imposed by the lower and upper boundaries of $\con(S_y\cup S_X)$,  $L:\mP_v = \mP_c - 0.5$ and $U:\mP_v = \mP_c + 0.5$. The associated $z$-scores for testing these linear restrictions are
$$
z_L = \frac{ \widehat{\mP_v}(y) - \widehat{\mP_c}(y)  +0.5 }{ \sqrt{ \frac{\widehat{\mP_v}(y)(1-\widehat{\mP_v}(y))}{N_v} + \frac{\widehat{\mP_c}(y)(1-\widehat{\mP_c}(y))}{N_c} }}
\qquad \text{ and } \qquad 
z_U = \frac{ \widehat{\mP_v}(y) - \widehat{\mP_c}(y)  - 0.5 }{ \sqrt{ \frac{\widehat{\mP_v}(y)(1-\widehat{\mP_v}(y))}{N_v} + \frac{\widehat{\mP_c}(y)(1-\widehat{\mP_c}(y))}{N_c} } }.$$
We reject the null that an observation came from $\con(S_y\cup S_X)$ if either (i) $\widehat{\mP_v}(y) < \widehat{\mP_c}(y)$ and $z_L < - 1.645$ or (ii) $\widehat{\mP_v}(y) > \widehat{\mP_c}(y)$ and $z_U > 1.645$.}

\subsection{Using Information on Mean Preferences}\label{sec:prop3_part2}

Because the \citet{MNOSS-2024-distinguishing, MNOSS-2026-connecting} dataset contains within-subject data on both choices and valuations, we can use it to assess our conditional predictions from Propositions \ref{prop:symm_prefs}(2) and \ref{prop:cond_cali}. Formal tests of these predictions are challenging as we have only an estimate, rather than a perfect measure, of the conditioning events $y>\mu$ and $y < \mu$, and thus we pursue a more qualitative analysis. At the same time, these predictions more naturally lend themselves to aggregating observations across experiments.

When noise is symmetric and independent of preferences and preferences are symmetric as well, an implication of Proposition \ref{prop:symm_prefs}(2) is that $\mP_c(y) > \frac{1}{2}$ when $y > \mu$ and $\mP_c(y) < \frac{1}{2}$ when $y < \mu$.\footnote{Additionally, if the distributions of heterogeneity (i.e., of $\eta$) and noise (i.e., of $\varepsilon _c$ and $\varepsilon_v$) are comparable across experiments, then the model would predict a positive relationship between $\widehat{\mP_c}(y)$ and $y-\overline{y_{vi}}$.} To assess this prediction, we need an estimate for $\mu$; we use the average valuation $\overline{y_{vi}}$ calculated for each of the 84 groupings.\footnote{Again, a grouping is a set of experiments that all use the same valuation task for a specific lottery $X$, but which differ in $y$. On average, there are 815 subjects per grouping.} We then analyze how $\widehat{\mP_c}(y)$ varies with $y-\overline{y_{vi}}$, where a more positive (negative) $y-\overline{y_{vi}}$ is a stronger signal that $y>\mu$ ($y<\mu)$.

Panel A of Figure \ref{fig:prop32and5analysis} plots $\widehat{\mP_c}(y)$ by $y-\overline{y_{vi}}$ in a bin-scatter. Consistent with the prediction, there is a tight connection between distance to average valuation and choice proportions. The correlation between binned values is 0.896, and 95\% of binned combinations (38 of 40) lie within the prediction region for Proposition \ref{prop:symm_prefs}(2). In addition, overall average behavior is also consistent with Proposition \ref{prop:symm_prefs}(2): For the 47,748 observations with $y-\overline{y_{vi}}<0$, the average choice proportion is 0.364, whereas for the 20,700 observations with $y-\overline{y_{vi}}>0$, the average choice proportion is 0.645.

\begin{figure}[t!]
\begin{center}
    \includegraphics[width=\textwidth]{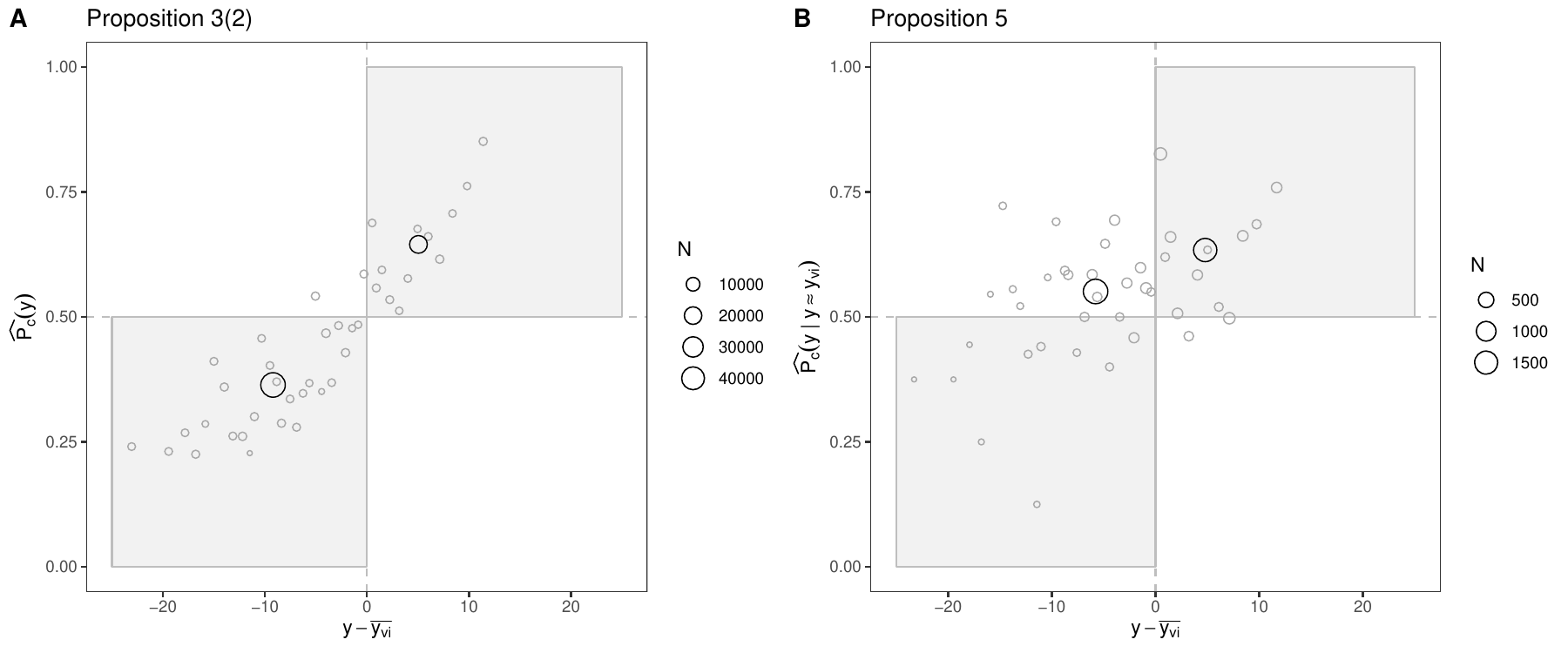}
    \caption{Evaluating Conditional Predictions from Propositions \ref{prop:symm_prefs}(2) and \ref{prop:cond_cali}}
    \label{fig:prop32and5analysis}
    \end{center}
\footnotesize{\emph{Notes:}
Figure presents observed choice proportions $\widehat{\mP_c}(y)$ (panel A) and conditional choice proportions $\widehat{\mP_c}(y| y \approx y_{vi})$ (panel B) against distance to mean valuation ($y-\overline{y_{vi}}$), where light gray regions correspond to predicted behavior. Panel A uses all 68,448 observations from the \cite{MNOSS-2024-distinguishing,MNOSS-2026-connecting} datasets. We first calculate $\overline{y_{vi}}$ for each of the 84 groupings and then calculate $y-\overline{y_{vi}}$ for each of the 68,448 observations. We bin x-axis into 40 approximately equally sized bins of $y-\overline{y_{vi}}$, and for each bin panel A plots the proportion $\widehat{\mP_c}(y)$ choosing safer option against average value of $y-\overline{y_{vi}}$, where circle size corresponds to number of observations in a bin. Panel B presents the same exercise except, after calculating $y-\overline{y_{vi}}$ for all observations, we restrict data to the 3280 observations with $|y-y_{vi}|\le 1$; thus, proportion choosing safer option in each bin is interpreted as $\widehat{\mP_c}(y| y \approx y_{vi})$. In both panels, black circles correspond to overall averages when $y-\overline{y_{vi}} < 0 $ versus when $y-\overline{y_{vi}}>0$.}
\end{figure}

When we make additional assumptions that preferences are unimodal and that the choice and valuation noise are independent and both have positive densities, an implication of Proposition \ref{prop:cond_cali} is that $\mP_c(y|y_{vi}=y) > \frac{1}{2}$ when $y > \mu$ and $\mP_c(y|y_{vi}=y) < \frac{1}{2}$ when $y < \mu$. To assess this prediction, we use $y-\overline{y_{vi}}$ as above, but now we restrict attention to individuals where  $ y_{vi} \approx y$, specifically $|y-y_{vi}| \le 1$ (i.e., whose valuation price list switch points are within 1 row of the choice experimental value of $y$). Of 68,448 individual observations, 3,280 satisfy this criterion. 

Panel B of Figure \ref{fig:prop32and5analysis} plots $\widehat{P_c}( y | y_{vi} \approx y)$ by $y-\overline{y_{vi}}$ in a bin-scatter. Here, the evidence looks inconsistent with stable preferences. While there is a positive correlation of 0.485 between binned values, only 55\% of binned combinations (22 of 40) lie within the prediction region for Proposition \ref{prop:cond_cali}. In terms of overall average behavior, for the the 1532 observations with $y-\overline{y_{vi}}>0$, the average choice proportion of 0.634 is consistent with Proposition \ref{prop:cond_cali}, however, for the 1748 observations with $y-\overline{y_{vi}}<0$, the average choice proportion 0.551 is inconsistent with Proposition \ref{prop:cond_cali}.\footnote{While panel B of Figure \ref{fig:prop32and5analysis} is inconsistent with the predictions of Proposition \ref{prop:cond_cali}, and this inconsistency could be due to a failure of stable preferences, it could also be due to a failure of ancillary assumptions. In fact, Appendix Figure \ref{fig:symmetry} assesses the symmetry assumptions by presenting a histogram of $y_{vi} - \overline{y_{vi}}$ across all 68,448 observations. While the mean value of $y_{vi} - \overline{y_{vi}}$ is 0 by construction, the median value is \$2.77 indicating a failure of symmetry.}

\subsection{Structured Empirical Analysis}\label{sec:MLEestimates}

When one has one has within-subject data on both choices and valuations (as in the \cite{MNOSS-2024-distinguishing,MNOSS-2026-connecting} dataset), and one is willing to make structural assumptions about the distributions of preference heterogeneity and noise, then one can develop a structural test for the null of stable mean preferences and, when stability is rejected, one can estimate the magnitude of deviations from stability. In this section, we implement the two methods described in Section \ref{sec:nnn_analysis} that are built on an assumption that heterogeneity and noise are joint normally distributed.

For each of the 84 groupings in the \citet{MNOSS-2024-distinguishing, MNOSS-2026-connecting} data, we use the probit specification in equation (\ref{eq:probit_condcal}) from Section \ref{subsub:probit_theory} to obtain estimates for $\alpha$, $\beta_1$, and $\beta_2$. The left-hand top panel of Table \ref{tab:probitsummary} summarizes the results by presenting the sample-weighted mean and median of each regression coefficient along with the sample-weighted 25th-75th percentiles. For each of the 84 groupings, we also conduct the MLE approach described in Section \ref{subsub:MLE_theory} (also see \cref{appsec:MLE_details}). This approach yields, for each grouping, estimates of the structural parameters in equation (\ref{eqn:struct_distn}): $\mu_v$, $\mu_c$, $\sqrt{\sigma^2_v+\gamma^2_v}$, $\sqrt{\sigma^2_c+\gamma^2_c}$, and $\gamma_{cv}$. The right-hand top panel of Table \ref{tab:probitsummary} summarizes the results in an analogous way.

\begin{table}[h!]
  \begin{center}
    \caption{Summary of Structural Estimation}
     \label{tab:probitsummary}
     \scalebox{0.9}{
    \begin{tabular}{lccclc}
    \hline\hline
\multicolumn{2}{c}{(1)} & & &  \multicolumn{2}{c}{(2)} \\
\multicolumn{2}{c}{Probit} & & &  \multicolumn{2}{c}{MLE} \\
\cline{1-2} \cline{5-6} \\
\vspace*{.1in}
Constant ($\hat{\alpha}$) & 0.13/0.14 & & & $\hat{\mu}_v$ & 32.04/35.4 \\
         & [-0.26,0.47] & & & &  [28.71,37.77] \\

\addlinespace

$y - \overline{y_{vi}}$ ($\hat{\beta}_1$) & 0.03/0.01 & & & $\hat{\mu}_c$ &  29.44/30.75\\
& [0.00,0.03] & & & & [20.56,41.4]\\
\addlinespace 

$y - y_{vi}$ ($\hat{\beta}_2$) & 0.03/0.03 & & & $\sqrt{ \widehat{\sigma_v^2 + \gamma_v^2} }$& 12.94/13.79 \\
& [0.03,0.04] & & & &   [9.36,15.27]\\
\addlinespace
 & & & & $\sqrt{ \widehat{\sigma_c^2 + \gamma_c^2} }$& 28.32/23.87 \\
&  & & & &  [17.26,37.25]\\
\addlinespace
&& & &$\sqrt{ \widehat{\gamma_{cv}}}$ & 10.99/9.88 \\
&&& & & [8.67,15.7] \\
\cline{1-2} \cline{5-6} \\ 
\addlinespace 
    \# Observations & 68448 & & &  \# Observations&  68448 \\
    \# Groupings & 84 & & &    \# Groupings  & 84 \\

\addlinespace
$\widehat{\mu_v - \mu_c} = \frac{\hat{\alpha}}{\hat{\beta}_1 + \hat{\beta}_2}$ & -0.11/1.97  & & & $\widehat{\mu_v - \mu_c}$  &  -2.2/-1.58\\
& [-4.86,6.44] & & & & [-9.15,5.52]\\
\addlinespace
$(\widehat{\mu_v - \mu_c})^{\text{norm}}$ & 0.47/0.60 & & & $(\widehat{\mu_v - \mu_c})^{\text{norm}}$ &  -0.21/-0.21 \\
&  [-1.76,2.96] & & & &  [-2.23,2.12]\\
    \addlinespace

$\Phi(\hat{\alpha})$  & 0.54/0.56 & & & $\Phi\left( 
\frac{\widehat{\mu_v - \mu_c}}{\hat{V}} \right) $  &0.51/0.48 \\
& [0.40,0.68] & & & & [0.38,0.62] \\ 
\addlinespace
& & & & $\widehat{\frac{\gamma_{cv}}{\sqrt{\widehat{\sigma_v^2  + \gamma_v^2}}\sqrt{\widehat{\sigma_c^2  + \gamma_c^2}} }}$ &  0.36/0.38 \\
&&& & & [0.31,0.42] \\

     \hline \hline 
    \end{tabular}
    }
    \end{center}
\footnotesize{\emph{Notes:} Upper panel presents summary of raw coefficient estimates for 84 separate probit regressions in column (1) and raw parameter estimates for 84 separate maximum likelihood estimations (MLE) in  column (2). For each coefficient or parameter, table presents sample-weighted mean/median along with sample-weighted 25th and 75th percentiles in brackets. Lower panel uses raw estimates to provide estimates of other objects. Estimated preference difference $\widehat{\mu_v - \mu_c}$ calculated as $\frac{\hat{\alpha}}{\hat{\beta}_1 + \hat{\beta}_2}$ for probit and directly from parameter estimates with sign adjustments as appropriate for MLE (see \cref{appsec:MLE_details}); for both, $\widehat{\mu_v - \mu_c}>0$ indicates greater risk aversion in choices. $(\widehat{\mu_v - \mu_c})^{\text{norm}}$ expresses mean preference difference in expected value units: if elicited value occurs with probability $q$, $(\widehat{\mu_v - \mu_c})^{\text{norm}} = q (\widehat{\mu_v - \mu_c})$. Estimated $\mP(\mu_v|y_{vi}=\mu_v)$ calculated as $\Phi(\hat{\alpha})$ in probit and $\Phi\left( 
\frac{\widehat{\mu_v - \mu_c}}{\hat{V}} \right)$ in MLE, where $\Phi$ is standard normal CDF and $\hat{V} = \sqrt{\widehat{\sigma_c^2 + \gamma_c^2} - \frac{\widehat{\gamma_{cv}}^2}{\widehat{\sigma_v^2 + \gamma_v^2}}}$. In MLE, $\frac{\widehat{\gamma_{cv}}}{\sqrt{\widehat{\sigma_v^2  + \gamma_v^2}}\sqrt{\widehat{\sigma_c^2  + \gamma_c^2}} }$ provides lower bound on correlation between choice and valuation preferences.
}
\end{table}

The null of stable preferences requires $\mu_v-\mu_c=0$, which we can test using either approach. For the probit estimation, the structure implies $\mu_v-\mu_c = \frac{\alpha}{\beta_1 + \beta_2}$ (see Section \ref{subsub:probit_theory}); hence, we can transform the estimated probit coefficients into an estimate for $\mu_v - \mu_c$. For the MLE, it is straightforward to transform the estimated $\mu_v$ and $\mu_c$ into an estimate for $\mu_v-\mu_c$. These estimates are summarized in the bottom panel of Table \ref{tab:probitsummary}. Both approaches are well-powered, and we reject the null of $\mu_v-\mu_c=0$ in 65 groupings (77\%) under the probit approach and in 69 groupings (82\%) under the MLE approach. Thus, with the additional structure of joint-normality, we frequently reject the null hypothesis of stability.

That said, when we study the nature of deviations from $\mu_v-\mu_c=0$, the message differs from the received wisdom that choices induce greater risk aversion than valuations. Overall, there appears to be no systematic direction. The mean estimate $\widehat{\mu_v-\mu_c}$ is $-\$0.11$ under probit and $-\$2.20$ under MLE (and normalizing the estimates to expected-value units yields means of $\$0.47$ and $-\$0.21$). The estimates indicate that choices induce greater risk aversion in 54\% of groupings under probit and 48\% under MLE, whereas the estimates indicate that valuations induce greater risk aversion in 46\% of groupings under probit and 52\% under MLE.

In terms of the behavioral impact of these differences, it is useful to quantify $\mP(\mu_v|y_{vi}=\mu_v)$, which is the choice probability in the idealized experiment where we are able to select $y=\mu_v$ and in addition focus on participants who state a valuation equal to that $\mu_v$. Under joint normality and stable preferences, this quantity would be $\frac{1}{2}$. To obtain an estimate for this quantity, we can use $\Phi(\hat{\alpha})$ under probit, and we can use $\Phi\left( \frac{ \widehat{\mu_v - \mu_c}}{\hat{V}} \right)$ under MLE. Estimates for these objects are also summarized in the bottom panel of Table \ref{tab:probitsummary}.

Figure \ref{fig:probitresults} provides a visual depiction of our estimates for both $\mu_v-\mu_c$ and $\mP(\mu_v|y_{vi}=\mu_v)$ for both the probit approach (in panel A) and the MLE approach (in panel B). For each panel, a point reflects one of the 84 groupings, and solid points reflect the groupings for which we reject the null of $\mu_v-\mu_c=0$. The two panels both highlight how we frequently reject the null, and that the magnitude of the deviations can be large in terms of either the magnitude of $\widehat{\mu_v - \mu_c}$ or the deviation of $\widehat{\mP(\mu_v|y_{vi}=\mu_v)}$ from $\frac12$. At the same time, neither panel suggests a systematic overall trend toward more risk aversion in choices (i.e., $\widehat{\mu_v-\mu_c} >0$ and $\widehat{\mP(\mu_v|y_{vi}=\mu_v}) > \frac{1}{2}$) or more risk aversion in valuations (i.e., $\widehat{\mu_v-\mu_c} < 0$ and $\widehat{\mP(\mu_v|y_{vi}=\mu_v)} < \frac{1}{2}$).

\begin{figure}[t!]
\begin{center}
    \includegraphics[width=\textwidth]{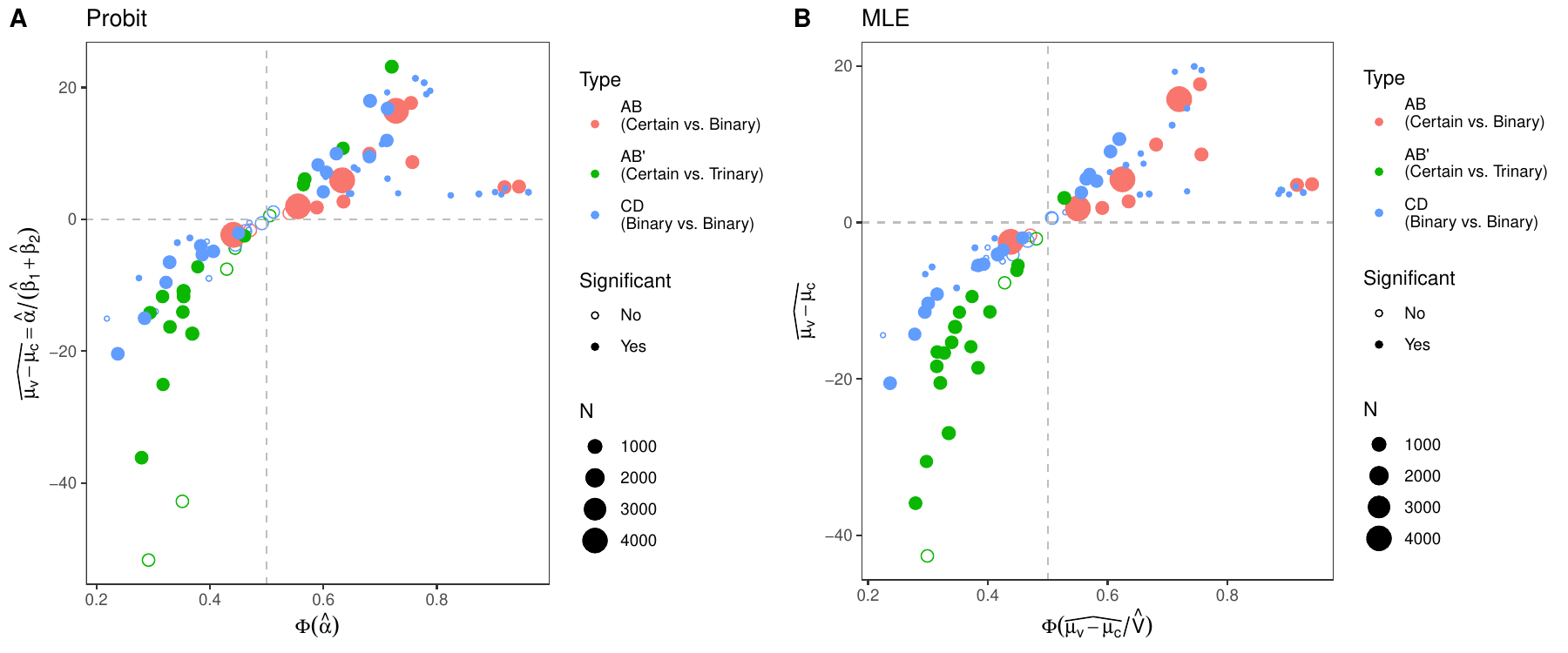}
    \caption{Visual Depiction of Probit and MLE Estimates}
    \label{fig:probitresults}
    \end{center}
     \footnotesize{\emph{Notes:}
     Figure presents estimates of $\widehat{\mu_v - \mu_c}$ (preference difference) against estimates of $\widehat{\mP(\mu_v|y_{vi}=\mu_v)}$ (choice probability in idealized experiment) for each of 84 separate probit estimates in panel A and 84 separate MLE estimates in panel B. In each panel, each circle corresponds to one estimation; circle size corresponds to number of observations in grouping used for that estimation; and filled circles represent groupings for which the null hypothesis of stable mean preferences ($H_0: \mu_v -\mu_c = 0$) is rejected at the 5\% level in a two-sided test. Estimates from the two approaches are strongly correlated: estimates for $\mu_v - \mu_c$ are correlated 0.92; estimates for $\Phi(\cdot)$ are correlated 0.97.}
\end{figure}

Beyond testing the null of $\mu_v-\mu_c=0$, our structural approaches yields further insights. First, because they provide a measure of the magnitudes of deviations, they permit us to study differences across different types of problems. For instance, as described in Section \ref{subsec:data-construct}, in the \citet{MNOSS-2024-distinguishing, MNOSS-2026-connecting} data, subjects make three types of comparisons: (i) a comparison between a certain amount and a binary lottery (which they label an $AB$ decision), (ii) a comparison between a certain amount and a trinary lottery (an $AB'$ decision), and (iii) a comparison between two binary lotteries (a $CD$ decision). Figure \ref{fig:probitresults} uses different colors to denote the different types of comparisons (see also Appendix Table \ref{tab:app_mlsummary} for a summary analogous to Table \ref{tab:probitsummary} for each type of comparison), and a clear pattern emerges: For certain vs. binary comparisons, the estimates indicate that choices tend to induce more risk aversion (i.e., a trend toward $\widehat{\mu_v-\mu_c} >0$ and $\widehat{\mP(\mu_v|y_{vi}=\mu_v)} > \frac{1}{2}$); for certain vs. trinary comparisons, the estimates indicate that valuations tend to induce more risk aversion (i.e., a trend toward $\widehat{\mu_v-\mu_c} <0$ and $\widehat{\mP(\mu_v|y_{vi}=\mu_v)} < \frac{1}{2}$); and for binary vs. binary comparisons, the estimates indicate no systematic deviation.\footnote{\citet{freeman2019eliciting} contrast certain vs. binary and binary vs. binary comparisons. They show significant deviations between $\widehat{\mP_c}(y)$ and $\widehat{\mP_v}(y)$ when considering  $(\$3, 1)$ vs. $(\$4,.8)$  using the traditional test of equality; and they show insignificant deviations when considering $(\$3, .5)$ vs. $(\$4,.4)$. The statistical differences between one safe alternative and one risky alternative as opposed to two risky alternatives resonates with our own findings.}$^\text{,}$\footnote{The \citet{MNOSS-2024-distinguishing, MNOSS-2026-connecting} data also includes variation in the probabilities used in lotteries, and in Appendix \ref{app:parameters} we document some systematic patterns in how $\widehat{\mu_v - \mu_c}$ varies with those probabilities.}

Our MLE approach also provides evidence on the nature of heterogeneity and noise. For instance, our theoretical analysis highlights the importance of whether decision noise has a larger impact on choices versus valuations. Table \ref{tab:probitsummary} provides a summary of estimates for $\sqrt{\gamma^2_c + \sigma^2_c}$ and  $\sqrt{\gamma^2_v + \sigma^2_v}$. In standard deviation units, choices are roughly twice as variable as valuations. Because under the null of stable preferences $\gamma^2_c = \gamma^2_v$, these estimates suggest substantially greater decision noise for choices than for valuations.

Finally, given the structure that we have assumed, another feature of stable preferences is that $y^*_{ci}$ should equal $y^*_{vi}$ at the individual level, and thus $corr(y^*_{ci},y^*_{vi})=1$. Using the MLE estimates, Table \ref{tab:probitsummary} provides a summary of estimates for the quantity $\frac{\gamma_{cv}}{\sqrt{\gamma^2_c + \sigma^2_c}\sqrt{\gamma^2_v + \sigma^2_v}}$, which reflects a lower bound on the correlation in preferences (since $\sigma^2_c \ge 0$ and $\sigma^2_v \ge 0$). Interestingly, these lower-bound correlations are relatively high, with an inter-quartile range of $0.31$ to $0.42$. These correspond favorably to correlations discussed in the psychology literature on stability of personality traits and correlations discussed in the economics literature on test-retest stability in preference measurements over time \citep[for discussion on the relationship between stability in economic preference measurements and psychological personality traits, see][]{meier2015temporal}.

\section{Discussion} \label{sec:conclusion}

Measurement of economic preferences requires confidence in our measurement or elicitation techniques. In this paper, we make both theoretical and empirical contributions to the study of preference stability in the domain of risk preferences across two elicitation techniques: choices versus valuations.

Theoretically, we highlight that the null of stable preferences can be assessed only in conjunction with ancillary assumptions about preference heterogeneity and noise, and how it is important to make these ancillary assumptions explicit. 
We derive a series of predictions under varying ancillary assumptions, motivating new tests of the null of stable preferences in which one compares an observed combination $(\mP_c(y),\mP_v(y))$  to null attainable regions instead of the traditional test of identical measurements.
 We also develop more structured approaches to testing for stability if the analyst is willing to make functional assumptions for preference heterogeneity and noise.
Critically, rejection in any of our tests may reflect a failure of stability or may reflect a failure of ancillary assumptions.

Empirically, we implement tests of the various predictions that we derive so as to obtain an updated sense of the extent to which the evidence is consistent with stable preferences. When we restrict ourselves to tests based on only $\mP_c(y)$ and $\mP_v(y)$, the message that emerges is that the vast majority of experiments from the prior literature and from the recent convenience dataset from \cite{MNOSS-2024-distinguishing, MNOSS-2026-connecting} are consistent with stable preferences. When we instead use tests based on richer empirical objects, we start to see some evidence that is inconsistent with stable preferences (and/or our ancillary assumptions). However, the nature of these deviations is not in line with the received wisdom that choices induce more risk aversion than valuations. In fact, overall, we see no systematic directional pattern, although we do see some systematic variation with problem types that could be explored in future work.

We conclude by discussing some broader messages from our analysis. 
Our analysis highlights the value of focusing on richer empirical objects than traditional choice and valuation proportions. In our domain, if we limit attention to predictions for $\mP_c(y)$ and $\mP_v(y)$, then even under our strongest assumptions of full joint normality, the prediction is merely that $(\mP_c(y),\mP_v(y))$ is in $S_y \cup S_X$ (minus the extreme edges). However, under the same assumptions, if the data contain information on both choices and valuations for the same subjects, we can use this more detailed data to implement our structural approaches. For the data we study---and under the same assumptions of full joint normality---when using data on only $\mP_c(y)$ and $\mP_v(y)$, we find substantial consistency with stable preferences, while when using richer data, we find substantial inconsistency.

 Important considerations emerge from our development for researchers linking models of decision making to data. First, it is critical to pay careful attention to heterogeneity and to make one's assumptions explicit. For instance, it is tempting to conjecture that if the average preference in the population is for $y$ over $X$, then the aggregate population averages for $\mP_c(y)$ and $\mP_v(y)$ ought to lie in $S_y$. Our analysis highlights that this conjecture is not generally correct, and holds only under certain assumptions on the nature of heterogeneity and noise. Second, when assessing whether data is consistent with a theoretical prediction, sampling variation must be accounted for given (typically small) experimental samples. Our results show the importance of developing appropriate statistical tests, as we document sizable differences between the proportion of experiments that are outside our attainable regions and the proportion that are \emph{significantly} outside our attainable regions.

We also provide some guidance to experimentalists deciding which elicitation techniques to use (even when they are not interested per se in the stability of preferences). In the data we study, it turns out that choice data and valuations data seem to be highly correlated, and thus it is not clear that either yields less biased measures of preferences.  At the same time, our structural analysis suggests that, at least under the null of stable preferences, valuations yield more precise measurements than choices. While one might be tempted to conclude, then, that valuations are better to use, we add two caveats. First, we suspect both techniques can be subject to bias due to framing effects, and such effects are not part of our analysis in this paper. Second, and more important, our analysis---especially our structural analysis---highlights the value of collecting both (and more generally of using multiple elicitation techniques). Doing so can provide a richer picture of preferences, and can permit one to assess the existence of inconsistencies, and possibly to correct for the impact of elicitation procedures on behavior.

Finally, while we have framed our theoretical analysis in the domain of risk preferences, it can be applied more broadly. Whenever one is interested in measuring preferences between some generic options A and B, one could do so via binary choices or via valuations, and all of our results apply to any such generic comparison (see the general version of our model in Appendix \ref{ap:general_model}). Moreover, our key intuitions clearly extend to other assessments of procedural invariance. For example, a substantial debate continues on appropriate payment mechanisms in experimental protocols depending on whether subjects ``isolate'' each experimental choice \citep[see, e.g.,][]{holt1986preference, starmer1991does, harrison2014experimental, cox2015paradoxes}. This debate revolves around evidence of inconsistencies between choice proportions depending on whether only a single choice is made or the same choice is embedded in a larger experiment. Our analysis extends naturally to such comparisons and implies that a simple comparison of choice proportions is perhaps not appropriate. The techniques that we develop in this paper could thus prove useful in this and many other domains.

\bibliography{refs}

\begin{thebibliography}{25}
\providecommand{\natexlab}[1]{#1}
\providecommand{\url}[1]{\texttt{#1}}
\expandafter\ifx\csname urlstyle\endcsname\relax
  \providecommand{\doi}[1]{doi: #1}\else
  \providecommand{\doi}{doi: \begingroup \urlstyle{rm}\Url}\fi

\bibitem[Allais(1953)]{allais1953comportement}
M.~Allais.
\newblock Le comportement de l'homme rationnel devant le risque: critique des postulats et axiomes de l'{\'e}cole am{\'e}ricaine.
\newblock \emph{Econometrica}, 21\penalty0 (4):\penalty0 503--546, 1953.

\bibitem[Bouchouicha et~al.(2026)Bouchouicha, Oprea, Vieder, and Wu]{Bouchouicha-et-al-wp-26}
R.~Bouchouicha, R.~Oprea, F.~M. Vieder, and J.~Wu.
\newblock Cognitive frictions and canonical patterns in risk-taking.
\newblock Technical report, Working Paper, 2026.

\bibitem[Brown and Healy(2018)]{brown2018separated}
A.~L. Brown and P.~J. Healy.
\newblock Separated decisions.
\newblock \emph{European Economic Review}, 101:\penalty0 20--34, 2018.

\bibitem[Cox et~al.(2015)Cox, Sadiraj, and Schmidt]{cox2015paradoxes}
J.~C. Cox, V.~Sadiraj, and U.~Schmidt.
\newblock Paradoxes and mechanisms for choice under risk.
\newblock \emph{Experimental Economics}, 18\penalty0 (2):\penalty0 215--250, 2015.

\bibitem[Freeman and Mayraz(2019)]{freeman2019choice}
D.~J. Freeman and G.~Mayraz.
\newblock Why choice lists increase risk taking.
\newblock \emph{Experimental Economics}, 22:\penalty0 131--154, 2019.

\bibitem[Freeman et~al.(2019)Freeman, Halevy, and Kneeland]{freeman2019eliciting}
D.~J. Freeman, Y.~Halevy, and T.~Kneeland.
\newblock Eliciting risk preferences using choice lists.
\newblock \emph{Quantitative Economics}, 10\penalty0 (1):\penalty0 217--237, 2019.

\bibitem[Grether and Plott(1979)]{grether1979economic}
D.~M. Grether and C.~R. Plott.
\newblock Economic theory of choice and the preference reversal phenomenon.
\newblock \emph{American Economic Review}, 69\penalty0 (4):\penalty0 623--638, 1979.

\bibitem[Harrison and Swarthout(2014)]{harrison2014experimental}
G.~W. Harrison and J.~T. Swarthout.
\newblock Experimental payment protocols and the bipolar behaviorist.
\newblock \emph{Theory and Decision}, 77\penalty0 (3):\penalty0 423--438, 2014.

\bibitem[Holt(1986)]{holt1986preference}
C.~A. Holt.
\newblock Preference reversals and the independence axiom.
\newblock \emph{American Economic Review}, 76\penalty0 (3):\penalty0 508--515, 1986.

\bibitem[Khaw et~al.(2021)Khaw, Li, and Woodford]{Khaw-Li-Woodford-21}
M.~W. Khaw, Z.~Li, and M.~Woodford.
\newblock Cognitive imprecision and small-stakes risk aversion.
\newblock \emph{Review of Economic Studies}, 88\penalty0 (4):\penalty0 1979--2013, 2021.

\bibitem[Lichtenstein and Slovic(1971)]{lichtenstein1971reversals}
S.~Lichtenstein and P.~Slovic.
\newblock Reversals of preference between bids and choices in gambling decisions.
\newblock \emph{Journal of Experimental Psychology}, 89\penalty0 (1):\penalty0 46--55, 1971.

\bibitem[Lichtenstein and Slovic(2006)]{Lichtenstein2006construction}
S.~Lichtenstein and P.~Slovic, editors.
\newblock \emph{The Construction of Preference}.
\newblock Cambridge University Press, New York, 2006.

\bibitem[Lu(2026)]{Lu-2026-Meta-Anal}
Y.~Lu.
\newblock Magnitude effects in preference reversals: A meta-analysis.
\newblock \emph{Review of Economic Design}, 30:\penalty0 169--202, 2026.

\bibitem[McGranaghan et~al.(2024)McGranaghan, Nielsen, O’Donoghue, Somerville, and Sprenger]{MNOSS-2024-distinguishing}
C.~McGranaghan, K.~Nielsen, T.~O’Donoghue, J.~Somerville, and C.~D. Sprenger.
\newblock Distinguishing common ratio preferences from common ratio effects using paired valuation tasks.
\newblock \emph{American Economic Review}, 114\penalty0 (2):\penalty0 307--347, 2024.

\bibitem[McGranaghan et~al.(2026)McGranaghan, Nielsen, O’Donoghue, Somerville, and Sprenger]{MNOSS-2026-connecting}
C.~McGranaghan, K.~Nielsen, T.~O’Donoghue, J.~Somerville, and C.~D. Sprenger.
\newblock Connecting common ratio and common consequence preferences.
\newblock \emph{Journal of Political Economy}, forthcoming, 2026.

\bibitem[Meier and Sprenger(2015)]{meier2015temporal}
S.~Meier and C.~D. Sprenger.
\newblock Temporal stability of time preferences.
\newblock \emph{Review of Economics and Statistics}, 97\penalty0 (2):\penalty0 273--286, 2015.

\bibitem[Miller et~al.(2026)Miller, Molinari, and Stoye]{miller2025testing}
D.~L. Miller, F.~Molinari, and J.~Stoye.
\newblock Testing sign congruence between two parameters.
\newblock \emph{Journal of Applied Econometrics}, 41\penalty0 (1):\penalty0 3--11, 2026.

\bibitem[Mosteller and Nogee(1951)]{mosteller1951experimental}
F.~Mosteller and P.~Nogee.
\newblock An experimental measurement of utility.
\newblock \emph{Journal of Political Economy}, 59\penalty0 (5):\penalty0 371--404, 1951.

\bibitem[Russek-Cohen and Simon(1993)]{russek1993qualitative}
E.~Russek-Cohen and R.~M. Simon.
\newblock Qualitative interactions in multifactor studies.
\newblock \emph{Biometrics}, 49\penalty0 (2):\penalty0 467--477, 1993.

\bibitem[Shubatt and Yang(2026)]{Shubatt-Yang-wp-26}
C.~Shubatt and J.~Yang.
\newblock Valuations under tradeoff complexity.
\newblock Technical report, Working Paper, 2026.

\bibitem[Slovic(1995)]{slovic1995construction}
P.~Slovic.
\newblock The construction of preference.
\newblock \emph{American Psychologist}, 50\penalty0 (5):\penalty0 364--371, 1995.

\bibitem[Starmer and Sugden(1991)]{starmer1991does}
C.~Starmer and R.~Sugden.
\newblock Does the random-lottery incentive system elicit true preferences? {A}n experimental investigation.
\newblock \emph{American Economic Review}, 81\penalty0 (4):\penalty0 971--978, 1991.

\bibitem[Thurstone(1931)]{thurstone1931indifference}
L.~L. Thurstone.
\newblock The indifference function.
\newblock \emph{Journal of Social Psychology}, 2\penalty0 (2):\penalty0 139--167, 1931.

\bibitem[Tversky and Kahneman(1992)]{TverskyKahneman1992}
A.~Tversky and D.~Kahneman.
\newblock Advances in prospect theory: Cumulative representation of uncertainty.
\newblock \emph{Journal of Risk and Uncertainty}, 5\penalty0 (4):\penalty0 297--323, 1992.

\bibitem[Tversky and Thaler(1990)]{tversky1990anomalies}
A.~Tversky and R.~H. Thaler.
\newblock Anomalies: preference reversals.
\newblock \emph{Journal of Economic Perspectives}, 4\penalty0 (2):\penalty0 201--211, 1990.

\end{thebibliography}

\newpage
\clearpage

\appendix

\renewcommand{\thefigure}{\thesection.\arabic{figure}}
\renewcommand{\thetable}{\thesection.\arabic{table}}
\renewcommand{\theequation}{\thesection.\arabic{equation}}
\setcounter{figure}{0}
\setcounter{table}{0}
\setcounter{page}{1}

\begin{center}
    {\Large Measuring Economic Preferences in the Presence of Noise: \\
    The Connections Between Choices and Valuations} \\[0.5em]
    {Ted O'Donoghue, Charles D. Sprenger, Po Hyun Sung, and Ben Wincelberg}
    \\[0.5em]
    August 30, 2026 \\[1em]
    {\Large Supplementary Appendix}
\end{center}

\bigskip

\bigskip

\section{Additional Tables and  Figures}\label{app:additional_figures}

\begin{figure}[h!]
\begin{center}
    \includegraphics[scale=0.8]{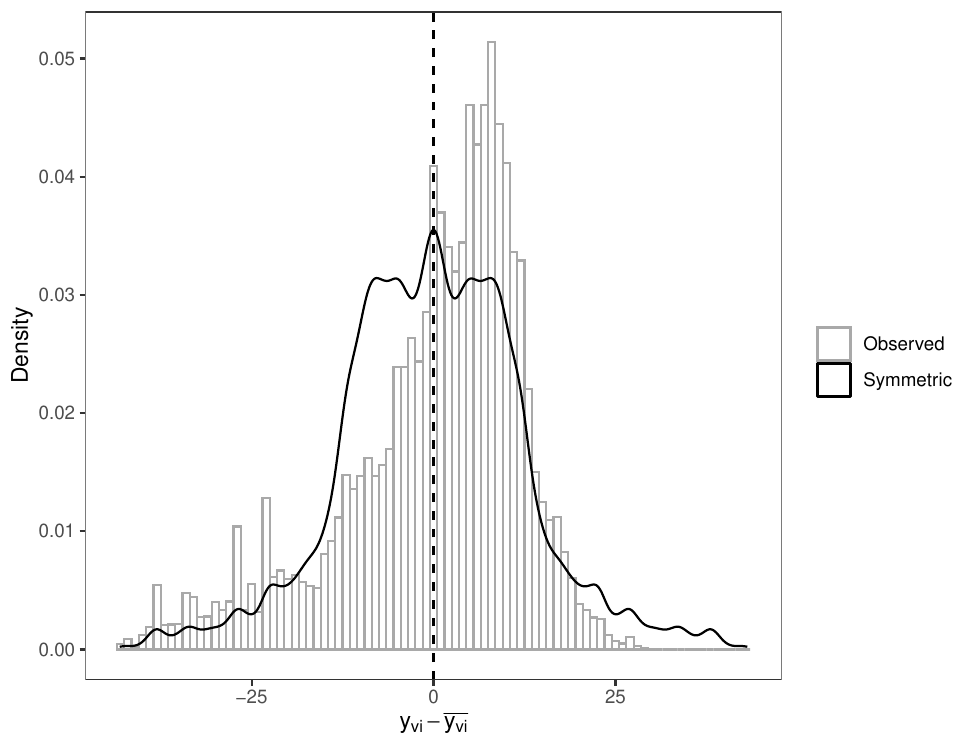}
    \caption{Evaluating Symmetry of $y_{vi} - \overline{y_{vi}}$}
    \label{fig:symmetry}
    \end{center}
     \footnotesize{\emph{Notes:} 
     Figure presents histogram of $y_{vi} - \overline{y_{vi}}$ for 68,448 observations in \citet{MNOSS-2024-distinguishing,MNOSS-2026-connecting} (gray bars). Solid black line corresponds to symmetrized distribution constructed from a combined dataset of $y_{vi} - \overline{y_{vi}}$ and $\overline{y_{vi}} - y_{vi}$. Median of $y_{vi} - \overline{y_{vi}}$ is 2.77.
     }
\end{figure}

\begin{table}[h!]
  \begin{center}
    \caption{Summary of Structural Estimation Separated by Problem Type}
     \label{tab:app_mlsummary}
     \scalebox{0.65}{
    \begin{tabular}{lccccclccc}
    \hline\hline
       Sample:& $AB$ & $AB'$  & $CD$ &&&  Sample:& $AB$ & $AB'$  & $CD$\\
   & Certain vs. & Certain vs. & Binary vs. &&&&  Certain vs. & Certain vs. & Binary vs. \\
      & Binary & Trinary & Binary &&&&  Binary & Trinary & Binary \\
& (1) & (2) & (3) &&&& (4) &(5) &(6) \\
\multicolumn{4}{c}{Probit} & & &  \multicolumn{4}{c}{MLE} \\
\cline{1-4} \cline{7-10} \\
\vspace*{.1in}
Constant ($\hat{\alpha}$) &0.35/0.34 & -0.22/-0.38 & 0.12/0.03 & & & $\hat{\mu}_v$ &  33.28/37.67 & 32.97/33.98 & 30.21/33.9 \\
         & [0.14,0.61] & [-0.44,-0.1] & [-0.29,0.47] & & & & [30.36,39.35] & [28.71,37.79] & [28.95,35.64] \\

\addlinespace

$y - \overline{y_{vi}}$ ($\hat{\beta}_1$) &0.04/0.03 & 0/-0.01 & 0.04/0.02 & & & $\hat{\mu}_c$ &  37.49/41.69 & 16.83/20 & 29.62/30.65\\
&  [0.01,0.04] & [-0.01,0] & [0.01,0.03] & & & & [36.88,44.71] & [6.89,24.45] & [25.3,37.93]\\
\addlinespace 

$y - y_{vi}$ ($\hat{\beta}_2$) & 0.04/0.04 & 0.03/0.03 & 0.03/0.03 & & & $\sqrt{ \widehat{ \sigma_v^2 + \gamma_v^2 }}$&  12.14/13.53 & 15.4/15.29 & 12.14/12.84 \\
&  [0.03,0.05] & [0.03,0.03] & [0.02,0.03]& & & &  [8.5,14.64] & [14.93,15.99] & [8.22,14.94]\\
\addlinespace
 & & &&& & $\sqrt{ \widehat{\sigma_c^2 + \gamma_c^2} }$& 18.35/17.26 & 50.74/47.89 & 23.69/22.23 \\
&  & & &&& &  [15.72,22.98] & [41.93,62.44] & [19.6,29.47\\
\addlinespace
&& &&& &$\sqrt{ \widehat{\gamma_{cv} }}$ & 8.92/9.58 & 17.59/17.41 & 8.75/9.16 \\
&&& &&& & [7.88,12.02] & [16.26,19.17] & [6.39,11.62] \\
\cline{1-4} \cline{7-10} \\ 
\addlinespace 
    \# Observations &25816 & 16816 & 25816 & & &  \# Observations&  25816 & 16816 & 25816\\
    \# Groupings & 14 & 20 & 50 & & &    \# Groupings  & 14 & 20 & 50 \\
    \addlinespace

$\widehat{\mu_v - \mu_c} = \frac{\hat{\alpha}}{\hat{\beta}_1+ \hat{\beta}_2}$ & 5.51/4.93 & -11.46/-11.7 & 1.68/1.12  & & & $\widehat{\mu_v - \mu_c}$  & 5.28/4.85 & -16.14/-15.87 & -0.59/-1.97\\
&[1.8,8.69] & [-17.34,-2.52] & [-4.52,8.28] & & & &  [1.8,8.69] & [-20.5,-9.46] & [-5.49,5.29]\\
\addlinespace
$(\widehat{\mu_v - \mu_c})^{\text{norm}}$ & 2.14/2.83 & -1.63/-2.25 & 0.17/0.03& & & $(\widehat{\mu_v - \mu_c})^{\text{norm}}$ &   2.03/2.73 & -3.27/-2.69 & -0.45/-0.21 \\
&   [1,4.35] & [-2.95,-0.23] & [-0.83,1.66] & & & &   [1.12,4.35] & [-4.58,-1.28] & [-1.41,0.92]\\
\addlinespace
$\Phi(\hat{\alpha})$  &  0.63/0.63 & 0.42/0.35 & 0.54/0.51 & & & $\Phi\left( \frac{\widehat{\mu_v - \mu_c}}{\hat{V}} \right) $  &0.62/0.62 & 0.37/0.35 & 0.5/0.47 \\
& [0.56,0.73] & [0.33,0.46] & [0.39,0.68] & & & & [0.55,0.72] & [0.32,0.4] & [0.38,0.6] \\ 

\addlinespace
& & &&& & $\frac{\widehat{\gamma_{cv}}}{\sqrt{\widehat{\sigma_v^2  + \gamma_v^2}}\sqrt{\widehat{\sigma_c^2  + \gamma_c^2}} }$ &   0.4/0.4 & 0.41/0.41 & 0.3/0.3 \\
&&&&& & &  [0.35,0.45] & [0.38,0.43] & [0.27,0.35] \\

     \hline \hline 
    \end{tabular}
    }
    \end{center}
\footnotesize{\emph{Notes:} 
Upper panel presents summary of raw coefficient estimates for 84 separate probit regressions in columns (1)-(3) and raw parameter estimates for 84 separate maximum likelihood estimations (MLE) in  column (4)-(6). For each coefficient or parameter, table presents sample-weighted mean/median along with sample-weighted 25th and 75th percentiles in brackets. Lower panel uses raw estimates to provide estimates of other objects. Estimated preference difference $\widehat{\mu_v - \mu_c}$ calculated as $\frac{\hat{\alpha}}{\hat{\beta}_1 + \hat{\beta}_2}$ for probit and directly from parameter estimates with sign adjustments as appropriate for MLE (see \cref{appsec:MLE_details}); for both, $\widehat{\mu_v - \mu_c}>0$ indicates greater risk aversion in choices. $(\widehat{\mu_v - \mu_c})^{\text{norm}}$ expresses mean preference difference in expected value units: if elicited value occurs with probability $q$, $(\widehat{\mu_v - \mu_c})^{\text{norm}} = q (\widehat{\mu_v - \mu_c})$. Estimated $\mP(\mu_v|y_{vi}=\mu_v)$ calculated as $\Phi(\hat{\alpha})$ in probit and $\Phi\left( \frac{\widehat{\mu_v - \mu_c}}{\hat{V}} \right) $ in MLE, where $\Phi$ is standard normal CDF and $\hat{V} = \sqrt{\widehat{\sigma_c^2 + \gamma_c^2} - \frac{\widehat{\gamma_{cv}}^2}{\widehat{\sigma_v^2 + \gamma_v^2}}}$. In MLE, $\frac{\widehat{\gamma_{cv}}}{\sqrt{\widehat{\sigma_v^2  + \gamma_v^2}}\sqrt{\widehat{\sigma_c^2  + \gamma_c^2}} }$ provides lower bound on correlation between choice and valuation preferences.}
\end{table}

\section{Proofs}\label{ap:proof}

\setcounter{figure}{0}
\setcounter{table}{0}
\setcounter{equation}{0}

\subsection{Proof of \cref{lm_weak_ind}}\label{ap:lm_weak_ind}

\begin{proof}[Proof of \cref{lm_weak_ind}]
Recall that $S_y=[\frac{1}{2},1]^2$ and $S_X=[0,\frac{1}{2}]^2$. To establish $\mathcal{A}_i = S_y \cup S_X$, we first prove that $\cA_i \subseteq S_y \cup S_X$, and then prove that $S_y \cup S_X \subseteq \cA_i$.

The former is straightforward. $F_c$ and $F_v$ both symmetric and mean zero implies: (i) if $y>y^*_i$ then $\mP_{ci}(y) = F_c(y-y^*_i) \ge \frac{1}{2}$ and $\mP_{vi}(y) = F_v(y-y^*_i) \ge \frac{1}{2}$, and thus $(\mP_{ci}(y),\mP_{vi}(y)) \in S_y$; (ii) if $y<y^*_i$ then $\mP_{ci}(y) = F_c(y-y^*_i) \le \frac{1}{2}$ and $\mP_{vi}(y) = F_v(y-y^*_i) \le \frac{1}{2}$, and thus $(\mP_{ci}(y),\mP_{vi}(y)) \in S_X$; (iii) if $y=y^*_i$ then $\mP_{ci}(y) = F_c(0) = \frac12$ and $\mP_{vi}(y) = F_v(0) = \frac12$, and thus $(\mP_{ci}(y),\mP_{vi}(y)) =\left( \frac12,\frac12 \right)$. Note that parts (2) and (3) follow from these statements, and it further follows that if $(a,b) \in \cA_i$, then $(a,b) \in S_y \cup S_X$.\footnote{Note that the same conclusion holds if we relax the symmetry assumption, requiring only that $\varepsilon_c$	and $\varepsilon_v$ have a median of zero.}

We prove $S_y \cup S_X \subseteq \cA_i$ by construction; many are possible, we use a construction that follows the narrative in the text. First, suppose $F_c$ is the CDF of $N(0,\sigma^2_c)$ and $F_v$ is the CDF of $N(0,\sigma^2_v)$. Note that, for any $a \in (0,\frac{1}{2})$ and $z<0$, there exists a $\sigma^2_c$ such that $F_c(z)=a$; and for any $a \in (\frac{1}{2},1)$ and $z>0$, there exists a $\sigma^2_c$ such that $F_c(z)=a$. Analogous properties hold for $F_v$. Hence, for any $(a,b)$ in the interior of $S_X$, if we select a $y_i^*>y$, there exists a $\sigma^2_c$ such that $F_c(y-y^*_i)=a$ and a $\sigma^2_v$ such that $F_v(y-y^*_i)=b$; thus any $(a,b)$ in the interior of $S_X$ is also in $\cA_i$. Analogously, for any $(a,b)$ in the interior of $S_y$, if we select a $y_i^*<y$, there exists a $\sigma^2_c$ such that $F_c(y-y^*_i)=a$ and a $\sigma^2_v$ such that $F_v(y-y^*_i)=b$; thus any $(a,b)$ in the interior of $S_y$ is also in $\cA_i$.

It remains to prove that the edges of $S_X$ and $S_y$ are in $\cA_i$. Suppose that for each $k \in \{c,v\}$, there exists $x''_k > x'_k \ge 0$ such that the $\varepsilon_k$ is uniformly distributed on $[-x''_k,-x'_k]\cup [x'_k,x''_k]$. Note that $\varepsilon_k$ is symmetric about zero with CDF
\begin{equation}\label{eqn:L1proof}
F_k(x)=\begin{cases}
    0 & x \le -x''_k \\
    \frac{1}{2(x''_k-x'_k)} (x - (-x''_k)) & -x''_k \le x \le -x'_k \\
    \frac{1}{2} & -x'_k \le x \le x'_k \\
    \frac{1}{2} + \frac{1}{2(x''_k-x'_k)} (x - x'_k) & x'_k \le x \le x''_k \\
    1 & x \ge x''_k.
\end{cases} 
\end{equation}

Consider first the case where $x'_c \ge x''_v$. If we set $y^*_i$ such that $y-y^*_i \in [-x''_v,x''_v]$, we can select all $(a,b)$ on the edge with $a=\frac{1}{2}$ and $b \in [0,1]$. If we set $y^*_i$ such that $y-y^*_i \in [x'_c,x''_c]$, we can select all $(a,b)$ on the edge with $a\in[\frac{1}{2},1]$ and $b=1$. And if we set $y^*_i$ such that $y-y^*_i \in [-x''_c,-x'_c]$, we can select all $(a,b)$ on the edge with $a\in[0,\frac{1}{2}]$ and $b=0$. If we instead consider the case where $x'_v \ge x''_c$, we can attain all $(a,b)$ on the remaining three edges, completing the proof.\footnote{To parallel the text, we used a normal-normal construction to prove that all points in the interior of $S_y \cup S_X$ are attainable; however, we could have also used the construction in \eqref{eqn:L1proof} for that as well. Again, many different constructions are possible.}

\end{proof}

\subsection{Proof of \cref{prop_big_hexagram}}\label{ap:prop_big_hexagram}

\begin{proof}[Proof of \cref{prop_big_hexagram}]

$\cA \subseteq \con(S_y \cup S_X)$ follows directly from Lemma \ref{lm_weak_ind} because the aggregate probabilities are expectations of individual probabilities lying in $S_y \cup S_X$. We prove $\con(S_y \cup S_X) \subseteq \cA$ by construction. Suppose a population consists of two groups: Group $1$ has $\eta_1$ such that $y^*_1<y$; Group $2$ has $\eta_2$ such that $y^*_2>y$; and the population shares are $\pi_1$ and $\pi_2$ such that $\pi_1+\pi_2=1$. By Lemma \ref{lm_weak_ind}, because Group 1 has $y-y^*_1>0$, and because there are no restrictions on their noise beyond symmetry, their decision probabilities could be any $(a,b) \in S_y$. Analogously, because Group 2 has $y-y^*_1<0$, and because there are also no restrictions on their noise beyond symmetry---including no requirement that their noise be the same as that for Group 1---their decision probabilities could be any $(a,b) \in S_X$. Finally, because there are no restrictions on the population proportions $\pi_1$ and $\pi_2$, it follows that the aggregate decision probabilities can be any mixture of any point in $S_y$ with any point in $S_X$, and thus $\con(S_y \cup S_X) \subseteq \cA$.

\end{proof}

\subsection{Proof of \cref{prop_weak_pop}}\label{ap:prop_weak_pop}

\begin{proof}[Proof of \cref{prop_weak_pop}]

\noindent Define $M \equiv \con(C_y \cup C_X) \cup \con(V_y \cup V_X)$, $\hat M = M - \left( \frac12,\frac12 \right)$, and $\hat \cA = \mathcal{A} - \left( \frac12,\frac12 \right)$, so the goal is to prove that $\hat \cA = \hat M$. Note that 
\begin{equation}\label{eq:translate}
 \hat M
=
\left\{(a,b)\in\left[-\frac12,\frac12\right]^2 :
|b-2a|\leq\frac12\right\}
\cup
\left\{(a,b)\in\left[-\frac12,\frac12\right]^2:
|2b-a|\leq\frac12\right\}.   
\end{equation}

Since preferences are independent of noise, $(a,b) \in \hat \cA$ if and only if there exist $F_c$ and $F_v$ that are symmetric such that $(a,b) \in \con(K)$, where
$$
K \equiv \left\{
\left(F_c(x)-\frac12,F_v(x)-\frac12\right) \; \middle\vert{} \; x\in \R
\right\} \subseteq \left[-\frac12,\frac12\right]^2.
$$

\bigskip

We first prove $\cA \subseteq M$, or equivalently, $\hat \cA \subseteq \hat M$. To this end, we fix $F_c$ and $F_v$ and show that $\con(K) \subseteq \hat M$. Since $F_c$ and $F_v$ are CDFs, $K$ is weakly increasing, i.e., for any two points $(x_1,y_1),(x_2,y_2) \in K$, either $(x_1,y_1)\ge (x_2,y_2)$ or $(x_1,y_1)\le (x_2,y_2)$. Moreover, since $F_c$ and $F_v$ are symmetric, $K=-K$, i.e., if $(x,y)\in K$, then $(-x,-y) \in K$. Thus, every point of $K$ lies in one of the two diagonal orthants, since otherwise a point and its negative would violate monotonicity. Hence, if we define $K^+ \equiv K\cap \left[0,\frac12\right]^2$, then $K=K^+\cup (-K^+)$.

By Fenchel's refinement of Carath\'eodory's theorem for connected sets,\footnote{Note that $K$ is connected since $F_c$ and $F_v$ are continuous. The theorem states that for a connected set $A \subseteq \R^n$, any point in $\con(A)$ is a convex combination of at most $n$ points in $A$. \cref{prop_weak_pop} holds without continuity of $F_c$ and $F_v$ with a slightly longer proof.} if $(a,b) \in \con(K)$, then $(a,b)$ is a convex combination of two points in $K$. Suppose, toward a contradiction, that $(a,b) \notin \hat M$. Since $a\cdot b <0$, $(a,b)$ must be a convex combination of a point in $K^+$ and a point in $K^-$. Since $K=-K$, 
 there exist $(x_1,y_1),(x_2,y_2) \in K^+$ such that 
\begin{align*}
    (a,b)&=\tau^+(x_1,y_1)+\tau^-(-x_2,-y_2)\\
    &=(\tau^+x_1-\tau^-x_2,\tau^+y_1-\tau^-y_2)
\end{align*}
for some $\tau^+,\tau^-\in [0,1]$ and $\tau^++\tau^-=1$.

Then by \eqref{eq:translate}, 
\begin{equation}\label{eq:violations}
    |b-2a|>\frac12
\qquad\text{and}\qquad
|2b-a|>\frac12.
\end{equation}
We consider first the case that $b-2a>\frac 12$ and $2b-a > \frac12$. Define
\begin{align*}
\Psi_C(x_1,y_1,x_2,y_2)
&:=
\tau^+(y_1-2x_1)-\tau^-(y_2-2x_2),\\
\Psi_V(x_1,y_1,x_2,y_2)
&:=
\tau^+(2y_1-x_1)-\tau^-(2y_2-x_2),
\end{align*}
so that 
\[
\Psi_C(x_1,y_1,x_2,y_2)=b-2a,
\qquad
\Psi_V(x_1,y_1,x_2,y_2)=2b-a.
\]

Suppose first that $\tau^+\geq\tau^-$, and note that either $(x_1,y_1)\le(x_2,y_2)$ or $(x_1,y_1)\ge(x_2,y_2)$. If $(x_1,y_1)\le(x_2,y_2)$ and thus $y_1 \le y_2$, then
\[
\Psi_C(x_1,y_1,x_2,y_2)
\leq
(\tau^+-\tau^-)y_1-2\tau^+x_1+2\tau^-x_2
\leq
\frac{\tau^+-\tau^-}{2}+\tau^-
=\frac12,\]
where the second inequality follows because $(x_1,y_1),(x_2,y_2) \in K^+$ implies boundary conditions $y_1 \le \frac12$, $x_1 \ge 0$, $x_2 \le \frac12$.
Similarly, if $(x_1,y_1)\ge(x_2,y_2)$ and thus $x_2 \le x_1$, then
\[
\Psi_C(x_1,y_1,x_2,y_2)
\leq
\tau^+y_1-\tau^-y_2-2(\tau^+-\tau^-)x_2
\leq \frac{\tau^+}{2}
\leq\frac12,
\]
applying the boundary conditions $y_1\le\frac12$, $y_2 \ge 0$, and $x_2\ge0$. Thus, if $\tau^+\ge\tau^-$ then $\Psi_C(x_1,y_1,x_2,y_2) \leq \frac12$, contradicting \eqref{eq:violations}. An analogous argument implies that $\Psi_V(x_1,y_1,x_2,y_2) \le \frac12$ when $\tau^+\le\tau^-$, also contradicting \eqref{eq:violations}. The case of $b-2a<-\frac 12$ and $2b-a < -\frac12$ follows a symmetric argument. The other two cases lie outside $\left[-\frac12,\frac12\right]^2$. Therefore $\con(K) \subseteq \hat M$. 

\medskip

We next show $M \subseteq \cA$. We prove this by construction. Suppose $F_c$ and $F_v$ are as in  \eqref{eqn:L1proof} from the proof of Lemma \ref{lm_weak_ind}, which again satisfies \cref{as_even_c_v}. Applying the logic there, when $x'_c \ge x''_v$, the resulting curve will include both the edge where $b=0$ and $a \in [0,\frac{1}{2}]$ and the edge where $b=1$ and $a \in [\frac{1}{2},1]$. Because we can choose any population distribution we want over this curve, we can attain any point in the convex hull of these two line segments, which is just $\con(C_y \cup C_X)$. Analogously, when $x'_v \ge x''_c$, the resulting curve will include both 
the edge where $a=0$ and $b \in [0,\frac{1}{2}]$ and the edge where $a=1$ and $b \in [\frac{1}{2},1]$, from which we can attain any point in $\con(V_y \cup V_X)$.

\end{proof}

\subsection{Proof of \cref{prop:symm_prefs}}\label{ap:symm_prefs}

\begin{proof}[Proof of \cref{prop:symm_prefs}]

From the text, the aggregate choice proportions are:
$$
\mP_c(y) = G_c(y-\mu) \qquad \text{and} \qquad \mP_v(y) = G_v(y-\mu)
$$
where $G_c$ is the CDF for random variable $\nu_{c}=\eta + \varepsilon_{c}$ and $G_v$ is the CDF for random variable $\nu_{v}=\eta + \varepsilon_{v}$. Since $\eta$, $\varepsilon_{c}$, and $\varepsilon_{v}$ are symmetric about zero and $\eta \perp (\varepsilon_c,\varepsilon_v)$, $\nu_{c}$ and $\nu_{v}$ are both symmetric about zero.

Given that $\nu_{c}$ and $\nu_{v}$ are both symmetric about zero, the proof that $\cA \subseteq S_y \cup S_X$ is exactly analogous to the proof that $\cA_i \subseteq S_y \cup S_X$ for Lemma \ref{lm_weak_ind}. In addition, letting $\eta=0$ degenerately, $S_y \cup S_X \subseteq \cA$ follows from Lemma \ref{lm_weak_ind}.

\end{proof}

\subsection{Proof of \cref{prop:cali}}\label{ap:cali}

\begin{proof}[Proof of \cref{prop:cali}]
As in the proof of Proposition \ref{prop:symm_prefs}, the aggregate choice proportions are:
$$
\mP_c(y) = G_c(y-\mu) \qquad \text{and} \qquad \mP_v(y) = G_v(y-\mu)
$$
where $G_c$ is the CDF for random variable $\nu_{c}=\eta + \varepsilon_{c}$ and $G_v$ is the CDF for random variable $\nu_{v}=\eta + \varepsilon_{v}$. Since $\eta$, $\varepsilon_{c}$, and $\varepsilon_{v}$ are symmetric about zero and $\eta \perp (\varepsilon_c,\varepsilon_v)$, $\nu_{c}$ and $\nu_{v}$ are both symmetric about zero.

When $y=\mu$, $\mP_c(y)=G_c(0),\mP_v(y)=G_v(0)$, which equal $\frac{1}{2}$ since $\nu_c$ and $\nu_v$ are symmetric and continuous. Letting $t=y-\mu$, it remains to show $G_c(t) \neq G_v(t)$ for all $t \neq 0$.

Because $\eta$ has a density by Assumption \ref{as_unimodal} and $\eta \perp (\varepsilon_c,\varepsilon_v)$, we have 
\begin{align*}
    G_v(t) - G_c(t) =\int_{-\infty}^{\infty}(F_v(t-\eta)-F_c(t-\eta))h(\eta)\dd \eta. 
\end{align*}
Setting $s = t-\eta$, we have
\begin{align*}
    G_v(t) - G_c(t) =\int_{-\infty}^{\infty}(F_v(s)-F_c(s))h(t-s)\dd s. 
\end{align*}
Define $D(s) \equiv F_v(s)-F_c(s)$, and note that, because $F_v$ and $F_c$ are both symmetric, $D$ is odd, i.e., $D(s)=-D(-s)$ for all $s$. We now have 
\begin{align*}
    G_v(t) - G_c(t) =\int_{-\infty}^{\infty}D(s)h(t-s)\dd s = \int_{-\infty}^{0}D(s)h(t-s)\dd s + \int_{0}^{\infty}D(s)h(t-s)\dd s. 
\end{align*}
We can rewrite the first integral as
\begin{align*}
    \int_{-\infty}^{0}D(s)h(t-s)\dd s = \int_{0}^{\infty}D(-s)h(t+s)\dd s = \int_{0}^{\infty}-D(s)h(t+s)\dd s, 
\end{align*}
where the first equality follows from substituting $-s$ for $s$ and the second equality applies that $D$ is odd. Hence, we have
\begin{align*}
    G_v(t) - G_c(t) = \int_{0}^{\infty}D(s)(h(t-s)-h(t+s))\dd s.
\end{align*}

Suppose $t>0$. For any $s > 0$, symmetry and unimodality of $h$ imply $h(t-s)-h(t+s)\ge 0$ because $|t-s|<t+s$. Moreover, the inequality is strict whenever $t-s$ lies in the interior of the support of $h$. Without loss of generality, we assume choices are noisier (\cref{as_ordered}). Then, for all $s>0$, we have $D(s)\ge 0$ with strict inequality whenever $F_v(s) < 1$. Thus the integrand is non-negative. We now show that $D(s)(h(t-s)-h(t+s)) > 0$ on some interval of $s$. To this end, let $\bar \eta=\esssup(\eta)$ and let $\bar{\varepsilon_v}=\esssup(\varepsilon_v)$. Note that if $s \in (0,\bar{\varepsilon_v})$, then $D(s)>0$ and if $|t-s|<\bar\eta$, or equivalently, $s \in (t-\bar\eta,t+\bar\eta)$, then $h(t-s)-h(t+s)>0$. Since $G_v(t)=\mP_v(y)\in(0,1)$ by assumption and since $\eta$ and $\varepsilon_v$ are independent, $t<\bar \varepsilon_v+\bar \eta$. It follows that $(t-\bar\eta,t+\bar\eta) \cap (0,\bar{\varepsilon_v})$ is a non-empty interval on which the integrand is positive. An analogous argument applies when $t < 0$, where $h(t-s)-h(t+s) \le 0$. Thus, integrating over $s$, we have $G_v(t)-G_c(t)\neq0$, or equivalently, $\mP_c(y)\neq \mP_v(y)$, whenever $y\neq \mu$ and $\mP_c(y),\mP_v(y)\in (0,1)$.

\end{proof}

\subsection{Proof of Proposition \ref{prop:cond_cali}}\label{ap:cond_cali}

\begin{proof}[Proof of \cref{prop:cond_cali}]
By independence of the choice and valuation noises, whenever the
conditional probability is defined and since $\eta$ and $\varepsilon_v$ have densities by \cref{as_unimodal} and \cref{as_f_pos}, we have
\begin{align}\label{eq_cond}
    \mP_c[y\mid y_{vi}=y] 
    &= \frac{\int F_c(y-\mu-\eta)f_v(y-\mu-\eta)h(\eta)\,\dd \eta}{\int f_v(y-\mu-\eta)h(\eta)\,\dd \eta}\\
    &=\frac{1}{2}+\frac{\int O_c(y-\mu-\eta)f_v(y-\mu-\eta)h(\eta)\,\dd \eta}{\int f_v(y-\mu-\eta)h(\eta)\,\dd \eta},\end{align}
where $O_c(x)=F_c(x)-\frac{1}{2}$, which is odd by the symmetric noises assumption. 

Thus at $y=\mu$, we have
\begin{align*}
    \mP_c[\mu\mid y_{vi}=\mu] 
    &=\frac{1}{2}+\frac{\int O_c(-\eta)f_v(-\eta)h(\eta)\,\dd \eta}{\int f_v(-\eta)h(\eta)\,\dd \eta}=\frac{1}{2},
\end{align*}
since the numerator of the second term is zero as $O_c$ is odd, $f_v$ is even by symmetric noises, and $h$ is even by symmetric preferences.

We next show that $\mP_c[y\mid y_{vi}=y]=\frac12$ only at $\mu$. Define $t = y - \mu$ and $z = t-\eta$. With this change of variables, we have
\begin{align}
    \mP_c[y\mid y_{vi}=y] - \frac{1}{2}
    &=\frac{\int O_c(z)f_v(z)h(t-z)\,\dd z}{\int f_v(z)h(t-z)\,\dd z}.
    \label{eq:condcali}
\end{align}

By symmetry, $O_c$ is odd and $f_v$ is even. Pairing $z$ and $-z$,
the numerator becomes
\begin{equation}\label{eq:condcali_pair}
\int_0^\infty
O_c(z)f_v(z)
\bigl[h(t-z)-h(t+z)\bigr]\,\dd z .
\end{equation}

Suppose first that $t>0$. Symmetry and unimodality of $h$ imply
\[
h(t-z)\ge h(t+z)
\qquad\text{for every }z>0,
\]
since $|t-z|<t+z$. Moreover, the inequality is strict
whenever either term is positive. Positive noise density around zero together with symmetry of noise imply that $O_c(z)>0$ for  all $z>0$. Hence the integrand in \eqref{eq:condcali_pair} is nonnegative.

The denominator in \eqref{eq:condcali} is positive (which follows from $\mP_c[y\mid y_{vi}=y]$ being well-defined), so
$f_v(z)h(t-z)>0$ for a positive-measure set of $z$. If this occurs
on a positive-measure set of $z>0$, then the inequality
$h(t-z)>h(t+z)$ is strict there. If instead it occurs on a
positive-measure set of $z<0$, symmetry of $f_v$ implies that on this same set $f_v(-z)h(t-z)>0$. Reflecting this set, we conclude that for a positive-measure set of $z>0$,
$f_v(z)h(t+z)>0$, which again implies
$h(t-z)>h(t+z)$. Thus the integrand in
\eqref{eq:condcali_pair} is strictly positive on a set of positive
measure. Therefore
\[
\mP_c[y\mid y_{vi}=y]>\frac12.
\]

A symmetric argument applies in the case that $t<0$, demonstrating that 
\[
\mP_c[y\mid y_{vi}=y]<\frac12.
\]
\end{proof}

\bigskip

\section{General Version of Model}\label{ap:general_model}

\setcounter{figure}{0}
\setcounter{table}{0}
\setcounter{equation}{0}

In this appendix, we develop a general version of our model that accommodates lottery-lottery comparisons, the preference-reversal paradigm, and measurement noise added to utility rather than the response variable.

Suppose we want a person to compare two lotteries, $X$ and $Y$. When these lotteries can be ranked in terms of risk, we let $Y$ denote the safer lottery, but they need not be rankable.\footnote{In fact, none of what follows requires the structure of lotteries. We focus on lotteries rather than abstract alternatives because they are the subject of our data.} Individual $i$ has a metric $\Omega_{ti}$ for comparing lotteries in task type $t \in \{c,v\}$---that is, if $\Omega_{ti}(W) > \Omega_{ti}(W')$, then $W \succ_{ti} W'$. This metric could be a certainty equivalent, a more general form of valuation (e.g.,  $\psi$ such that $(\psi,p) \sim_{ti} W$), an expected utility, or some other form of utility; we do assume that this metric respects first-order stochastic dominance (i.e., if $W$ FOSD $W'$ then $\Omega_{ti}(W) > \Omega_{ti}(W')$). Importantly, decision noise will enter at the level of this metric.

For the choice task between $X$ and $Y$, we assume that individual $i$ chooses $Y$ when $\Omega_{ci}(Y)+\varepsilon_{c}(Y) \ge \Omega_{ci}(X)+\varepsilon_{c}(X)$ for some shocks $\varepsilon_{c}(X)$ and $\varepsilon_{c}(Y)$. Defining $\varepsilon_c \equiv \varepsilon_{c}(X)-\varepsilon_{c}(Y)$, we have $\mP_{ci}(Y)=\mP(\varepsilon_c \le \Omega_{ci}(Y) - \Omega_{ci}(X))$.

To represent a generalized version of a valuation task, we introduce a parameterized lottery $Z(\psi)$, where $\psi \in \R$ and $\psi > \psi'$ implies $Z(\psi)$ FOSD $Z(\psi')$. One example is $Z(\psi) \equiv (\psi,p)$. Given a $Z(\psi)$, a valuation task for lottery $W$ asks person $i$ to state the value $\psi_{Wi}$ such that $Z(\psi_{Wi})$ is of equal value to them as lottery $W$. For such tasks, we assume the person states the value $\psi_{Wi}$ such that $\Omega_{vi}(Z(\psi_{Wi})) = \Omega_{vi}(W) + \varepsilon_v(W)$ (which is unique because $\Omega_i$ respects first-order stochastic dominance). We also define $\psi^*_{Wi}$ such that $\Omega_{vi}(Z(\psi^*_{Wi})) = \Omega_{vi}(W)$, which is the person's indifference value according to their underlying preferences (analogous to $y^*_i$ in Section \ref{sec:theory}).

Across our data sources, there are three relevant types of valuation tasks:
\begin{itemize}
 \item \textit{$X$ Valuation Task}: The parameterized lottery $Z(\psi)$ is such that $Z(\psi_Y) = Y$ for some $\psi_Y$, and we elicit $\psi_{Xi}$ such that $Z(\psi_{Xi})$ is of equal value as $X$. For this type, $\mP_{vi}(Y) = \mP(\psi_{Xi} \le \psi_Y) = \mP(\Omega_{vi}(X) + \varepsilon_v(X) \le \Omega_{vi}(Z(\psi_Y))) = \mP(\varepsilon_v(X) \le \Omega_{vi}(Y)-\Omega_{vi}(X))$.
    \item \textit{$Y$ Valuation Task}: The parameterized lottery $Z(\psi)$ is such that $Z(\psi_X) = X$ for some $\psi_X$, and we elicit $\psi_{Yi}$ such that $Z(\psi_{Yi})$ is of equal value as $Y$. For this type, $\mP_{vi}(Y) = \mP(\psi_{Yi} \ge \psi_X) = \mP(\Omega_{vi}(Y) + \varepsilon_v(Y) \ge \Omega_{vi}(Z(\psi_X))) = \mP(-\varepsilon_v(Y) \le \Omega_{vi}(Y)-\Omega_{vi}(X))$.
    \item \textit{$XY$ Valuation Task}: Given a parameterized lottery $Z(\psi)$, we elicit both the $\psi_{Xi}$ such that $Z(\psi_{Xi})$ is of equal value as $X$ and the $\psi_{Yi}$ such that $Z(\psi_{Yi})$ is of equal value as $Y$. For this type, $\mP_{vi}(Y) = \mP(\psi_{Yi} \ge \psi_{Xi}) = \mP(\Omega_{vi}(Y) + \varepsilon_v(Y) \ge \Omega_{vi}(X)+\varepsilon_v(X)) = \mP(\varepsilon_v(X)-\varepsilon_v(Y) \le \Omega_{vi}(Y)-\Omega_{vi}(X))$.
\end{itemize}
For all three types, $\mP_{vi}(Y)=\mP(\varepsilon_v \le \Omega_{vi}(Y) - \Omega_{vi}(X))$ for some appropriately defined $\varepsilon_v$. Hence, when we apply the null of stable preferences---that is, that $\Omega_{ci}(W) = \Omega_{vi}(W) = \Omega_i(W)$ for all $W$---the two key equations are:
\begin{equation}\label{eqn:general_ind_key_eqn}
    \mP_{ci}(Y) = \mP(\varepsilon_c \le \Omega_i(Y) - \Omega_i(X)) ~~\text{ and } ~~ \mP_{vi}(Y) = \mP(\varepsilon_v \le \Omega_i(Y) - \Omega_i(X)).
\end{equation}

Note that \eqref{eqn:general_ind_key_eqn} takes the same form as \eqref{eqn:indiv_key_eqns} except that $\Omega_i(Y) - \Omega_i(X)$ replaces $y-y_i^*$. Assumptions about decision noise have the same meaning as in the text, and the two main individual results go through as follows.\footnote{The proofs for all results for the general model are essentially the same as the proofs for the special case considered in the text, and thus are omitted.}

\medskip

\noindent \textbf{Remark C1.} If noise is identical for choices and valuations, then $\mathcal{A}_i = \{(a,b)|a=b\}$.

\medskip

\noindent \textbf{Lemma C1.} Under symmetric noises:    
 \begin{enumerate}
     \item $\mathcal{A}_i = S_y \cup S_X$; and
     \item If $\Omega_i(Y) > \Omega_i(X)$, then $(\mP_{ci}(Y),\mP_{vi}(Y))\in S_y$, and if $\Omega_i(Y) < \Omega_i(X)$, then $(\mP_{ci}(Y),\mP_{vi}(Y))\in S_X$; and 
     \item If $\Omega_i(Y) = \Omega_i(X)$, then $(\mP_{ci}(Y),\mP_{vi}(Y)) = \left( \frac12,\frac12 \right)$.
 \end{enumerate}  

\medskip

To model heterogeneity in preferences, we use the following structure:
\begin{itemize}
    \item Heterogeneity in preferences: $\Omega_i(Y) - \Omega_i(X) = \Delta - \eta$, where $\Delta \in \R$ is the population mean preference, and $\eta$ is a mean-zero random variable with CDF $H$.\footnote{We assume that $\eta$ enters negatively so that, as in the main text, a larger $\eta$ implies a stronger preference for the riskier lottery $X$.} 
    \item Heterogeneity in decision noise: The additive noise shocks $\varepsilon_c$ and $\varepsilon_v$ can be correlated with $\eta$.
\end{itemize}
This structure is identical to the structure in the text except that $\Delta$ replaces $y-\mu$. Assumptions about the distribution of preferences have the same meaning as in the text, and the main aggregate results go through as follows.

\medskip

\noindent \textbf{Proposition C1.} Under stable preferences and (conditionally) symmetric noises, $\mathcal{A} = \con(S_y \cup S_X)$.

\medskip

\noindent \textbf{Proposition C2.} Under stable preferences and symmetric noises, if preferences are independent of noise, then $\mathcal{A} = \con(C_y \cup C_X) \cup \con(V_y \cup V_X)$.

\medskip

\noindent \textbf{Proposition C3.} Under stable preferences and symmetric noises, if preferences are symmetric and independent of noise, then:
 \begin{enumerate}
     \item $\mathcal{A} = S_y \cup S_X$; and
     \item If $\Delta > 0$, $(\mP_c(Y),\mP_v(Y)) \in  S_y$, and if $\Delta < 0$,  $(\mP_c(Y),\mP_v(Y)) \in  S_X$; and 
     \item If $\Delta = 0$, then $(\mP_{ci}(Y),\mP_{vi}(Y)) = \left( \frac12,\frac12 \right)$.
 \end{enumerate}  

\medskip

\noindent \textbf{Proposition C4.} Under stable preferences and symmetric and ordered noises, if preferences are symmetric, unimodal, and independent of noise, then $\mP_c(y)=\mP_v(y)\in (0,1)$ if and only if $\Delta = 0$, in which case $\mP_c(y)=\mP_v(y)=\frac{1}{2}$.

\medskip

For an analogue to \cref{prop:cond_cali}, we need a way of expressing that a person's response(s) on a valuation task suggest that $X$ and $Y$ are of equal value to them. If we define $\psi_{Yi} \equiv \psi_Y$ for $X$ valuation tasks and $\psi_{Xi} \equiv \psi_X$ for $Y$ valuation tasks, then for all three types of valuation tasks, the needed condition is $\psi_{Yi} = \psi_{Xi}$.

\medskip

\noindent \textbf{Proposition C5.}
Under stable preferences and symmetric noises, if preferences are symmetric, unimodal, and independent of noise, and the noise is independent with positive densities, then $ \mP_c(Y \mid \psi_{Yi} = \psi_{Xi}) =\frac{1}{2}$ if and only if $\Delta = 0$. Moreover, if $\Delta > 0$ then $\mP_c(Y \mid \psi_{Yi} = \psi_{Xi}) > \frac{1}{2}$, and if $\Delta < 0$ then $\mP_c(Y \mid \psi_{Yi} = \psi_{Xi}) < \frac{1}{2}$.

\medskip

Finally, for our structural assessments of stability, we can do analogous assessments based on the following population distribution of underlying preferences:
\begin{equation*}
\begin{pmatrix}
\Omega_{ci}(Y)-\Omega_{ci}(X) \\
\Omega_{vi}(Y)-\Omega_{vi}(X) \\
\end{pmatrix}
\sim N \left (
\begin{pmatrix}
\Delta_c  \\
\Delta_v \\
\end{pmatrix} ,
 \begin{pmatrix}
\gamma_c^2   & \gamma_{cv} \\
\gamma_{cv} & \gamma_v^2 \\
\end{pmatrix} 
\right).
\end{equation*}

\subsection{Translation to Specific Cases}

Given the general version of the model, we can now discuss translations to specific cases, along with a few comments where relevant.

\bigskip

\noindent \textbf{Lottery vs. certain amount, noise at level of certainty equivalents:} This is the case from the main text where we assume $Y=(y,1)$ and use $X$ valuations. If the metric $\Omega_i$ is the certainty equivalent of a lottery, then clearly $\Omega_i(Y)=y$. If we further define $y_i^*$ to be the certainty equivalent for $X$, then $\Omega_i(Y)-\Omega_i(X) = y - y_i^*$, and then the general results above translate directly to the results in the main text.

\bigskip

\noindent \textbf{Lottery vs. lottery, noise at the level of valuations response variable:} Two types of lottery-lottery comparisons appear in our data sources.

First, we might compare lotteries $X$ and $Y$ using $X$ valuations. If the metric $\Omega_i$ is in terms of the valuation response variable $\psi_{Xi}$, then $\Omega_i(X)=\psi^*_{Xi}$ and $\Omega_i(Y) = \psi_Y$. If we then define $y=\psi_Y$ and $y_i^*=\psi^*_{Xi}$, then $\Omega_i(Y)-\Omega_i(X) = y - y_i^*$.

Second, we might compare lotteries $X$ and $Y$ using $Y$ valuations. If the metric $\Omega_i$ is in terms of the valuation response variable $\psi_{Yi}$, then $\Omega_i(X)=\psi_{X}$ and $\Omega_i(Y) = \psi^*_{Yi}$. If we then define $y = -\psi_X$ and $y_i^*=-\psi^*_{Yi}$, then $\Omega_i(Y)-\Omega_i(X) = y - y_i^*$.

In either case, the general results above translate directly to the results in the main text. We note that this formulation is consistent with how we analyze lottery-lottery comparisons in Section \ref{sec:empirical}.

\bigskip

\noindent \textbf{Preference-reversal paradigm, noise at level of certainty equivalents:} In the  preference-reversals paradigm, we compare lotteries $X$ and $Y$ using $XY$ valuations. Suppose the metric $\Omega_i$ is in terms of the certainty equivalent of a lottery, where we use ${\rm CE}_{i}(W)$ to denote $i$'s certainty equivalent of lottery $W$. Then $\Omega_i(Y)-\Omega_i(X) = {\rm CE}_i(Y) - {\rm CE}_i(X)$. Hence, the general results hold with ${\rm CE}_i(Y) - {\rm CE}_i(X)$ in place of $\Omega_i(Y)-\Omega_i(X)$.

\bigskip

\noindent \textbf{Lottery vs. lottery, noise at level of expected utility:} Suppose the metric $\Omega_i$ is in terms of an expected-utility functional, where we use ${\rm EU}_i(W)$ to denote $i$'s expected utility from lottery $W$. Because $\Omega_i(Y)-\Omega_i(X) = {\rm EU}_i(Y) - {\rm EU}_i(X)$, the general results hold with ${\rm EU}_i(Y) - {\rm EU}_i(X)$ in place of $\Omega_i(Y)-\Omega_i(X)$.\footnote{If individuals use some other form of utility functional (e.g., the cumulative prospect theory functional from \cite{TverskyKahneman1992}), things would be analogous.}

\bigskip

When decision noise enters at the level of expected utility, the assumption of symmetric preferences---that is, that the population distribution of ${\rm EU}_i(Y) - {\rm EU}_i(X)$ is symmetric---may be less appealing. In particular, a more natural assumption is symmetry on some underlying preference parameter, for instance, an assumption that $u(z;\eta)=z^{\Delta + \eta}$ with $\eta$ symmetric and mean-zero. However, such an assumption would not generate symmetry in ${\rm EU}_i(Y) - {\rm EU}_i(X)$, and indeed the following example demonstrates that \cref{prop:symm_prefs} does not hold.

\bigskip

\noindent \textbf{Example:} Suppose $X \equiv (100,0.3)$ and $Y=(12,1)$. Suppose $\Delta = 0.5$, in which case a person with the mean preference prefers $Y$ to $X$. Suppose $\eta$ is equally likely to be $-0.2$, $-0.1$, $0.1$, and $0.2$. Then:
\begin{itemize}
    \item For $\eta=-0.2$, ${\rm EU}(Y;\eta)-{\rm EU}(X;\eta) = 12^{0.3} - (0.3) 100^{0.3} = 0.91$
    \item For $\eta=-0.1$, ${\rm EU}(Y;\eta)-{\rm EU}(X;\eta) = 12^{0.4} - (0.3) 100^{0.4} = 0.81.$
    \item For $\eta=0.1$, ${\rm EU}(Y;\eta)-{\rm EU}(X;\eta) = 12^{0.6} - (0.3) 100^{0.6} = -0.31.$
    \item For $\eta=0.2$, ${\rm EU}(Y;\eta)-{\rm EU}(X;\eta) = 12^{0.7} - (0.3) 100^{0.7} = -1.84.$
\end{itemize}
Now suppose $F_c$ is 50-50 $\bar c$ and $-\bar c$, and $F_v$ is 50-50 $\bar v$ and $-\bar v$. If $\bar c \in (0.31, 0.81)$, then the choice probabilities for the four types are 1, 1, $\frac{1}{2}$, and 0, yielding an average of $\mP_c(y)=5/8$. If at the same time $\bar v \in (0.91,1.84)$, then the valuation probabilities for the four types are $\frac{1}{2}$, $\frac{1}{2}$, $\frac{1}{2}$, and 0, yielding an average of $\mP_v(y)=3/8$. Hence, $(\mP_c(y),\mP_v(y))=(5/8,3/8)$, which is outside $S_y \cup S_X$.

\bigskip

\section{Details of Dataset Construction (for Section \ref{subsec:data-construct})}

\setcounter{figure}{0}
\setcounter{table}{0}
\setcounter{equation}{0}

\subsection{Data from the Preference-Reversal Literature}\label{appsubsec:PR_data}

To construct a dataset based on preference-reversal literature, we draw from a recent meta-analysis by \citet{Lu-2026-Meta-Anal}. In Table 2 from the Online Supplementary Material, there are 462 experiments from 54 studies for which one could calculate the empirical quantities $(\widehat{\mP_{c}}(y), \widehat{\mP_{v}}(y))$.\footnote{Note that \citet{Lu-2026-Meta-Anal}'s Table 2 also includes a number of additional entries for studies that report only the proportions of the two types of preference reversals. Because it is not possible to construct $(\widehat{\mP_{c}}(y), \widehat{\mP_{v}}(y))$ from only these proportions, we do not consider these studies.}

We then hired a research assistant to reconstruct the data ourselves. Specifically, the research assistant collected the 54 original studies. There were three outcomes:

\begin{itemize}
    \item For 33 studies, the research assistant was able to confirm the numbers from \cite{Lu-2026-Meta-Anal}.
    \item For 9 studies, the research assistant found discrepancies between the numbers from \cite{Lu-2026-Meta-Anal} and the numbers in the original study; for our dataset, we use the corrected numbers.
    \item For 12 studies, the research assistant was unable to construct the needed numbers from the original study; for our dataset, we excluded these studies.\footnote{\cite{Lu-2026-Meta-Anal} comments that some numbers were obtained from personal correspondence with the original authors.}
\end{itemize}

In the end, our final dataset consists of 314 experiments from 42 studies, and that is the data that we use here. In line with the literature, for each experiment, we take the $P$-bet to be the safer option. We highlight three further details. 

First, a number of the original studies report only the combined numbers across multiple conditions---i.e., across multiple ``experiments''. For instance, in the seminal work of \citet{grether1979economic}, each experiment covers six different pairs of bets (summarized in their Table 2), but they report decisions across all six pairs (in their Tables 5, 6, 8, and 9). Rather than exclude these, we chose to consider the combined numbers to be one experiment.

Second, for a number of the original studies, there is some ambiguity in how to define an experiment. So as not to give ourselves extra flexibility, we use the decisions from \cite{Lu-2026-Meta-Anal}. As an example, Experiment 1 from \citet{grether1979economic} splits subjects into one group that does not have incentives and a second group that does have incentives. In principle, one could consider the two groups to be different experiments, or one could combine the behavior of the two groups into one experiment. \cite{Lu-2026-Meta-Anal} chose to do the former, and thus so do we.

Third, some studies permit participants to report indifference. For such studies, we include responses of indifference in the total count of observations when calculating $\widehat{\mP_{c}}(y)$ or $\widehat{\mP_{v}}(y)$. For instance, if in a choice task, 30 choose the safe option, 10 choose the risky option, and 10 express indifference, we would calculate $\widehat{\mP_{c}}(y)=0.60$.

\subsection{Data from the Choices-versus-MPLs Literature}\label{appsubsec:CV_data}

To construct data based on the choices-versus-MPLs literature, we pull the needed data ourselves from the three key studies in this literature \citep{brown2018separated,freeman2019eliciting, freeman2019choice}. Specifically:

\begin{itemize}
    \item \cite{brown2018separated}: In this study, we identified 1 experiment. Specifically, we use the results reported in their Table 3 for their Experiment 1, where we use O-14 as the choice treatment and L-RPS as the valuations treatment. For the choice experiment, $N=61$ and $\widehat{\mP_{c}}(y)=0.44$; for the valuation experiment, $N=60$ and $\widehat{\mP_{v}}(y)=0.48$. Their Experiment 2 does not include a choice treatment, and thus it does not yield an experiment for us.
    
    \item \cite{freeman2019eliciting}: In this study, we identified 2 experiments. The study includes some minor treatment differences in terms of number of questions asked and how payments are made. However, since the authors ignore these differences in presenting their headline results in their Figure 1, we focus on the two experiments represented in that figure. For the Q1 experiment, the choice treatment has $N=81$ and $\widehat{\mP_{c}}(y)=0.77$, while the valuation treatment has $N=355$ and $\widehat{\mP_{v}}(y)=0.55$. For the Q2 experiment, the choice treatment has $N=83$ and $\widehat{\mP_{c}}(y)=0.70$, while the valuation treatment has $N=354$ and $\widehat{\mP_{v}}(y)=0.62$.

    \item \cite{freeman2019choice}: In this study, we identified 6 experiments, one associated with each of the six different single-choice treatments (SC-Q95, SC-Q90, SC-Q85, SC-Q80, SC-Q75, and SC-Allais). For the valuations treatment, we chose to use the combined R-List treatments. We then pull the proportions from their Table 3---e.g., for Q85, $\widehat{\mP_{c}}(y)=0.77$ and $\widehat{\mP_{v}}(y)=0.59$ (because the table reports the proportion choosing risky, these are one minus the numbers reported in the table). The paper does not report precise sample sizes for each of the thirteen treatments, but each has an average of 120. Hence, roughly, each choice experiment has $N=120$, one valuation experiment (R-List Allais) has $N=120$, and the other five valuation experiments each have  $N=360$ (as each combines three treatments).
\end{itemize}

For all three studies, we follow the authors' labels for which lottery is riskier.

\subsection{Data from \citet{MNOSS-2024-distinguishing,MNOSS-2026-connecting}}\label{appsubsec:MNOSS_data}

\subsubsection{Overview of Studies}

The primary focus of \citet{MNOSS-2024-distinguishing} and \citet{MNOSS-2026-connecting} is to better understand how people react to variants of the common-ratio (CR) and common-consequence (CC) manipulations proposed by \cite{allais1953comportement}. \cite{MNOSS-2024-distinguishing} highlight an inference problem when using paired-choice tasks to study CR problems, and propose using paired-valuation tasks instead to overcome this problem. \cite{MNOSS-2026-connecting} highlight how there has been limited exploration of the parameter space for both CR and CC problems, and they conduct such exploration using both paired-valuation tasks and paired-choice tasks. While neither study focuses directly on comparing decisions for choices versus valuations, it turns out their data is ideally situated for doing so.

These studies ask subjects to make three types of comparisons for risky lotteries:
\begin{itemize}
    \item[] $AB$ Comparison: $ \text{Safer~} A \equiv (M, 1) ~vs. ~\text{Riskier~} B \equiv (H, p)$
    \item[] $AB'$ Comparison: $\text{Safer~} A \equiv (M, 1) ~vs. ~\text{Riskier~} B' \equiv (H, pr;  M, 1-r)$
    \item[] $CD$ Comparison: $\text{Safer~} C \equiv (M, r) ~vs. ~\text{Riskier~} D \equiv (H, pr)$
\end{itemize}

These studies use two types of valuation tasks that differ in terms of the scalar value that is elicited (both types are implemented using a multiple-price list). In an $m$-valuation, they fix $(H,p,r)$ and ask subjects to state a value $m_{vi}$ that makes two options of equal value. In an $h$-valuation, they fix $(M,p,r)$ and ask subjects to state a value $h_{vi}$ that makes two options of equal value. For instance, for a $CD$ comparison with $p=0.5$ and $r=0.4$, the two tasks might be:
\begin{itemize}
     \item[] $m$-Valuation: State an $m_{vi}$ such that $(m_{vi},0.4)$ is of equal value to $(30,0.2)$.
     \item[] $h$-Valuation: State an $h_{vi}$ such that $(15,0.4)$ is of equal value to $(h_{vi},0.2)$.
\end{itemize}
As we describe below, across the two studies, $h$-valuations are used for all three types of comparisons, while $m$-valuations are used only for $AB$ and $CD$ comparisons. It follows that $m$-valuations always elicit a value within the safer option, and thus are analogous to the $X$ valuation tasks defined in Appendix \ref{ap:general_model}, while $h$-valuations always elicit a value within the riskier option, and thus are analogous to the $Y$ valuation tasks defined in Appendix \ref{ap:general_model}.

These studies also present the same subjects with linked choice tasks. Specifically, if a subject faces an $m$-valuation for a particular $(H,p,r)$, then they also face one choice task that uses that same $(H,p,r)$ along with a specific  $m_c$ that is randomly selected from a set of values. Analogously, if a subject faces an $h$-valuation for a particular $(M,p,r)$, then they also face one choice task that uses that same $(M,p,r)$ along with a specific $h_c$ that is randomly selected from a set of values. The sets of values were chosen based on pilot data to (ideally) span both sides of the mean preference.\footnote{Analogous to the choices-versus-MPLs literature, a choice task corresponds to one row of the multiple-price list for the linked valuation task. But unlike the choices-versus-MPLs literature, the comparison can be done within-subjects.} 

The primary analysis in \cite{MNOSS-2024-distinguishing} studies paired-valuation tasks. Specifically, they study CR problems by comparing an $AB$ and a $CD$ valuation for the same $(p,r)$, and they do so using both $m$-valuations and $h$-valuations. The primary analysis in \cite{MNOSS-2026-connecting} also studies paired-valuation tasks. Specifically, they study CR problems by comparing an $AB$ and a $CD$ valuation for the same $(p,r)$, CC problems by comparing an $AB'$ and a $CD$ valuation for the same $(p,r)$, and also preferences for probabilistic mixtures by comparing an $AB$ and a $AB'$ valuation for the same $(p,r)$. They do so using only $h$-valuations. Each study use the data on choice tasks to support their conclusions.

\subsubsection{Construction of Data for Our Purposes}

Again, neither study focuses directly on comparing decisions for choices versus valuations, but it turns out their data is ideally situated for doing so.

\cite{MNOSS-2024-distinguishing} collects data on $AB$ and $CD$ decisions for five different values of $p \in \{0.1,0.2,0.5,0.8,0.9\}$ and three different values of $r \in \{0.2,0.4,0.6\}$ for both $m$-valuations and $h$-valuations. Since $r$ is irrelevant for $AB$ decisions, they collect a total of 10 $AB$ valuations (five $m$-valuations and five $h$-valuations). Since $r$ matters for $CD$ decisions, they collect a total of 30 $CD$ valuations (15 $m$-valuations and 15 $h$-valuations). For each of these 40 valuations, they collect data for linked choices for four different values of $M$ or $H$. Hence, this study yields data on 40 groupings with 4 experiments within each group, and thus a total of 160 experiments. For the $AB$ decisions, there is an average of 225 subjects per experiment and thus 900 subjects per grouping. For the $CD$ decisions, there is an average of 75 subjects per experiment and thus 300 subjects per grouping.

\cite{MNOSS-2026-connecting} collects data on $AB$, $AB'$, and $CD$ decisions for four different values of $p \in \{0.3,0.5,0.8,0.9\}$ and five different values of $r \in \{0.1,0.2,0.3,0.5,0.8\}$ using only $h$-valuations. Again, since $r$ is irrelevant for $AB$ decisions, they collect a total of 4 $AB$ valuations. Since $r$ matters for $AB'$ and $CD$ decisions, they collect a total of 20 $AB'$ valuations and 20 $CD$ valuations. For each of these 44 valuations, they collect data for linked choices for six different values of $H$. Hence, this study yields data on 44 groupings with 6 experiments within each group, and thus a total of 264 experiments. For the $AB$ decisions, there is an average of 701 subjects per experiment and thus 4204 subjects per grouping. For the $AB'$ and $CD$ decisions, there is an average of 140 subjects per experiment and thus 841 subjects per grouping.

Combining the data from both studies, there is a total of 424 experiments across 84 groupings. In total, there are 68,448 observations for both choices and valuations, and thus an average of 161 subjects per experiment 815 subjects per grouping.

\subsubsection{Harmonizing $m$-Tasks and $h$-Tasks}

In an $m$-valuation, a larger $m_{vi}$ reflects a stronger preference for the riskier option, whereas in an $h$-valuation, a larger $h_{vi}$ reflects a stronger preference for the safer option. Analogously, for a choice task linked to an $m$-valuation, a larger selected value for $M$ makes the safer option more attractive, whereas for a choice task linked to an $h$-valuation, a larger selected value for $H$ makes the riskier option more attractive. Given this asymmetry, we must code the data in a way that harmonizes $m$-tasks and $h$-tasks.

When calculating $\widehat{\mP_c}(y)$ and $\widehat{\mP_v}(y)$ for a particular experiment, this harmonization is straightforward. Throughout, we have defined these as the proportion that select the safer alternative, and it is clear that the safer alternative is $A$ in an $AB$ task, $A$ in an $AB'$ task, and $C$ in a $CD$ task. Hence, for each experiment, we merely define $\widehat{\mP_{c}}(y)$ to be the proportion of subjects choosing the safer alternative and $\widehat{\mP_{v}}(y)$ to be to the proportion of subjects having a valuation that would indicate a preference for the safer alternative.

When using quantitative information on valuations, as needed for the tests from Sections \ref{sec:leveraging_richer-data} and \ref{sec:nnn_analysis}, we must be a little more careful. In our theory, a larger $y_{vi}$ means the person has a stronger preference for the risky option, and a larger $y$ implies the safer option is more attractive. $m$-tasks match this logic, and thus we can use data from $m$-tasks at face value. In other words, we can merely set $y_{vi}=m_{vi}$ and $y=m_c$, because the equations from our theory apply directly when $y-\overline{y_{vi}} = m_c - \overline{m_{vi}}$ and $y-y_{vi} = m_c - m_{vi}$. And for the MLE, as in our theory, a person chooses the safer option when $m_c$ is larger than $y_{ci}$.

For $h$-tasks, however, the logic is inverted, and thus we must correct for this. To do so, we set $y-\overline{y_{vi}} = \overline{h_{vi}} - h_c$ and $y-y_{vi} = h_{vi} - h_c$, and then apply the equations from our theory. And for the MLE, a person chooses the safer option when $h_c$ is smaller than $y_{ci}$.

\section{Details of Maximum Likelihood Estimation}\label{appsec:MLE_details}

\setcounter{figure}{0}
\setcounter{table}{0}
\setcounter{equation}{0}

We conduct our MLE approach on each of the 84 groupings defined in Section \ref{subsec:data-construct}. However, we must be careful to harmonize the estimation for $m$-valuations, where a larger valuation indicates less risk aversion, and $h$-valuations, where a larger valuation indicates more risk aversion (see Appendix \ref{appsubsec:MNOSS_data} for details on these).

For either type of valuation task, we let $w$ denote the object that is elicited (so $w \in \{m,h\}$). We then assume, analogous to equation \ref{eqn:struct_distn}, that a person's choice and valuation are determined from $w_{ci}$ and $w_{vi}$, where
\begin{equation}\label{appeqn:struct_distn}
 \begin{pmatrix}
 w_{ci} \\
 w_{vi} \\
\end{pmatrix}
\sim N \left (
 \begin{pmatrix}
 \mu_c  \\
\mu_v \\
\end{pmatrix} ,
 \begin{pmatrix}
\sigma^2_c + \gamma_c^2   & \gamma_{cv} \\
\gamma_{cv} & \sigma^2_v + \gamma_v^2 \\
\end{pmatrix} 
\right).
\end{equation}

Within a grouping, each individual faces the valuation task plus one of the (four or six) choice tasks. Hence, individual $i$'s data comes in the form $(w_c,a_{ci},w^L_{vi},w^H_{vi})$, where
\begin{itemize}
    \item $w_c$ is the value randomly assigned to $i$ for their choice task.
    \item $a_{ci} \in \{0,1\}$ is $i$'s chosen option in that choice task, where $a_{ci}=1$ denotes choosing the safer option and $a_{ci}=0$ denotes choosing the riskier option.
    \item $w^L_{vi}$ and $w^H_{vi}$ reflect the values of $i$'s switching rows in the valuation task---that is, in the price list, they choose the fixed option for rows with $w \le w^L_{vi}$, and they choose the varying option for rows with $w \ge w^H_{vi}$. Note that if the person always chooses the varying option, we code $w^L_{vi} = -\infty$; and if the person always chooses the fixed option, we code $w^H_{vi} = \infty$.
\end{itemize}

Given this structure, it is straightforward to derive an individual's likelihood function, although we must do it separately for $m$ versus $h$ valuations. For $m$ valuations, we use:\footnote{The function $N$ represents the CDF for a bivariate normal distribution.}
\begin{equation*}
L_i(m_c,0,m^L_{vi},m^H_{vi}) = 
N \left(
 \begin{pmatrix}
m_c , \infty \\
m^L_{vi}, m^H_{vi} \\
\end{pmatrix}
 ; \left (
 \begin{pmatrix}
 \mu_c  \\
\mu_v \\
\end{pmatrix} ,
 \begin{pmatrix}
\sigma^2_c + \gamma_c^2   & \gamma_{cv} \\
\gamma_{cv}& \sigma^2_v + \gamma_v^2 \\
\end{pmatrix} 
\right)
\right), \text{ and}
\end{equation*}
\begin{equation*}
L_i(m_c,1,m^L_{vi},m^H_{vi}) = 
N \left(
 \begin{pmatrix}
 -\infty, m_c  \\
m^L_{vi}, m^H_{vi} \\
\end{pmatrix}
 ; \left (
 \begin{pmatrix}
 \mu_c  \\
\mu_v \\
\end{pmatrix} ,
 \begin{pmatrix}
\sigma^2_c + \gamma_c^2   & \gamma_{cv} \\
\gamma_{cv}& \sigma^2_v + \gamma_v^2 \\
\end{pmatrix} 
\right)
\right).
\end{equation*}
For $h$ valuations, we instead use:
\begin{equation*}
L_i(h_c,0,h^L_{vi},h^H_{vi}) = 
N \left(
 \begin{pmatrix}
-\infty, h_c \\
h^L_{vi}, h^H_{vi} \\
\end{pmatrix}
 ; \left (
 \begin{pmatrix}
 \mu_c  \\
\mu_v \\
\end{pmatrix} ,
 \begin{pmatrix}
\sigma^2_c + \gamma_c^2   & \gamma_{cv} \\
\gamma_{cv}& \sigma^2_v + \gamma_v^2 \\
\end{pmatrix} 
\right)
\right), \text{ and}
\end{equation*}
\begin{equation*}
L_i(h_c,1,h^L_{vi},h^H_{vi}) = 
N \left(
 \begin{pmatrix}
h_c, \infty  \\
h^L_{vi}, h^H_{vi} \\
\end{pmatrix}
 ; \left (
 \begin{pmatrix}
 \mu_c  \\
\mu_v \\
\end{pmatrix} ,
 \begin{pmatrix}
\sigma^2_c + \gamma_c^2   & \gamma_{cv} \\
\gamma_{cv}& \sigma^2_v + \gamma_v^2 \\
\end{pmatrix} 
\right)
\right).
\end{equation*}

Combining all individuals in the grouping, the aggregate log likelihood is:
$$
L = \sum_i log L_i(w_c,a_{ci},w^L_{vi},w^H_{vi}).
$$
Maximizing the aggregate log likelihood with respect to the parameters $\mu_c$, $\mu_v$, $(\sigma^2_c + \gamma^2_c)$, $(\sigma^2_v + \gamma^2_v)$, and $\gamma_{cv}$ is a straightforward exercise yielding estimates of the corresponding parameters using conventional techniques. In our implementation, for each of the 84 groupings, we use the conjoint gradient algorithm with a maximum of 1000 iterations. If convergence is not achieved within 1000 iterations, we take the final values as our estimates.\footnote{Of our 84 MLE implementations, only 8 fail to achieve convergence within 1000 iterations.}

Note that the raw estimates in Table \ref{tab:probitsummary} and Appendix Table \ref{tab:app_mlsummary} do not distinguish between $m$-valuations and $h$-valuations, and more generally are not really comparable across groups. In interpreting these estimates, we are more interested in their difference, as that reflects which elicitation technique is more likely to induce risk aversion. For that, however, we must account for the fact that larger values have opposite meanings for $m$-valuations versus $h$-valuations. Given raw estimates $\hat \mu_v$ and $\hat \mu_c$, we define
$$
\widehat{\mu_v -\mu_c} = 
\begin{cases}
    \hat \mu_v - \hat \mu_c & \text{if it is an $m$-valuation} \\
    \hat \mu_c - \hat \mu_v & \text{if it is an $h$-valuation}
\end{cases} 
$$
Note that $\widehat{\mu_v -\mu_c}$ is defined such that a positive value means that choices induce more risk aversion than valuations for either type of valuation task. When describing the implications of the MLE estimates in the lower panel of Table \ref{tab:probitsummary} and Appendix Table \ref{tab:app_mlsummary} and in panel B of Figure \ref{fig:probitresults}, we report results in terms of $\widehat{\mu_v -\mu_c}$.

\section{Problem Parameters and Estimates of Instability}\label{app:parameters}

\setcounter{figure}{0}
\setcounter{table}{0}
\setcounter{equation}{0}

The 84 probit and maximum likelihood estimates presented in Table \ref{tab:probitsummary} and Table \ref{tab:app_mlsummary} reveal that different broad conditions lead to different quantitative estimates of mean preference instability, $\widehat{\mu_v -\mu_c}$. For $AB$ problems, the tendency is towards $\widehat{\mu_v -\mu_c}>0$ and thus choices yield greater risk aversion; for $AB'$ problems, the tendency is towards $\widehat{\mu_v -\mu_c}<0$ and thus valuations yield greater risk aversion; and for $CD$ problems, the tendency is towards $\widehat{\mu_v -\mu_c}\approx 0$ (particularly when expressed in expected value units).  

Importantly, there is also substantial heterogeneity within problem types, with each of $AB$, $AB'$, and $CD$ problems indicating some instances of $\widehat{\mu_v -\mu_c}>0$ and some instances of $\widehat{\mu_v -\mu_c}<0$.
Understanding the drivers of this heterogeneity may lead to further insights on when preference stability is more or less likely to hold.

To analyze the sources of this within-problem heterogeneity, in Table \ref{tab:app_instability} we regress probit estimates and MLE estimates of $\widehat{\mu_v -\mu_c}$ on the probability parameters $p$ and $r$ controlling for whether the corresponding grouping covers an $AB$, $AB'$, or $CD$ problem, and whether the grouping involves $m$-valuation tasks or $h$-valuation tasks. Columns (1) and (3) examine $p$ effects for all 84 groupings. Because $r$ is not defined for $AB$ problems, columns (2) and (4) examine $r$ effects for the 70 groupings which cover $AB'$ and $CD$ problems. We find that both $p$ and $r$ are highly and differentially correlated with estimates of instability. Choices yield greater risk aversion than valuations at low values of $p$, and less risk aversion than valuations at high values of $p$. In contrast, choices yield less risk aversion than valuations at low values of $r$, and greater risk aversion than valuations at high values of $r$.

Findings of preference instability may be quite sensitive not only to the type of lotteries used (e.g., certain vs. binary compared to certain vs. trinary compared to binary vs binary), but also to the probabilities implemented. Understanding the deeper reasons for these predictable effects on instability will be a critical path for future studies.

\begin{table}[t!]
  \begin{center}
    \caption{Problem Parameters and Estimates of Instability}
     \label{tab:app_instability}
     \scalebox{1}{
    \begin{tabular}{lcccclccc}
    \hline\hline
    \multicolumn{3}{c}{Probit Estimates} & & &  \multicolumn{3}{c}{MLE Estimates} \\
       Dependent Variable :& \multicolumn{2}{c}{$\widehat{\mu_v -\mu_c}$} &&&   Dependent Variable :&  \multicolumn{2}{c}{$\widehat{\mu_v -\mu_c}$} \\
& (1) & (2)  &&&& (3) & (4) \\
\cline{1-3} \cline{6-8} \\
\vspace*{.1in}
 $p$ & -21.72 & &&& $p$& -20.82 &  \\ 
  & (3.89) & &&& & (2.61) &  \\ 
  \addlinespace
 $r$ &  & 26.57 &&&$r$&  & 15.79 \\ 
  &  & (6.24) & &&& & (4.93) \\ 
\addlinespace
$\mathbf{1}(AB')$ & -16.71 & &&&  $\mathbf{1}(AB')$ & -20.77 &  \\ 
  & (3.64) &  &&& & (2.44) &  \\ 
  \addlinespace 
  $\mathbf{1}(CD)$ & -3.03 & 14.61 &&&  $\mathbf{1}(CD)$ & -4.48 & 17.12 \\ 
  & (3.04) & (3.16) & &&& (2.04) & (2.50) \\ 
  \addlinespace
 $\mathbf{1}(\text{\emph{m}-valuation})$ & -4.91 & -3.63 & &&  $\mathbf{1}(\text{\emph{m}-valuation})$ & -3.16 & -1.35 \\ 
  & (2.73) & (3.48) &&&& (1.83) & (2.75) \\ 
  \addlinespace 
  Constant & 18.91 & -21.48 & &&  Constant& 17.70 & -22.08 \\ 
  & (3.59) & (3.46) &&&& (2.41) & (2.74) \\ 

\cline{1-3} \cline{6-8} \\ 
\addlinespace 
    \# Observations & 84 & 70 & &&\# Observations & 84 & 70\\
R$^{2}$ & 0.45 & 0.38 &&& R$^{2}$ & 0.69 & 0.48 \\ 
     \hline \hline 
    \end{tabular}
    }
    \end{center}
\footnotesize{\emph{Notes:} Table presents OLS regressions of either probit or MLE estimated value of $\widehat{\mu_v-\mu_c}$ on problem parameters $p$ and $r$
controlling additionally for whether the grouping covers an $AB$, $AB'$, or $CD$ decision and whether the grouping's valuations elicit a middle prize (\emph{m}-valuation) or a high prize. Columns (1) and (3) focus on all 84 groupings and examines effect of $p$; columns (2) and (4) focus on the 70 groupings for $AB'$ and $CD$ decisions and examines effect of $r$. Standard errors in parentheses.
}
\end{table}

\clearpage

\section{Testing Sign Congruence \citep{miller2025testing}}\label{app:signcongruence}

\setcounter{figure}{0}
\setcounter{table}{0}
\setcounter{equation}{0}

In Section \ref{subsec:qual_test}, we deploy the test recommended by \citet{miller2025testing} for sign congruence between two estimators.\footnote{This test was originally proposed in the statistics literature by \cite{russek1993qualitative}.} In this section, we provide more detail.

\citet{miller2025testing} develop a methodology for testing the sign congruence between two theoretical quantities, $\mu_1$ and $\mu_2$, that is, for testing the null $H_0: \mu_1 \cdot \mu_2 \ge 0$. The test requires estimators $\hat{\mu}_1$ and $\hat{\mu}_2$ that satisfy the following assumption:

\medskip

\noindent \textbf{Assumption G1.} The estimators
 $\hat{\mu}_1$ and $\hat{\mu}_2$ are asymptotically normal
\begin{equation*}
\sqrt{n}
 \begin{pmatrix}
\hat{\mu}_1 - \mu_1 \\
\hat{\mu}_2 - \mu_2 \\
\end{pmatrix}
\to_d  N \left (
 \begin{pmatrix}
0  \\
0 \\
\end{pmatrix} ,
 \begin{pmatrix}
\sigma_1^2   &  \rho \sigma_1 \sigma_2 \\
\rho \sigma_1 \sigma_2 & \sigma_2^2 \\
\end{pmatrix} 
\right),
\end{equation*}
with convergent variance-covariance estimates, 
$$\frac{\hat{\sigma_1}}{\sigma_1} \to_p 1,~ 
\frac{\hat{\sigma_2}}{\sigma_2} \to_p 1,~
\frac{\hat{\rho}}{\rho} \to_p 1.
$$

\bigskip

Their recommended test rejects $H_0$ if 
$$
\hat{\mu}_1 \cdot \hat{\mu}_2 < 0 \qquad \text{ and } \qquad \min \left\{ \frac{|\hat{\mu}_1|}{\hat{\sigma}_1/\sqrt{n}} , \frac{|\hat{\mu}_2|}{\hat{\sigma}_2/\sqrt{n}} \right\} \ge \overline{z}_\alpha. 
$$
The critical value is $\overline{z}_\alpha = \Phi^{-1}(1-\alpha)$ when it is known that $\rho \ge 0$; otherwise, the critical value depends on $\rho$. For $\alpha = 0.05$, the critical threshold is $\overline{z}_\alpha =1.645$ when $\rho \ge 0$; it remains approximately  $\overline{z}_\alpha =1.645$ until $\rho$ is around $-0.8$; and then it  reaches $\overline{z}_\alpha =1.96$ when $\rho = -1$. \citet{russek1993qualitative} provide a tabulation of the relevant thresholds for $\alpha = 0.05$ depending on the correlation. Where necessary (and feasible) we use this tabulation to target the threshold for $\alpha = 0.05$.

Under the assumptions of Proposition \ref{prop:symm_prefs}, the attainable set satisfies $(\mP_c(y) - 0.5) \cdot (\mP_v(y) - 0.5) \ge 0$. We can thus directly apply the \citet{miller2025testing} test if we maintain Assumption G1 for the relevant empirical quantities---that is, for the sample proportions $\widehat{\mP_c}(y)$ and $\widehat{\mP_v}(y)$ and for their sample standard errors $\sqrt{\frac{\widehat{\mP_c}(y)(1-\widehat{\mP_c}(y))}{N_c}}$ and $\sqrt{\frac{\widehat{\mP_v}(y)(1-\widehat{\mP_v}(y))}{N_v}}$. Because the data from the prior literature do not contain information on the sample correlation, we apply the test assuming $\rho=0$ (and thus use $\overline{z}_\alpha =1.645$) for both of our datasets.\footnote{This simplification appears largely justified: In the  \cite{MNOSS-2024-distinguishing, MNOSS-2026-connecting} data, 414 of the 424 experiments exhibit positive correlation, and no correlation is smaller than $-0.19$.}

Under the assumptions of Proposition \ref{prop_weak_pop}, the attainable set is $\con(C_y \cup C_X) \cup \con(V_y \cup V_X)$. We can test this prediction using an adjustment to the basic test. \citet{miller2025testing} note that more general hypotheses akin to sign congruence can be tested by considering  a change of base. We apply a base change to $\mP_c(y)$ and $\mP_v(y)$ using the cone spanned by the vectors $\lbrace (4/3, 2/3), (2/3, 4/3) \rbrace$. Define:
\begin{equation*}
A \equiv 
 \begin{pmatrix}
4/3 & 2/3 \\
2/3& 4/3 \\
\end{pmatrix}^{-1} =  \begin{pmatrix}
1 & -\frac{1}{2} \\
-\frac{1}{2}& 1 \\
\end{pmatrix}
~~~ \text{ and } ~~~
\begin{pmatrix}
\nu_1 \\
\nu_2 
\end{pmatrix} \equiv  A\begin{pmatrix}
\mu_1\\
\mu_2
\end{pmatrix} = \begin{pmatrix}
\mu_1 - \frac{1}{2} \mu_2\\
- \frac{1}{2}\mu_1 + \mu_2 
\end{pmatrix}.
\end{equation*}
If $\mu_1 = \mP_c(y)-2/3$ and $\mu_2 = \mP_v(y) - 1/3$ 
then, $\nu_1 = \mP_c(y)- \frac{1}{2} - \frac{1}{2}\mP_v(y) $ and $\nu_2 = -\frac{1}{2} \mP_c(y) + \mP_v(y)$. Thus,  $\nu_1 = 0$ is equivalent to $\mP_v(y) = 2\mP_c(y) - 1 $, the lower boundary of $\con(C_y \cup C_X)$. Similarly, $\nu_2 = 0$ is equivalent to $\mP_v(y) = \frac{1}{2} \mP_c(y)$,  the lower boundary of 
$\con(V_y \cup V_X)$. With such definitions for $\mu_1$ and $\mu_2$, an observation lies below both lines if $\nu_1 >0$ and $\nu_2 < 0$. Defining $\mu_1' = \mP_c(y)-1/3$ and $\mu_2' = \mP_v(y) - 2/3$ with the same $A$ matrix delivers $\nu_1'=0$ and $\nu_2'=0$ being equivalent to the upper boundaries of  $\con(C_y \cup C_X)$ and $\con(V_y \cup V_X)$. An observation lies above both lines if $\nu_1' <0$ and $\nu_2' > 0$.

\citet{miller2025testing} show that asymptotic normality of $(\nu_1, \nu_2)$ with an arbitrary base change follows from asymptotic normality of $(\mu_1, \mu_2)$ and thus the same recommended test applies. However, in our case the  specific change of base induces negative correlation
between $\nu_1$ and $\nu_2$. Because, again, the data from the prior literature do not contain information on the sample correlation, we proceed under an assumption that $\mP_c(y)$ and $\mP_v(y)$ are uncorrelated. Under this assumption,
$$
corr(\nu_1,\nu_2) = corr(\nu'_1,\nu'_2) = \frac{-(\sigma_1^2 +\sigma_2^2)/2}{\sqrt{\sigma_1^2 +1/4 \sigma_2^2}\sqrt{1/4\sigma_1^2 + \sigma_2^2} } < 0.
$$
This negative correlation can be sufficiently large in magnitude to require adjustment of the critical value associated with $\alpha = 0.05$. When we apply this test to assess whether an observation comes from $\con(C_y \cup C_X) \cup \con(V_y \cup V_X)$, we calculate the sample analog for this correlation and use it to identify the corresponding threshold for $\overline{z}_\alpha$ from \citet{russek1993qualitative}
to target $\alpha = 0.05$.

\end{document}